\documentclass[8pt]{article}

\def\scri{\mathscr{I}}
\def\wphi{\chi}
\usepackage{tikz-cd}
\usepackage[utf8]{inputenc} 
\usepackage{amsmath}
\usepackage{mathrsfs}
\usepackage{amsfonts}
 \usepackage{cancel}
\usepackage{xcolor}
\usepackage{caption}
\usepackage{amsthm}
\usepackage{enumerate}
\usepackage{amssymb}

\usepackage{mathtools}
\mathtoolsset{showonlyrefs=true} 

\usepackage{psfrag}
\usepackage[colorinlistoftodos]{todonotes}

\DeclareFontEncoding{LS1}{}{}
\DeclareFontSubstitution{LS1}{stix}{m}{n}
\DeclareSymbolFont{symbols2}{LS1}{stixfrak} {m} {n}
\DeclareMathSymbol{\operp}{\mathbin}{symbols2}{"A8}

\usepackage{xcolor}
 	\definecolor{carblue}{rgb}{0.30, 0.50, 1}

\def\aa{\sigma}
\def\bb{\tau}
\def\cc{\sigma}

\def\q{q}

\newcommand{\nn}{\nonumber}
\newcommand{\dif}{\mathrm{d}}
\newcommand{\lr}[1]{\left( #1 \right)}
\newcommand{\lrbrace}[1]{\left\lbrace #1 \right\rbrace}
\newcommand{\lrprod}[1]{\left< #1 \right>}
\newcommand{\lrbrkt}[1]{\left[ #1 \right]}

\newcommand{\restr}[2]{\left.{#1}\right\rvert_{#2}}

\newcommand{\skwend}[1]{{\mathrm{SkewEnd}\lr{#1}}}
\newcommand{\rank}{\mathrm{rank}~}
\newcommand{\sign}{\mathrm{sign}}
\newcommand{\Tr}{\mathrm{Tr}}

\newcommand{\Lor}{O}
\newcommand{\X}{\mathrm{X}}

\def\la{\langle}
\def\ra{\rangle}

\newcommand{\conf}[1]{\mathrm{Conf}\lr{#1}}

\newcommand{\spn}[1]{\mathrm{span}\{ #1 \} }

\newtheorem{theorem}{Theorem}[section]
\newtheorem{proposition}{Proposition}[section]
\newtheorem{corollary}{Corollary}[theorem]
\newtheorem{lemma}[theorem]{Lemma}
\newtheorem{remark}{Remark}[section]
\newtheorem{definition}{Definition}[section]

\def\x{x}
\def\y{y}

 \title{Skew-symmetric endomorphisms in $\mathbb{M}^{1,n}$: A unified canonical form with applications to conformal geometry}

\author{
  Marc Mars and Carlos Pe\'on-Nieto\\
  Instituto de F\'{\i}sica  Fundamental y Matem\'aticas, Universidad de Salamanca \\
Plaza de la Merced s/n 37008, Salamanca, Spain}

\begin{document}

\maketitle

\begin{abstract} 
We show the existence of families of orthonormal, future directed bases which allow to cast every skew-symmetric endomorphism of $\mathbb{M}^{1,n}$ ($\skwend{\mathbb{M}^{1,n}}$) in a single canonical form depending on a minimal number of parameters. This canonical form is shared by every pair of elements in $\skwend{\mathbb{M}^{1,n}}$ differing by an orthochronous Lorentz transformation, i.e. it defines the orbits of the orthochronous Lorentz group under the adjoint action on its algebra.
Using this form, we obtain the quotient topology of $\skwend{\mathbb{M}^{1,n}}/O^+(1,n)$. From known relations between $\skwend{\mathbb{M}^{1,n}}$ and the conformal Killing vector fields (CKVFs) of the sphere $\mathbb{S}^n$, a canonical form for CKVFs follows immediately. This form is used to find adapted coordinates to an arbitrary CKVF that covers all cases at the same time. We do the calculation for even $n$ and obtain the case of odd $n$ as a consequence. Finally, we employ the adapted coordinates to obtain a wide class of TT-tensors for $n=3$, which provide Cauchy data at conformally flat null infinity $\mathscr{I}$. Specifically, this class of data is characterized for generating $\Lambda>0$-vacuum spacetimes with  two-symmetries, one of which axial, admitting a conformally flat $\mathscr{I}$. The class of data is infinite dimensional, depending on two arbitrary functions of one variable as well as a number of constants. Moreover, it contains  the data for the Kerr-de Sitter spacetime, which we explicitly identify within.
    \end{abstract}

\tableofcontents

 \section{Introduction}

 Having a Lie group $G$ acting on a space $X$ respresenting a set of physical quantities is always a desirable feature in a physical problem, as Lie groups represent symmetries (either global or gauge) and their presence is often translated into a simplification of the formal aspects of the problem. Roughly speaking, in this situation the ``relevant'' part for the physics effectively happens in the quotient space $X/G$. For the study of these quotient spaces, one may be interested in obtaining a unified form to give a representative for every orbit in  $X/G$, i.e.  a {\it canonical form} (also known as normal form). A particularly relevant case  is when $X$ is a Lie algebra $\mathfrak{g}$ and $G$ its Lie group acting by the adjoint action, in which case the orbits are also called conjugacy classes of $G$ (see e.g. \cite{burgoyne77}). In the first part of this paper, we study the conjugacy classes of the pseudo-orthogonal group $O(1,n)$ (or Lorentz group), for which we will obtain a canonical form. Our main interest on this, addressed in the second part of the paper, lies in its relation with the Cauchy problem of general relativity (GR) (cf. \cite{bruhat1952} and e.g. \cite{choquetGR}, \cite{friedrend00}) and more precisely, its formulation at null infinity $\mathscr{I}$ for the case of positive cosmological constant $\Lambda$ (cf. \cite{Fried86initvalue}, \cite{Friedrich2002}). In the remainder of this introduction we summarize our ideas and results, and we will also briefly review some results on conjugacy classes of Lie groups related to our case, as well as the Cauchy problem of GR with positive $\Lambda$.


 A typical example of a canonical form in the context described above, is the well-known Jordan form, which represents the conjugacy classes of $GL(n,\mathbb{K})$ (where $\mathbb{K}$ is usually $\mathbb{R}, \mathbb{C}$ or the quaternions $\mathbb{H}$). Besides this example, the problem of finding a canonical representative for the conjugacy classes of a Lie group has been adressed numerous times in the literature. The reader may find a list of canonical forms for algebras whose groups leave invariant a non-degenerate bilinear form in \cite{djokovic83} (this includes symmetric, skew-symmetric and simplectic algebras over $\mathbb{R}, \mathbb{C}$ and $\mathbb{H}$)
 as well as the study of the affine orthogonal group (or Poincaré group) in \cite{cushman06} or \cite{ida20}. Notice that these works deal, either directly or indirectly, with our case of interest $O(1,n)$, whose algebra $\mathfrak{o}(1,n)$ will be represented in this paper as skew-symmetric endomorphisms of Minkowski spacetime $\mathbb{M}^{1,n}$.  When giving a canonical form, it is usual to base it on criteria of irreducibility rather than uniformity (e.g. \cite{cushman06}, \cite{djokovic83}, \cite{ida20}).
 This is similar to what is done when the Darboux decomposition is applied to two-forms (i.e. elements of $\mathfrak{o}(1,n)$), for example in \cite{Kdslike} or for the low dimensional case $n=3$ (e.g.  \cite{hall2004symmetries}, \cite{syngeSr}).  As a consequence, all canonical forms found for the case of $\mathfrak{o}(1,n)$ require two different types of matrices to represent all orbits, one and only one fitting a given element. 
  Our first aim in this paper is to give a unique matrix form which represents each element $F\in \mathfrak{o}(1,n)$, depending on a minimal number of parameters that allows one to easily determine its orbit under the adjoint action of $O(1,n)$. This is obviously achieved by loosing explicit irreducibility in the form. However, this canonical form will be proven to be fruitful by giving several applications.  This same issue has also been adressed in \cite{marspeon20} for the case of low dimensions, i.e. $O(1,2)$ and $O(1,3)$, where in addition, several applications are worked out. The present work constitutes a natural generalization of the results in \cite{marspeon20} to arbitrary dimensions.
  
 The Lorentz group is well-known to be of particular interest in physics, as for example, it is the group of isotropies of the special theory of relativity and the Lorentz-Maxwell electrodynamics (e.g. \cite{LichnerowiczTheoRelGravEM}, \cite{Rainich1924}, \cite{syngeSr}). Its study in arbitrary dimensions have received renewed interest with theories of high energy physics such as conformal field theories \cite{IntroCFTschBook} or string theories \cite{ida20}. Related to the former and for our purposes here, a fact of special relevance is that the orthochronous component $O^+(1,n) \subset O(1,n)$ is homomorphic to the group of conformal transformations of the $n$-sphere, $\conf{\mathbb{S}^n}$. The conformal structure of $\mathscr{I}$ happens to be fundamental for the Cauchy problem at null infinity of GR for spacetimes with positive $\Lambda$ , as it is the gauge group for the set of initial data. Such a set consists of 
a manifold $\Sigma$ endowed with a (riemannian) conformal structure $[\gamma]$, representing the geometry of null infinity $\mathscr{I} := (\Sigma,[\gamma])$, together with the conformal class $[D]$ of a transverse (i.e. zero divergence), traceless, symmetric tensor $D$ (TT-tensor) of $\mathscr{I}$. If the spacetime generated by the data is to have a Killing vector field, the TT-tensor must satisfy a conformally covariant equation depending on a conformal Killing vector field (CKVF) of $\mathscr{I}$, the so-called Killing Initial Data (KID) equation \cite{KIDPaetz}.
 
 The class of data in which $[\gamma]$ contains a constant curvature metric (or alternatively the locally conformally flat case) includes the family of Kerr-de Sitter black holes and its study could be a possible route towards a characterization result for this family of spacetimes. Even in this particular (conformally flat) case, it is difficult to give a complete list of TT-tensors. An example can be found in \cite{Beig1997}, where the author gives a class of solutions with a direct and elegant method, but the solution is restricted in the sense that global topological conditions are imposed on $\mathscr{I}$. Namely, the solutions obtained by this method must be globally regular on $\mathbb{S}^3$ and hence cannot contain the family of Kerr-de Sitter, which is known (see e.g. \cite{Kdslike}) to have $\mathscr{I}$ with topology $\mathbb{S}^3$ minus two points, which correspond with the loci where the Killing horizons ``touch'' $\mathscr{I}$. The local problem for TT tensors is much more difficult to solve with generality, so our idea is to simplify it by imposing two KID equations to the data, so that the corresponding spacetimes have at least two symmetries. Using the homomorphism between $O^+(1,n+1)$ and $\conf{\mathbb{S}^n}$, we induce a canonical form for CKVFs
 from the canonical form obtained for $\mathfrak{o}(1,n+1)$.  Since this form covers all orbits of CKVFs under the adjoint action of $\conf{\mathbb{S}^n}$, our adapted coordinates fit every CKVF and in addition,  since the KID equation is conformally covariant, we can choose a conformal gauge where this CKVF is a Killing vector field, which makes the KID equation trivial.  Hence, a remarkable feature from our method is that by solving one simple equation, we are solving many cases at once. 
 This has already been done in the case of $\mathbb{S}^2$  in \cite{marspeon20} and here we extend it to the more interesting and difficult case of (open domains of) $\mathbb{S}^3$. 
 
 Specifically, we obtain the most general class of TT-tensors on a conformally flat $\scri$ such that the $\Lambda>0$-vacuum four-dimensional spacetime generated by these data admits two local isometries, one of them axial. It is worth highlighting that this is a broad class  (of infinite dimensions as it depends on functions) of TT-tensors and it contains the Kerr-de Sitter Cauchy data at $\mathscr{I}$.  This provides a potentially interesting "sandbox`` to try the consistence of possible definitions of (global) mass and angular momentum (see \cite{todszab19} for a review on the state of the art). Recall that symmetries are well-known to be related to conserved quantities, in particular, axial symmetry is related to conservation of angular momentum and time symmetry to conservation of energy. Moreover, for a spacetime to have constant mass, one may require no radiation escaping from or coming within the spacetime, a condition which, following the criterion of \cite{FASeno20}, is guaranteed by conformal the flatness of $\mathscr{I}$. Finally, the presence of the Kerr-de Sitter data within the set of TT-tensors contributes to its physical relevance and furnishes the possibility of looking for new characterization results for this family of spacetimes.

 As an additional sidenote concerning our results, notice that both the canonical form of CKVFs as well as the adapted coordinates are obtained in arbitrary dimensions, so similar applications may be worked out in arbitrary dimensions which, needless to say, is a considerable harder problem. On the possible extension to more dimensions of this type of TT-tensors, one should mention that the Cauchy problem at $\mathscr{I}$ for positive cosmological constant is known to be well-posed in arbitrary even dimensions \cite{Anderson2005}. However the KID equations are only known to be a necessary consequence of having symmetries, but sufficiency is an open problem in spacetime dimensions higher than four.

 The paper is organized as follows. In order to properly define the canonical form, in Section \ref{appclassification} we rederive a classification result for skew-symmetric endomorphisms (cf. Theorem \ref{theoremclasif}), employing only elementary linear algebra methods. The results of this section are known (see e.g. \cite{goodmanwallach}, \cite{hall2003lie}, \cite{knapp}), but the method is original and we believe more direct than other approaches in the literature. We also include the derivations in order to make the paper self-contained. 
 Section 2 leads to the definition of canonical form in Section \ref{seccanonform}.  Section \ref{secsimpleend} deals with a particular type of skew-symmetric endomorphisms (the so-called {\it simple}, i.e. of minimal matrix rank), which will be useful for future sections. In Section \ref{seclorclass} we work out some applications of our canonical form: identifying invariants which characterize the conjugacy classes of the orthochronous Lorentz group (cf. Theorem \ref{theoLorendo}) and obtaining the topological structure of this quotient space (cf. Section \ref{secstructure}). 

 In Section \ref{secconfvecs} we use the homomorphism between $O^+(1,n+1)$ and $\conf{\mathbb{S}^n}$ and apply the canonical form obtained for skew-symmetric endomorphisms to give a canonical form for CKVFs, together with a decomposed form (cf. Proposition \ref{proprecan}) which analogous to the one given for skew-symmetric endomorphisms. In Section \ref{secadapted}, we adapt coordinates to CKVFs in canonical form, first in the even dimensional case, from which the odd dimensional case is obtained as a consequence. Finally, in Section \ref{secTTtens} we employ the adapted coordinates to find
 the most general class of data at $\mathscr{I}$ corresponding to spacetime dimension four, such that $\mathscr{I}$ is conformally flat and the $(\Lambda>0)$-vacuum spacetime they generate admits at least two symmetries, one of which is axial. It is remarkable how easily are these equations solved with all the tools developed so far. With this solution at hand, we are able to identify the Kerr-de Sitter family within.

\section{Classification of Skew-symmetric endomorphisms}\label{appclassification}
 
  In this section we derive a classification result for skew-symmetric endomorphisms of Lorentzian vector spaces.
%
 Let $V$ be a $d$-dimensional vector space endowed with a pseudo-Riemannian metric $g$. If $g$ is of signature $(-,+,\cdots,+)$, then $\lr{V,g}$ is said to be Lorentzian. Scalar
product with $g$ is denoted by $\la\,,\,\ra$. An endomorphism
$F: V \longrightarrow V$ is skew-symmetric when it satisfies
\begin{align}\label{defskwend}
  \la x, F(y) \ra = - \la F(x), y \ra \quad \quad \forall x,y \in V.
\end{align}
We denote this set by $\skwend{V} \subset \mathrm{End}\lr{V}$.
We take, by definition, that eigenvectors of an endomorphism are always non-zero. {We use the standard notation for spacelike and timelike vectors
  as well as for spacelike, timelike and degenerate vector subspaces. 
 In our convention all vectors with vanishing norm are null} 
 (in particular, the zero vector is null).
  We denote $\ker F$
and $\mathrm{Im}~F$, respectively, to the kernel and image of
$F \in \mathrm{End}\lr{V}$. \begin{lemma}\label{lemmabasics}[Basic facts about skew-symmetric endomorphisms]
  Let $F$ be a skew-symmetric endomorphism in a pseudo-riemannian vector space
  $V$. Then
  \begin{itemize} 
  \item[a)] $\forall w \in V$, $F(w)$ is perpendicular to $w$, i.e. $\lrprod{F(w),w} = 0$.
        \item[b)] $\mbox{Im}~F \subset (\ker F)^{\perp}$ and
      $\ker F \subset (\mathrm{Im}~F)^{\perp}$. 
\item[c)] If $w \in \ker F \cap \mathrm{Im}~F$ then $w$ is null.
        \item[d)] If $w \in V$ is a non-null eigenvector of $F$, then its eigenvalue is zero.
        \item[e)] If  $w$ is an eigenvector of $F$ with zero eigenvalue, then
          all vectors in $\mathrm{Im}~F$ are orthogonal to $w$, i.e.
                    $\mathrm{Im}~F \subset w^{\perp}$.
        \item[f)] If $F$ restricts to a subspace $U \subset V$ (i.e. $F(U)\subset U$), then  it also restricts to $U^\perp$.            
        \end{itemize}
            \end{lemma}
           \begin{proof}
              a) is immediate from $\la w, F(w) \ra = - \la F(w),w \ra$.
              For $b)$, let $v\in
              \ker F$ and 
              $w$ be of the form $w = F(u)$ for some $u \in V$, then
                            \begin{align*}
                \la w,v \ra = \la F(u),v \ra = - \la u ,F(v) \ra =0
              \end{align*}
              the last equality following because $F(v)=0$.
                            $c)$ is a consequence
              of $b)$ because $w$ belongs both to $\ker F$ and to its orthogonal, so in particular it must be orthogonal to itself, hence null.
                            $d)$ is immediate from
              \begin{align*}
                0 = \la w,F(w) \ra = \lambda \la w,w\ra
              \end{align*}
              so if $w$ is non-null, its eigenvalue $\lambda$ must be zero. $e)$ is a corollary of $b)$ because by hypothesis $w \in \ker F$ so
              \begin{align*}
                \mathrm{Im}~F \subset (\ker F)^{\perp} \subset w^{\perp}
              \end{align*}
              the last inclusion being a consequence of the general fact $U_1 \subset U_2 \Rightarrow U_2^{\perp} \subset U_1^{\perp}$. Finally, $f)$ is true because for any $u$ in a $F$-invariant subspace $U$ and $w \in U^\perp$
              \begin{align*}
               0 = \lrprod{F(u),w} = -\lrprod{u,F(w)}.
              \end{align*}
           \end{proof}

Another well-known property of skew-symmetric endomorphisms that we will use is that $\dim \mathrm{Im}~F$ is always even. Equivalently, $\dim \ker F$ has the same parity than $\dim V$. To see this, consider the 2-form ${\boldsymbol F}$ assigned to every $F\in \skwend{V}$ by the standard relation
 \begin{equation}\label{Fflat}
  \boldsymbol{F}(e,e') = \lrprod{e, F(e')},\quad\quad \forall e,e' \in V.
 \end{equation}
The matrix representing $\boldsymbol F$ is skew in the usual sense, hence the dimension of $\mathrm{Im}~\boldsymbol F \subset {V}^\star$ (the dual of $V$) is the rank of that matrix, which is known to be even (see e.g. \cite{gantmacher1960theomat}) and clearly $\dim \mbox{Im}~\boldsymbol F = \dim \mbox{Im}~F$.

The strategy that we will follow to classify skew-symmetric endomorphisms of $V$ with $g$ Lorentzian is via $F$-invariant spacelike planes. Conditions for $F$-invariance of spacelike planes are stated in the following lemma:

\begin{lemma}\label{lemmasplkeigenplane}
 Let $F \in \skwend{V}$. Then $F$ has a $F$-invariant spacelike plane $\Pi_s$ if and only if 
  \begin{equation}\label{eigeneqs}
  F(u)  = \mu v,\quad \quad
  F(v)  = -\mu u,
\end{equation}
for $\Pi_s = \spn{u,v}$ with $u,v \in V$ spacelike, orthogonal, unit and $\mu \in \mathbb{R}$. Moreover, \eqref{eigeneqs} is satisfied for $ \mu \neq 0$ if and only if $\pm i \mu$ are eigenvalues of $F$ with (null) eigenvectors $u \pm i v$, for $u,v \in V$ spacelike, orthogonal with the same square norm.
\end{lemma}
\begin{proof}
 If \eqref{eigeneqs} is satisfied for $u,v \in V$ spacelike, orthogonal, unit, then $\Pi_s = \spn{u,v}$ is obviously $F$-invariant spacelike. On the other hand, if $\Pi_s$ is $F$-invariant, then it must hold that 
 \begin{equation*}
  F(u) = a_1 u + a_2 v,\quad \quad F(v) = b_1 u + b_2 v, \quad \quad a_1, a_2, b_1, b_2 \in \mathbb{R},
 \end{equation*}
 for a pair of orthogonal, unit, spacelike vectors $u,v$ spanning $\Pi_s$.
Using skew-symmetry and the orthogonality and unitarity of $u,v$, the constants are readily determined: $a_2 = b_2 = 0$ and $a_2 = -b_1 =:\mu$, which implies \eqref{eigeneqs}. This proves the first part of the lemma. 

For the second part, it is immediate that if \eqref{eigeneqs} holds with $\mu \neq 0$, then $\pm i \mu$ are eigenvalues of $F$ with respective eigenvectors $u \pm i v$. The orthogonality  of $u,v$ follows from $\lrprod{F(u),u} = 0 = \mu \lrprod{v,u}$ and the equality of norm from skew-symmetry $\lrprod{F(u),v} = -\lrprod{u,F(v)} \Rightarrow \mu \lrprod{v,v} = \mu \lrprod{u,u}$.
Assume now that $F$ has an eigenvalue $i \mu \neq 0$ with (necessarily null) eigenvector $w = u + i v$, for $u,v \in V$. Since $F$ is real, neither $u$ nor $v$ can be zero. From the nullity property $\lrprod{w,w} =0$, it follows that $\lrprod{u,u} - \lrprod{v,v}=0$ and $\lrprod{u,v} = 0$. Hence, $u,v$ are orthogonal with the same norm, so they are
either null and proportional, which can be discarded because it would imply that $u$ (and $v$) is a real eigenvector with complex eigenvalue; or otherwise $u,v$ are spacelike, thus the lemma follows.  
 
\end{proof}

There is an analogous result for $F$-invariant timelike planes:

 \begin{lemma}\label{lemmatmlkeigenplane}
 Let $F \in \skwend{V}$. Then $F$ has a $F$-invariant timelike plane $\Pi_t$ if and only if 
  \begin{equation}\label{eigeneqstime}
  F(e)  = \mu v,\quad \quad
  F(v)  = \mu e,
\end{equation}
for $\Pi_t = \spn{e,v}$ with $e,v \in V$ for $e$ timelike unit orthogonal to $v$ spacelike, unit and $\mu \in \mathbb{R}$. Moreover, \eqref{eigeneqs} is satisfied for $\mu \neq 0$ if and only if $\pm \mu$ are eigenvalues of $F$ with (null) eigenvectors $e \pm v$, for $e,v \in V$ orthogonal, timelike and spacelike respectively with opposite square norm.
\end{lemma}
\begin{proof}
For the first claim, repeat the first part of the proof of Lemma \ref{lemmasplkeigenplane} assuming $u = e$ timelike.

For the second claim, assume \eqref{eigeneqs} is satisfied with $\mu \neq 0$. Then it is immediate that $\lrprod{F(e),e} = 0 = \mu \lrprod{v,e}$, hence $e,v$ are orthogonal and by skew-symmetry $\lrprod{F(e),v} = - \lrprod{e,F(v)} \Rightarrow \mu \lrprod{v,v} = - \mu \lrprod{e,e}$, i.e. must have opposite square norm.
Conversely, let $\pm \mu \neq 0$ be a pair of eigenvalues with respective null eigenvectors $q_\pm$, that w.l.o.g can be chosen future directed. Then $e:= q_+ + q_-$ and $v:= q_+-q_-$ are orthogonal, with opposite square norm $\lrprod{e,e} = 2 \lrprod{q_+,q_-} = -\lrprod{v,v} <0$, and they satisfy \eqref{eigeneqstime}. 
\end{proof}

The $F$-invariant spacelike or timelike planes will be often be refered to as ``eigenplanes'' and $\mu$ will be denoted as the ``eigenvalues'' of $\Pi$. Notice that a simple change of order in the vectors switches the sign of the eigenvalue $\mu$. Thus, unless otherwise stated, we will consider the eigenvalues of eigenplanes (both spacelike and timelike) non-negative by default. 

The first question we address here is under which conditions such a plane exists (cf. Proposition \ref{clasiF}). But before doing so, we need to prove some results first.

\begin{lemma}\label{lemmaeigencond}
 Let $V$ be a Lorentzian vector space $F \in \skwend{V}$. Then  there exist two vectors 
 $\x,\y \in V$, with $\x \neq 0$, such that one of the three following 
exclusive possibilities hold
 \begin{enumerate}[(i)]
  \item $\x$ is a null  eigenvector of $F$.
  \item $\x$ is a non-null eigenvector (with zero eigenvalue).
  \item $\x,\y$ are orthogonal, spacelike and with the same norm,
and define an eigenplane of $F$ with non-zero eigenvalue, i.e. 
    \begin{align*}
  F(x)  = \mu y,\quad \quad
  F(y)  = -\mu x,\quad \mu \in \mathbb{R}\backslash \lrbrace{0}.
\end{align*}
 \end{enumerate}
If, instead, $V$ is riemannian, only cases (ii) and (iii) can arise.  
\end{lemma}
\begin{proof}
 From the Jordan block decomposition theorem we know that there is at least one, possibly complex, eigenvalue $s_1 + i s_2$ with eigenvector $\x + i \y$, that is,  $F(\x+ i \y) = (s_1+i s_2)(\x +i \y)$, or equivalently: 
  \begin{align}
F(\x) & = s_1 \x - s_2 \y\label{eigenCond1},  \\
F(\y) & = s_2 \x + s_1 \y\label{eigenCond2}.
\end{align}
This system is invariant under the interchange $(\x,\y) \rightarrow (- \y, \x)$, 
so without loss of generality we may assume $\x \neq 0$. The  respective
scalar products of (\ref{eigenCond1}) and (\ref{eigenCond2}) with $\x$, $\y$
yield
%
\begin{equation}\label{matrix}
\left . 
\begin{array}{cc}
s_1 \lrprod{\x, \x} - s_2 \lrprod{\x, \y} & =0 \\
s_1 \lrprod{\y, \y} + s_2 \lrprod{\x, \y} & =0 
\end{array}
\right \} 
\quad \quad \Longleftrightarrow \quad \quad
\left ( 
\begin{matrix}
\lrprod{\x, \x} & - \lrprod{\x, \y}  \\
\lrprod{ \y, \y} & \lrprod{\x, \y}  
\end{matrix}
\right )  
\left ( 
\begin{matrix}
s_1 \\
s_2 
\end{matrix}
\right )
= 
\left (
\begin{matrix}
0 \\
0 
\end{matrix}
\right ). 
\end{equation}
Observe that if $s_1 + i s_2  \neq 0$  the determinant of the matrix
must vanish. i.e. $\lrprod{\x,\y} \lr{ \lrprod{\x, \x} + \lrprod{\y, \y } } 
=0$. 
Hence, we can distinguish the following possibilities:

\begin{enumerate}[a)]
 \item[(a)]  $s_1 = s_2 =0$. Then $\x$ is an eigenvector of
$F$ with vanishing eigenvalue so we fall into cases $(i)$ or  $(ii)$.

\item[(b)] $s_1 + i s_2 \neq 0$. 
From $\lrprod{\x,\y} \lr{ \lrprod{\x, \x} + \lrprod{\y, \y } }=0$ we distinguish
two cases:
\addtolength{\itemindent}{1cm}
\item[(b.1)] $\la \x, \y \ra =0$. If $s_1 \neq 0$ then 
\eqref{matrix} forces $\x$ and $\y$ to be both null and, being also
orthogonal to each other,
there is $a \in \mathbb{R}$ such that $\y = a \x$ and we fall into case $(i)$. So, we can assume $s_1 =0$
(and then  $s_2 \neq 0$). Let $\mu := -s_2$, thus $(iii)$ follows from equations \eqref{eigenCond1}, \eqref{eigenCond2} and Lemma \ref{lemmasplkeigenplane}. 
\item[(b.2)] $\la \x, \y \ra \neq 0$.  Then
$\la \x, \x \ra = - \la \y ,\y \ra$  and the
matrix problem \eqref{matrix} reduces to
\begin{align}
s_1 \lrprod{\x,\x} -s_2 \lrprod{\x,\y} = 0. \label{s1s2}
\end{align}
In addition, \eqref{eigenCond1} and
\eqref{eigenCond2} imply
\begin{align*}
 \lrprod{F(\x),\y} 
& = s_1 \lrprod{\x,\y} -s_2 \lrprod{\y,\y} = 
s_1 \lrprod{\x,\y} +s_2 \lrprod{\x,\x} =  \lrprod{F(\y),\x}.
\end{align*}
But skew-symmetry requires
 $\lrprod{F(\x),\y} 
=- \lrprod{F(\y),\x}$, so $\lrprod{F(\y),\x} =0$  and we conclude
 \begin{align*}
s_1 \lrprod{\x,\y} +s_2 \lrprod{\x,\x}=0.
\end{align*}
Combining with (\ref{s1s2}) yields
\begin{eqnarray*}
\left ( 
\begin{matrix}
\lrprod{\x, \x} & - \lrprod{\x, \y}  \\
\lrprod{ \x, \y} & \lrprod{\x, \x}  
\end{matrix}
\right )  
\left ( 
\begin{matrix}
s_1 \\
s_2 
\end{matrix}
\right )
= 
\left (
\begin{matrix}
0 \\
0 
\end{matrix}
\right ). 
\end{eqnarray*}
The determinant of this matrix is non-zero which yields a contradiction with 
$s_1 + i s_2 \neq 0$. So this case is empty.

To conclude the proof, we must consider the case when the vector space V is riemannian. The proof is identical except from the fact that all cases involving null vectors are imposible from the start.
\end{enumerate}
\end{proof}

\begin{remark}\label{timelikeremark}
One may wonder why the lemma includes the possibility of having a spacelike eigenplane (case (iii)), but not a timelike eigenplane. 
The reason is that invariant timelike planes, which are indeed possible,  
fall into case (i) by Lemma \ref{lemmatmlkeigenplane}, because $e \pm v$ are null eigenvectors.
\end{remark}

 In the case of Riemmanian signature, Lemma \ref{lemmaeigencond} can be reduced to the following single statement: 
\begin{corollary}
\label{CorRiem}
Let $V$ be Riemannian of dimension $d$ and $F \in \skwend{V}$. If $d = 1$ then $F =0$ and if $d \geq 2$ then there exist two orthogonal and unit vectors $u,v$
satisfying
\begin{align}
\label{pairriem}
F(u)  = \mu v,  \quad \quad
F(v)  = -\mu u, \quad \quad \mu \in  \mathbb{R}
\end{align}
\end{corollary}
\begin{proof}
The case $d=1$ is trivial, so let us assume $d \geq 2$.
By the last statement of
Lemma \ref{lemmaeigencond} either there exists an eigenvector
$x$ with zero eigenvalue or the pair $\{u,v\}$ claimed in the corollary exists. 
In  the former case, we consider the vector subspace $x^{\perp}$. Its dimension
is at least one and $F$  restricts to this space 
so again either the pair $\{u,v\}$ exists
or there is $y \in x^{\perp}$ satisfying $F(y)=0$. But then $\{x,y\}$
are orthogonal and non-zero. Normalizing we find a pair $\{u,v\}$ that satisfies
\eqref{pairriem} with $\mu=0$,
\end{proof}



Lemma \ref{lemmaeigencond} lists a set of cases, one of which must always occur. However, we now show that, if the dimension is sufficiently high, case $(i)$ of that lemma implies one of the other two:

\begin{lemma}\label{auxprob} 
 Let $F \in \skwend{V}$, with $V$ Lorentzian of dimension at least four. If $F$ has a null eigenvector, then it also has either a spacelike eigenvector or a spacelike eigenplane. 
\end{lemma}
\begin{proof}
  Let $k  \in V$ be a null eigenvector of $F$.
The space $A:= k^{\perp}\subset V$ is a null hyperplane and $F$ restricts to $A$.
On this space we  define the standard equivalence relation
$\y_0 \sim \y_1$ iff $\y_0-\y_1 = a  k$, $a \in \mathbb{R}$. The quotient $A / \sim$
(which has dimension at least two) inherits
a positive definite metric $\overline{g}$ and $F$ also descends to the
quotient.
More precisely, if we denote the equivalence class of any $\y \in
A$  by $\overline{\y}$,  then for any $\overline{\y} \in
A / \sim$ and any $\y \in \overline{\y}$ the expression
$\overline{F} (\overline{\y}) = \overline{F(\y)}$ is well-defined (i.e.
independent of the choice of representative $\y$) and hence
defines an endomorphism $\overline{F}$
of $A /\sim$ which, moreover,
satisfies
\begin{align*}
  \la \overline{F} (
  \overline{\y_1}), \overline{\y_2} \ra {}_{\overline{g}}
  =
  - \la \overline{\y_1},  \overline{F} (\overline{\y_2}) \ra_{\overline{g}}.
\end{align*}
In other words $\overline{F}$ is a skew-symmetric endomorphism in the
riemannian vector space $A/\sim$.  By Corollary \ref{CorRiem} (here we use that the dimension of $A/\sim$ is at least two) there exists
a pair of orthogonal and $\overline{g}$-unit vectors
$\{ \overline{e_1}, \overline{e_2}\} $ satisfying
\begin{align}\label{eigena}
  \overline{F}(\overline{e_1})  =  a \, \overline{e_2}, 
  \quad \quad \quad \overline{F} (\overline{e_2})  =  - a \, \overline{e_1},\quad\quad a \in \mathbb{R}.
\end{align}
  Select representatives $e_1 \in
\overline{e_1}$ and  $e_2 \in \overline{e_2}$. In terms of $F$, the condition
(\ref{eigena}) and 
the fact that $k$ is eigenvector require the existence of constants
$\aa,a, \lambda_1$ and $\lambda_2$ such that
\begin{align*}
  F ( k) = \aa  k, \quad \quad
  F (e_1) =  a e_2 + \lambda_1  k, \quad \quad 
  F (e_2) = - a  e_1 + \lambda_2  k.
\end{align*}
Whenever $a^2 + \aa^2 \neq 0$ the vectors
\begin{align*}
  u :=  e_1 - \frac{1}{a^2 + \aa^2} \left ( a \lambda_2 + \aa \lambda_1
  \right )  k, \quad \quad
  v :=  e_2 + \frac{1}{a^2 + \aa^2} \left ( a \lambda_1- \aa \lambda_2 
  \right )  k
\end{align*}
satisfy $F (u) = a v$ and $F(v) = - a u$. Since $u$
and $v$ are spacelike, unit and orthogonal to each other
the claim of the  proposition follows (with $\mu =a$).
If $\aa = a =0$, then either $\lambda_1 = \lambda_2 =0$
and then $\{e_1,e_2\}$  are directly
the vectors
$\{u,v\}$ claimed in the proposition (with $\mu=0$), or 
at least one of the $\lambda$s 
(say $\lambda_2$) is not zero. Then $e := e_1 - \frac{\lambda_1}{\lambda_2}
 e_2$ is a spacelike eigenvector of $F$. 
\end{proof}

Now we have all the ingredients to show one of the main results of this section, that will eventually allow us to classify skew-symmetric endomorphisms of Lorentzian vector spaces.

\begin{proposition}\label{clasiF}
Let $V$ be a Lorentzian vector space of dimension at least five
and $F\in \skwend{V}$. 
Then, there exists a spacelike eigenplane. 

\end{proposition}

\begin{proof}

We examine each one of the three
possibilities described in Lemma \ref{lemmaeigencond}. Case $(iii)$
yields the result trivially, so we can assume that $F$ has an eigenvector $\x$.

If we are in case $(ii)$, the vector $\x$ is either 
spacelike or timelike. If it is timelike
we consider the riemannian space $\x^{\perp}$ where $F$ restricts. 
We may apply Corollary \ref{CorRiem} (note that $\x^{\perp}$ has dimension at least four)  and conclude that the vectors $\{u,v\}$ exist. So it remains to consider the case when $x$ is spacelike and  F admits no timelike eigenvectors. We
restrict to $\x^{\perp}$ which is Lorentzian and of dimension at least 
four. Applying again Lemma \ref{lemmaeigencond}, either there exists a spacelike eigenplane, or a second eigenvector $y \in x^\perp$, which can only be spacelike or null. If $y$ is spacelike, $\lrbrace{u:=x,~v :=y}$ span a spacelike eigenplane with $\mu = 0$. If $y$ is null, we may apply Lemma \ref{auxprob}  to $F\mid_{x^\perp}$ to conclude that either a spacelike eigenplane exists, or there is a spacelike eigenvector $e \in x^\perp$, so the pair $\{ u:= e, v:= \x\}$ satisfies \eqref{eigeneqs} with
$\mu =0$. This concludes the proof of case $(ii)$. 

In case $(i)$, i.e.
when there is a null eigenvector $x$  we can apply
Lemma \ref{auxprob} and conclude that either
$\{u,v\}$ exist, or there is a spacelike eigenvector $e\in V$, in which case we are into case $(ii)$, already solved. This completes the proof.

\end{proof}

Proposition \ref{clasiF} provides the basic tool to classify systematically skew-symmetric endomorphisms if the dimension $d$ is at least five. The idea is to start looking for a first spacelike eigenplane $\Pi$. Then, we restrict to $\Pi^\perp$, that is Lorentzian of dimension $d-2$. If $d-2\geq 5$, Proposition \ref{clasiF} applies again and we can keep on going until we reach a subspace of dimension three if $d$ odd or dimension four if $d$ even. Therefore, for a complete classification it only remains to solve the problem in three and four dimensions. 
This has already been done in \cite{marspeon20}, where a canonical form based on the classification of skew-symmetric endomorphisms is introduced. The results from \cite{marspeon20} that we shall need are summarized in Proposition \ref{propcanonF4} and Corollary \ref{propcanonF3} and their main consequences in the present context are discussed in Remarks \ref{remarcanonF4} and \ref{remarcanonF3} below, where we also relate the canonical form with the classification of skew-symmetric endomorphisms. For a proof and extended discussion, we refer the reader to \cite{marspeon20}.  In the remainder, when we explicitly write a matrix of entries ${F^\alpha}_\beta$, where $^\alpha$ is the row and $_\beta$ the column, we refer to a linear transformation expressed in a vector basis $\lrbrace{e_\alpha}_{\alpha = 0}^{d-1}$ acting on the vectors $v = v^\alpha e_\alpha \in V$ by 
\begin{equation}
 F(v) = {F^\alpha}_\beta v^\beta e_\alpha.
 \end{equation}

\begin{proposition}\label{propcanonF4}
For every non-zero $F \in \skwend{V}$, with $V$ Lorentzian four-dimensional, there exists an orthonormal basis $B:=\lrbrace{e_0,e_1,e_2,e_3}$, with $e_0$ timelike future directed, into which $F$ is
 \begin{equation}\label{canonFdim4}
  F =\left(
\begin{array}{cccc}
 0 & 0 & -1 + \frac{\aa }{4} & \frac{\bb }{4} \\
 0 & 0 & -1-\frac{\aa }{4} & -\frac{\bb }{4} \\
 -1+\frac{\aa }{4} & 1+\frac{\aa }{4} & 0 & 0 \\
 \frac{\bb }{4} & \frac{\bb }{4} & 0 & 0 \\
\end{array}
\right),\quad\quad \aa, \bb \in \mathbb{R},
 \end{equation}
where $\aa := - \frac{1}{2}\Tr{F^2}$ and $\bb^2 := - 4 \det F$, with $\bb \geq 0$. Moreover, if $\bb=0$ the vector $e_3$ can be taken to be any spacelike unit vector lying in the kernel of $F$.
\end{proposition}

\begin{corollary}\label{propcanonF3}
 For every non-zero $F \in \skwend{V}$, with $V$ Lorentzian three-dimensional, there exists an orthonormal basis $B:=\lrbrace{e_0,e_1,e_2}$, with $e_0$ timelike future directed, into which $F$ is
\begin{equation}\label{canonFdim3}
 F = \begin{pmatrix}
  0 &  0 & -1 + \frac{\cc}{4} \\
  0 & 0 & - 1 - \frac{\cc}{4} \\
  -1 + \frac{\cc}{4} & 1 + \frac{\cc}{4} & 0
  \end{pmatrix},\quad \quad \cc := -\frac{1}{2}\Tr\lr{F^2} \in \mathbb{R}.
\end{equation}
\end{corollary}
\begin{remark}
\label{remarcanonF4}
 A classification result follows because only two exclusive possibilities arise:
\begin{enumerate}
\item If either $\aa$ or $\bb$ do not vanish, $F$ has a timelike eigenplane and an orthogonal spacelike eigenplane with respective eigenvalues 
 \begin{equation}\label{eqmutmus}
  \mu_t := 
 \sqrt{(-\aa + \rho)/2}\quad \mbox{and}\quad \mu_s :=  \sqrt{(\aa + \rho)/2}\quad \mbox{for}\quad \rho := \sqrt{\aa^2 + \bb^2} \geq 0. 
 \end{equation}
 The inverse relation between $\mu_t, \mu_s$ and $\aa, \bb$ is $\aa = \mu_s^2 - \mu_t^2$ and $\bb = 2 \mu_t  \mu_s $.

 \item Otherwise,  $\aa = \bb = 0$ if and only if $\ker F$ is degenerate two-dimensional. Equivalently, $F$ has null eigenvector orthogonal to a spacelike eigenvector both with vanishing eigenvalue.

\end{enumerate}
One can easily check that when $\bb = 0$, the sign of $\aa$ determines the causal character of $\ker F$, namely $\aa <0$ if $\ker F$ is spacelike, $\aa = 0$ if $\ker F$ is degenerate and $\aa >0$ if $\ker F$ is timelike. Obviously, $\tau \neq 0$ implies $\ker F = \lrbrace{0}$. The characteristic polynomial of $F$ is directly calculated from \eqref{canonFdim4} 
 \begin{equation}\label{charpolF4}
  \mathcal{P}_F(x) = (x^2 -\mu_t^2)(x^2 + \mu_s^2).
 \end{equation}
 \end{remark}

 \begin{remark}\label{remarcanonF3}
 For a classification result in the three dimensional case, one can see by direct calculation that $q := (1 + \cc/4) e_0 + (1- \cc/4) e_1$ generates $\ker F$ and furthermore $\lrprod{q, q} = -\cc$. Hence, the sign of $\cc$ determines the causal character of $\ker F$, namely it is spacelike if $\cc < 0$, degenerate if $\cc = 0$ and timelike if $\cc > 0$. Moreover, when $\sigma \neq 0$, $F$ has an eigenplane with opposite causal character than $q$ and eigenvalue $\sqrt{|\cc|}$. The characteristic polynomial of $F$ reads
 \begin{equation}\label{charpolF3}
  \mathcal{P}_F(x) = x (x^2 + \cc).
 \end{equation}
\end{remark}

We have now all the necessary ingredients to give a complete classification of skew-symemtric endomorphisms of Lorentzian vector spaces. In what follows we identify Lorentzian (sub)spaces of $d$-dimension with the Minkowski space $\mathbb{M}^{1,d-1}$. Also, for any real number $x \in \mathbb{R}$, $[x]\in \mathbb{Z}$ denotes its integer part.

\begin{theorem}[Classification of skew-symmetric endomorphisms
in Lorentzian spaces]\label{theoremclasif}
 Let $F \in \skwend{V}$ with $V$ Lorentzian of dimension $d>2$. Then $V$ has a set of $[\frac{d-1}{2}]-1$ mutually orthogonal spacelike eigenplanes $\lrbrace{\Pi_i}$,  ${i = 1, \cdots, [\frac{d-1}{2}]-1}$, so that $V$ admits one of the following decompositions into direct sum of $F$-invariant subspaces:
 \begin{enumerate}[a)]
  \item If $d$ even $V = \mathbb{M}^{1,3}\oplus \Pi_{\frac{d-4}{2}} \oplus \cdots \oplus \Pi_1$ and either $F\mid_{\mathbb{M}^{1,3}} = 0$ or otherwise one of the following cases holds:
  \begin{enumerate}
   \item[a.1)]  $F\mid_{\mathbb{M}^{1,3}}$ has a spacelike eigenvector $e$ orthogonal to a null eigenvector with vanishing eigenvalue and then $\mathbb{M}^{1,3} = \mathbb{M}^{1,2} \oplus \spn{e}$.
   \item[a.2)]  $F\mid_{\mathbb{M}^{1,3}}$ has a spacelike eigenplane $\Pi_{\frac{d-2}{2}}$ (as well as a timelike eigenplane $\mathbb{M}^{1,1}$ orthogonal to $\Pi_{\frac{d-2}{2}}$) and then $\mathbb{M}^{1,3} = \mathbb{M}^{1,1} \oplus \Pi_{\frac{d-2}{2}}$.
  \end{enumerate}
   
  \item If $d$ odd $V = \mathbb{M}^{1,2}\oplus \Pi_{\frac{d-3}{2}} \oplus \cdots \oplus \Pi_1$ and either $F\mid_{\mathbb{M}^{1,2}} = 0$ or otherwise one of the following cases holds:
   \begin{enumerate}
   \item[b.1)]  $F\mid_{\mathbb{M}^{1,2}}$ has a spacelike eigenvector $e$ and then $\mathbb{M}^{1,2} = \mathbb{M}^{1,1} \oplus \spn{e}$.
   \item[b.2)]  $F\mid_{\mathbb{M}^{1,2}}$ timelike eigenvector $t$ and then $\mathbb{M}^{1,3} = \spn{t} \oplus \Pi_{\frac{d-1}{2}}$.
   \item[b.3)]  $F\mid_{\mathbb{M}^{1,2}}$ has a null eigenvector with vanishing eigenvalue.
  \end{enumerate}
 \end{enumerate}

\end{theorem}
\begin{proof}
 The proof is a simple combination of the previous results. First, if $d \geq 5$, we can apply Proposition \ref{clasiF} to obtain the first spacelike eigenplane $\Pi_1$. Then $\Pi_1^{\perp}$ is Lorentzian of dimension $d-2$. If $d-2 \geq 5$, we can apply again Proposition \ref{clasiF} to obtain a second eigenplane $\Pi_2$. Continuing with this process, depending on $d$, two things can happen:
 
 \begin{enumerate}[a)]
  \item If $d$ even, we get $\frac{d-4}{2} \lr{= [\frac{d-1}{2}]-1}$ spacelike eigenplanes, until we eventually reach a Lorentzian vector subspace of dimension four, $\mathbb{M}^{1,3}$, where Proposition \ref{clasiF} cannot be applied. In $\mathbb{M}^{1,3}$, either $F\mid_{\mathbb{M}^{1,3}} = 0$ or otherwise  cases $a.1)$ and $a.2)$ follow from Remark \ref{remarcanonF4}, cases $2$ and $1$ respectively.
  
   \item If $d$ odd, we get $\frac{d-3}{2}  \lr{ = [\frac{d-1}{2}]-1}$ spacelike eigenplanes, until we reach a Lorentzian vector subspace of dimension three, $\mathbb{M}^{1,2}$. In $\mathbb{M}^{1,2}$, either $F\mid_{\mathbb{M}^{1,2}}=0$ or by Remark \ref{remarcanonF3}  there exists a unique eigenvector $\sigma$ with vanishing eigenvalue. If $\sigma$ null, case $b.3)$ follows. If it is spacelike $e:=\sigma$, $F$ restricts to $e^\perp = \mathbb{M}^{1,1}\subset \mathbb{M}^{1,2}$ and $b.1)$ follows. If $\sigma$ timelike, the same argument applies with $t:= \sigma$ and $t^\perp  \subset \mathbb{M}^{1,2}$ defines the remaining spacelike plane $\Pi_{\frac{d-1}{2}}$.
 \end{enumerate}

\end{proof}

\section{Canonical form for skew-symmetric endomorphisms}\label{seccanonform}

 Our aim here is to extend the results in Proposition \ref{propcanonF4} and Corollary \ref{propcanonF3} to arbitrary dimensions. To do that, we will employ the classification Theorem \ref{theoremclasif} derived in Section \ref{appclassification}, from which it immediately follows a decomposition of any $F \in \skwend{V}$ into direct sum of skew-symmetric endomorphisms of the subspaces that $F$ restricts to, namely
\begin{align}
 F & = \restr{F}{\mathbb{M}^{1,3}} \bigoplus_{i = 1}^{[\frac{d-1}{2}]-1} \restr{F}{\Pi_i}\quad\quad \mbox{if $d$ even}, \label{decompFeven} \\ 
 F & = \restr{F}{\mathbb{M}^{1,2}} \bigoplus_{i = 1}^{[\frac{d-1}{2}]-1} \restr{F}{\Pi_i}\quad\quad\mbox{if $d$ odd}\label{decompFodd},
\end{align}
where $\Pi_i$ are spacelike eigenplanes.
In what follows, we will denote 
\begin{equation}
 p:=[(d-1)/2] -1.
\end{equation}
 Notice that the blocks $\restr{F}{\mathbb{M}^{1,3}}$ and $\restr{F}{\mathbb{M}^{1,2}}$ may also admit different subdecompositions depending on the case, but our purpose is to remain as general as possible, so we leave this part unaltered.
It will be convenient for the rest of the paper to give a name to the decompositions \eqref{decompFeven} and \eqref{decompFodd}: 
\begin{definition}
 Let $F\in \skwend{V}$ non-zero for $V$ Lorentzian $d$-dimensional. Then, a decomposition of the form \eqref{decompFeven} or \eqref{decompFodd} is called {\bf block form} of $F$. A basis that realizes a block form is called {\bf block form basis}.
\end{definition}

Writing $F$ in block form form allows us to work with $F$ as a sum of  skew-symmetric endomorphisms of riemmanian two-planes plus one  skew-symmetric endomorphism of a three or four dimensional Lorentzian vector space. For the latter we will employ the canonical forms in Proposition \ref{propcanonF4} and Corollary \ref{propcanonF3}, and for the former, it is immediate that in every (suitably oriented) orthonormal basis of $\Pi_i$ 
 \begin{equation}\label{canonFdim2}
  F\mid_{\Pi_i}=\begin{pmatrix}
   0 & -\mu_i  \\
   \mu_i & 0 
  \end{pmatrix},\quad\quad 0 \leq \mu_i \in \mathbb{R}. 
 \end{equation}
 
 Having defined a canonical form for four, three and two dimensional endomorphisms (i.e. matrices \eqref{canonFdim4},  \eqref{canonFdim3} and \eqref{canonFdim2} respectivley),  the idea is to extend this result to arbitrary dimensions finding a systematic way to construct a block form \eqref{decompFeven}, \eqref{decompFodd} such that each of the blocks are in canonical form. This is not immediate, firstly, because the block form does not require the blocks $\restr{F}{\mathbb{M}^{1,3}}$ or $\restr{F}{\mathbb{M}^{1,2}}$ to be non-zero and secondly, because, unlike in the four and three dimensional cases, the parameters $\aa, \bb$ of the four and three dimensional blocks cannot be invariantly defined as, for example,  traces of $F^2$ or determinant of $F$. The first of these concerns is easily solved by suitably choosing a block form: 
\begin{lemma}\label{lemmaprecan} 
 Let $F \in \skwend{V}$ be non-zero for $V$ Lorentzian of dimension $d$. Then there exists a block form \eqref{decompFeven} and  \eqref{decompFodd} such that $\restr{F}{\mathbb{M}^{1,3}}$ and $\restr{F}{\mathbb{M}^{1,2}}$ are non-zero and they either contain no spacelike eigenplanes or they contain one with largest eigenvalue (among all spacelike eigenplanes of $F$). In addition, the rest of spacelike eigenplanes $\Pi_i$ are sorted by decreasing value of $\mu_i^2$, i.e. $\mu_1^2 \geq \mu_2^2 \geq \cdots \geq \mu_{p}^2$.
\end{lemma}


\begin{proof}
 If $\ker F$ is degenerate, it must correspond with cases $a.1)$ ($d$ even) or $b.3)$ ($d$ odd) of Theorem \ref{theoremclasif}. Hence, in any block form the blocks $\restr{F}{\mathbb{M}^{1,3}}$ and $\restr{F}{\mathbb{M}^{1,2}}$ are non-zero and they do not contain any spacelike eigenplane, as claimed in the lemma.  So let us assume that $\ker F$ is non-degenerate or zero, which discards cases $a.1)$ and $b.3)$ of Theorem \ref{theoremclasif}. In all possible cases, any block form admits the following splitting in
\begin{equation}\label{auxdecom} 
 \restr{F}{\mathbb{M}^{1,3}} = \restr{F}{\Pi_t} \oplus \restr{F}{\Pi_s}, \quad \quad \restr{F}{\mathbb{M}^{1,2}} = \restr{F}{\spn{v}} \oplus \restr{F}{v^\perp},
\end{equation}
with $\Pi_s, \Pi_t$ spacelike and timelike eigenplanes with (possibly zero) respective eigenvalues $\mu_s$ and $\mu_t$, $v$ a timelike or spacelike eigenvector (in $\ker F$) and $v^\perp \subset \mathbb{M}^{1,2}$ an eigenplane with opposite causal character than $v$. If $v$ is spacelike, then either $\restr{F}{v^\perp}$ is non-zero, in which case $\restr{F}{\mathbb{M}^{1,2}} \neq 0$ and clearly contains no spacelike eigenplanes (which is one of the possibilities in the lemma), or $\restr{F}{v^\perp}=0$ and then $\restr{F}{\mathbb{M}^{1,2}} = 0$, so we can rearrange the decomposition \eqref{auxdecom} using some timelike vector $v' \in v^\perp$ instead of $v$, i.e. $\restr{F}{\mathbb{M}^{1,2}}  = \restr{F}{\spn{v'}}\oplus \restr{F}{v'^\perp} $ . Hence, in the case of $d$ odd, we may assume that $v$ is timelike and $v^\perp \subset \mathbb{M}^{1,2}$ is a spacelike eigenplane. Let $\Pi_{\mu}$ be a spacelike eigenplane of $F$ with largest eigenvalue $\mu$ among $\Pi_s$ ($d$ even) or $v^\perp$ ($d$ odd) and $\Pi_{1},\cdots, \Pi_{p}$. Then, switching $\restr{F}{\Pi_s}$ or $\restr{F}{\Pi_{v^\perp}}$ by $\restr{F}{\Pi_{\mu}}$  we  construct
\begin{equation}
 \hat F\mid_{\mathbb{M}^{1,3}} : = \restr{F}{\Pi_t} \oplus \restr{F}{\Pi_{\mu}}, 
 \quad \quad 
 \hat F\mid_{\mathbb{M}^{1,2}}  : = \restr{F}{\spn{v}} \oplus \restr{F}{\Pi_{\mu}}.
\end{equation}
 The resulting matrix is still in block form and has non-zero blocks $\restr{\hat F}{\mathbb{M}^{1,3}}$, $\restr{\hat F}{\mathbb{M}^{1,2}}$ containing a spacelike eigenplane with largest eigenvalue, which is the other possibility in the lemma.
 The last claim follows by simply rearranging the remaining spacelike eigenplanes $\Pi_i$ by decreasing order of $\mu_i^2$.
\end{proof}

With a skew-symmetric endomorphism $F$ in the block form given in Lemma \ref{lemmaprecan} we can take each of the blocks to its respective canonical form. Let us denote  $F_{\aa \bb} := \restr{F}{\mathbb{M}^{1,3}}$ (if $d$ even), $F_{\cc} := \restr{F}{\mathbb{M}^{1,2}}$ (if $d$ odd) and $F_{\mu_i} :=\restr{F}{\Pi_i}$ when written in the canonical forms \eqref{canonFdim4}, \eqref{canonFdim3} and \eqref{canonFdim2} respectively. Consequently
\begin{equation}\label{eqFgammamu}
  F = F_{\aa \bb} \bigoplus_{i = 1}^{p} F_{\mu_i}\quad \mbox{(d even)}, \quad \quad F = F_\cc \bigoplus_{i = 1}^{p} F_{\mu_i}\quad \mbox{(d odd)},
 \end{equation}
where, notice, each of the blocks is written in an orthonormal basis of the corresponding subspace, which moreover is future directed if the subspace is Lorentzian, i.e. $\mathbb{M}^{1,3}$ or $\mathbb{M}^{1,2}$ (c.f. Proposition \ref{propcanonF4} and Corollary \ref{propcanonF3}). Hence, the form given in \eqref{eqFgammamu} corresponds to a future directed, orthonormal basis of $\mathbb{M}^{1,d-1}$. 

Our aim now is to give an invariant definition of $\sigma, \tau, \mu_i$. A possible way to do this is through the eigenvalues of $F^2$. 
One may wonder why not to use directly the eigenvalues of $F$. One reason is that since we are interested in real Lorentzian vector spaces $V$ (although, for practical reasons, we may rely on the complexification $V_\mathbb{C}$ for some proofs), it is more consistent to give our canonical form in terms of real quantities, while the eigenvalues of $F$ may be complex. In addition, the canonical form
will require to sort them in some way, for which using real numbers is better
suited.

The characteristic polynomial of $F$ is known (e.g. \cite{Kdslike}) to possess the following parity:
\begin{equation}\label{paritycharpol}
 \mathcal{P}_F(x) = (-1)^d \mathcal{P}_F(-x).
\end{equation}
 Thus, a simple calculation relates the characteristic polynomials of $F$ and $F^2$
\begin{equation}\label{charpols}
 \begin{split}
  \mathcal{P}_{F^2}(x) & = \det (x Id_d - F^2) = \det \lr{ \sqrt{x} Id_d - F}  \det \lr{ \sqrt{x} Id_d + F} \\ 
  & =  (-1)^d \mathcal{P}_{F}(\sqrt{x}) \mathcal{P}_F(-\sqrt{x}) = \lr{ \mathcal{P}_{F}(\sqrt{x})}^2,
 \end{split}
\end{equation}
 $\sqrt{x}$ being any of the square roots of $x$ in $\mathbb{C}$ and $Id_d$ the $d \times d$ identity matrix.  We can extract some conclusions from \eqref{charpols}:

\begin{lemma} \label{remarkmultiplic}
Let $F \in \skwend{V}$ for $V$ Lorentzian of dimension $d$. Then the non-zero eigenvalues of $F^2$ have even multiplicity $m_a$ and the zero eigenvalue has multilplicity $m_0$ with the parity of $d$. In addition, $F$ possesses $p_a$ (resp. exaclty one) spacelike (resp. timelike) eigenplanes with eigenvalue $\mu \neq 0$ if and only if $F^2$ has a negative (resp. positive) non-zero eigenvalue $-\mu^2$ (resp. $\mu^2$) with multiplicity $m_a := 2p_a$ (resp. exactly two).
\end{lemma}
\begin{proof}
 It is an immediate consequence of equation \eqref{charpols} that non-zero eigenvalues of $F^2$ must have even multiplicity $m_a$. Moreover, since the sum of all multiplicites adds up to the dimension $d$, the multiplicity of the zero $m_0$ has the parity of $d$. 
 
 Combining Lemma \ref{lemmasplkeigenplane} and equation \eqref{charpols}, $F$ has a spacelike eigenplane $\Pi$ with non-zero eigenvalue $ \mu$ if and only if $F^2$ has a negative double\footnote{We adopt the convention that a root with multiplicity $m\geq 2$ is also double} eigenvalue $-\mu^{2}$.  If $d \leq 4$, there cannot be any other spacelike eigenplanes in $\Pi^\perp$, so applying the same argument to $\restr{F}{\Pi^\perp} \in \skwend{\Pi^\perp}$, the multiplicity $m_a$ of $-\mu^2$ must be $m_a = 2$. If $d >4$ and $m_a \geq 4$, then $-\mu^2$ is an eigenvalue of $(\restr{F}{\Pi^\perp})^2$ with multiplicity $m_a - 2$, thus $F$ has a second spacelike eigenplane with eigenvalue $\mu$ in $\Pi^\perp$. Repeating this argument, $F^2$ has a negative eigenvalue $-\mu^2$ with multilplicity $m_a$ if and only if $F$ has $p_a = m_a/2$ spacelike eigenplanes with eigenvalue $\mu$.
 
 Finally, by Lemma \ref{lemmatmlkeigenplane} and equation \eqref{charpols}, $F$ has a timelike eigenplane $\Pi$ with non-zero eigenvalue $ \mu$ if and only if $F^2$ has a positive double eigenvalue $\mu^{2}$. Obviously, the maximum number of timelike eigenplanes that $F$ can have is one. Thus, $\restr{F}{\Pi^\perp}$ cannot have timelike eigenplanes and hence $(\restr{F}{\Pi^\perp})^2$ has no additional positive eigenvalues. Consequently, the multiplicity of $\mu^2$ is exactly two.

\end{proof}

  
  Taking into account Lemma \ref{remarkmultiplic}, we will employ the eigenvalues of $-F^2$ rather than those of $F^2$, so we assign positive eigenvalues of $F^2$ with spacelike eigenplanes and negative eigenvalues to timelike eigenplanes. This amounts to employ the roots of the characteristic polynomial $\mathcal{P}_{F^2}(-x)$. 
  
 We now discuss how to invariantly define the parameters $\aa, \bb, \mu_i$ for $d$ even and $\cc, \mu_i$ for $d$ odd. The result of the argument is formalized below in Definition \ref{defgammamu}. Recall that the characteristic polynomial of a direct sum of two or more endomorphisms is the product of their individual characteristic polynomials, in particular, the characteristic polynomial of $-F^2$ equals to the product of the characteristic polynomials of $-F^2_{\aa\bb}$ or $-F_\cc^2$ times those of each $-F_{\mu_i}^2$ (c.f. equation \eqref{eqFgammamu}).  
 Let us define:
\begin{equation}\label{defQF2}
\mathcal{Q}_{F^2}(x) := \lr{\mathcal{P}_{F^2}(-x)}^{1/2} \quad \mbox{(d even)}, \quad\quad \mathcal{Q}_{F^2}(x) := \lr{\frac{\mathcal{P}_{F^2}(-x)}{x}}^{1/2} \quad \mbox{(d odd)},
\end{equation}
 Starting with $d$ even, from formula \eqref{eqFgammamu} it is immediate that $\mu_i^2$ are double roots of $\mathcal{P}_{F^2}(-x^2)$, which by Lemma \ref{lemmaprecan} satisfy $\mu_1^2 \geq \cdots \geq \mu_{p}^2 \geq 0$. On the other hand, let  $\mu_t := \sqrt{(-\aa + \rho)/2}$ and $i \mu_s :=  i \sqrt{(\aa + \rho)/2}$ with $\rho := \sqrt{\aa^2 + \bb^2} \geq 0$, that by Remark \ref{remarcanonF4}, are roots of $\mathcal P_{F_{\aa\bb}}(x)$, thus roots of $\mathcal P_{F}(x)$ . By equation \eqref{charpols}, $-\mu_t^2, \mu_s^2$ are double roots of $\mathcal{P}_{F^2}(-x)$. The set $\lrbrace{-\mu_t^2,\mu_s^2, \mu_1^2, \cdots, \mu_{p}^2}$ are in total $p + 2 = [(d-1)/2]+1 = d/2$ elements, each of which is a double root of $\mathcal{P}_{F^2}(-x)$. In other words,  $\lrbrace{-\mu_t^2,\mu_s^2, \mu_1^2, \cdots, \mu_{p}^2}$ is the set of all roots of the polynomial\footnote{$\mathcal{Q}_{F^2}(x)$ is a polynomial because all the roots of $\mathcal{P}_{F^2}(-x)$ are double.} $\mathcal{Q}_{F^2}(x)$. If $\ker F$ is degenerate, then $\ker F_{\aa \bb}$ is degenerate and by Remark \ref{remarcanonF4} it must happen $\mu_t = \mu_s = 0$. Hence $\mu_1^2 \geq \mu_2^2 \geq \cdots \mu_{p}^2 \geq \mu_s^2 = - \mu_t^2 =0$. Otherwise, also by Remark \ref{remarcanonF4}, $F_{\aa \bb}$ contains a spacelike eigenplane with eigenvalue $\mu_s$ (which by Lemma \ref{lemmaprecan} is the largest) as well as a timelike eigenplane with eigenvalue $\mu_t$. In this case $\mu_s^2 \geq \mu_1^2 \geq \cdots \mu_{p}^2 \geq 0 \geq -\mu_t^2$.

We next discuss $\sigma, \mu_i$ for $d$ odd. Again, from \eqref{eqFgammamu} we have that $\mu_i^2$ are double roots of $\mathcal{P}_{F^2}(-x^2)$, which by Lemma \ref{lemmaprecan} also satisfy $\mu_1^2 \geq \cdots \geq \mu_{p}^2 \geq 0$. By Remark \ref{remarcanonF3}, $\sqrt{\cc}$ is a root of $P_{F_\cc}(x)$, thus a root of $P_{F}(x)$, so by formula \eqref{charpols}, $\cc$ is a double root of $\mathcal{P}_{F^2}(-x)$. Also, $\mathcal{P}_{F^2}(-x)$ has at least one zero root and hence, $\mathcal{P}_{F^2}(-x)/x$ is a polynomial with $d-1$ roots (counting multiplicity). Then, the set $\lrbrace{\sigma, \mu_1^2, \cdots, \mu_{p}^2}$ are all double roots of $\mathcal{P}_{F^2}(-x)/x$, which are $p +1 = [(d-1)/2] = (d-1)/2$ elements. Therefore $Q_{F^2}$ as defined in \eqref{defQF2} is also a polynomial and
  $\lrbrace{\sigma, \mu_1^2, \cdots, \mu_{p}^2}$ is the set of all its roots. If $\ker F$ is timelike, then $\ker F_\cc$ is timelike, which happens if and only if $\cc > 0$ (c.f. Remark \ref{remarcanonF3}) and also $F_\cc$ has a spacelike eigenplane with eigenvalue $\sqrt{|\cc|}$, that by Lemma \ref{lemmaprecan} is the largest eigenvalue among spacelike eigenplanes. Thus $\sigma \geq \mu_1^2 \geq \cdots \geq \mu_{p}^2$. In the case $\ker F$ not timelike, the inequalities become $\mu_1^2 \geq \cdots \geq \mu_{p}^2 \geq 0 \geq \cc$.

 {Summarizing}, the paramaters $\aa, \bb, \mu_i$ correspond {to} the set of all roots of $\mathcal{Q}_{F^2}$ sorted in a certain order fully determined by the causal character of $\ker F$. This allows us to put forward the following definition:

\begin{definition}\label{defgammamu}
 Let $\mathrm{Roots}\lr{\mathcal{Q}_{F^2}}$ denote the set of roots of $\mathcal{Q}_{F^2}(x)$ repeated as many times as their multiplicity. Then
 \begin{enumerate}
  \item[a)] If $d$ odd, $\lrbrace{\cc; \mu_1^2, \cdots, \mu_{p}^2} := \mathrm{Roots}\lr{\mathcal{Q}_{F^2}}$ sorted by $\cc \geq \mu_1^2\geq \cdots \geq \mu_{p}^2$ if $\ker F$ is timelike and $\mu_1^2\geq \cdots \geq \mu_{p}^2\geq 0 \geq \cc$ otherwise. 
  
  \item[b)] If $d$ even, $\aa := \mu_s^2 - \mu_t^2,~\bb :=2 |\mu_t \mu_s|$ with $\lrbrace{-\mu_t^2, \mu_s^2; \mu_1^2, \cdots, \mu_{p}^2} := \mathrm{Roots}\lr{\mathcal{Q}_{F^2}}$ sorted by $\mu_1^2\geq \cdots \geq \mu_{p}^2\geq\mu_s^2 = -\mu_t^2 = 0$ if $\ker F$ is degenerate and $\mu_s^2 \geq \mu_1^2\geq \cdots \geq \mu_{p}^2\geq 0 \geq  -\mu_t^2$ otherwise.  
  
 \end{enumerate}

\end{definition}

In addition, we also summarize the results concerning the canonical form in the following Theorem:

\begin{theorem}\label{theocanonicalF}
 Let $F \in \skwend{V}$ non-zero, with $V$ Lorentzian of dimension $d \geq 3$ and $p := [(d-1)/2]-1$. Then there exists an orthonormal, future oriented basis such that $F$ is given \eqref{eqFgammamu} where $F_{\aa\bb} := F\mid_{\mathbb{M}^{1,3}}$, $F_{\cc} := F\mid_{\mathbb{M}^{1,2}}$, $F_{\mu_i}:=F\mid_{\Pi_i}$ are given by \eqref{canonFdim4}, \eqref{canonFdim3}, \eqref{canonFdim2} respectively and $\sigma, \tau, \mu_i$ are given in Definition \ref{defgammamu}. In particular, $F_{\aa\bb} $, $F_{\cc}$ are non-zero and they either do not contain a spacelike eigenplane or they contain one with maximal eigenvalue (among all spacelike eigenplanes of $F$) and the eigenvalues $\mu_i$ are sorted by $\mu_1^2 \geq \mu_2^2 \geq \cdots \mu_{p}^2$.
 
%
\end{theorem}

\begin{definition}
 For any $F\in \skwend{V}$, for $V$ Lorentzian $d$-dimensional, the form of $F$ given in Theorem \ref{theocanonicalF} is called {\bf canonical form} and the basis realizing it is called {\bf canonical basis}. 
\end{definition}

The first and obvious reason why the canonical form is useful is that it allows one to work with all elements $F \in \skwend{V}$ at once. The fact that we can give a canonical form for every element without splitting into cases is a great strenght, since we can perform a general analysis just in terms of the parameters that define the canonical form.
Moreover, as we will show in Section \ref{seclorclass}, this form is the same for all the elements in the orbit generated by the adjoint action of the orthochronous Lorentz group $\Lor^+(1,d-1)$. Thus, the canonical form is specially suited for problems with $\Lor^+(1,d-1)$ invariance (or covariance) which, as discussed in Section \ref{secconfvecs}, is directly related to certain conformally covariant problems in general relativity.

We finish this section with two corollaries that will be useful later. The first one is trivial from the canonical form \eqref{eqFgammamu}
 
\begin{corollary}\label{corolcharac}
 The characteristic polynomial of $F \in \skwend{V}$ is
 \begin{equation}\label{eqcharac}
  \mathcal{P}_F(x)  = (x^2 -\mu_t^2)(x^2 +\mu_s^2) \prod_{i=1}^{p}(x^2 + \mu_i^2)\quad \mbox{($d$ even)},\quad\quad
  \mathcal{P}_F(x)  = x (x^2 + \cc) \prod_{i=1}^{p}(x^2 + \mu_i^2)\quad \mbox{($d$ odd)} ,
\end{equation}
where $-2 \mu_t^2 := \aa - \sqrt{\aa^2 + \bb^2}$, $2 \mu_s^2 := \aa + \sqrt{\aa^2 + \bb^2}$ .
\end{corollary}

The second gives a formula for the rank of $F$. We base our proof in the canonical form \eqref{eqFgammamu} because it is straightforward. However, we remark that this corollary can also be regarded as a consequence of Theorem \ref{theoremclasif}.

\begin{corollary}\label{lemmaequivcases}
 Let $F \in \skwend{V}$, with $V$ Lorentzian of dimension $d$ and $m_0$ the multiplicity of the zero eigenvalue. Then, only of the following exclusive cases hold:
  \begin{itemize}
    \item[a)] $\ker F$ is non-degenerate or zero if and only if $\rank F = d - m_0$.
    
      \item[b)] $\ker F$ is degenerate if and only if $m_0 >2$ and $\rank F = d - m_0 + 2$.
       \end{itemize}
\end{corollary}
\begin{proof}
 Consider $F$ in canonical form \eqref{eqFgammamu} and let $k \in \mathbb N$ be the number of parameters $\mu_i$ that vanish. 
 For $d$ even we have $\dim \ker F = 2 k + \dim \ker F_{\aa\bb}$. 
  On the one hand, $\ker F$ degenerate implies $\ker F_{\aa \bb}$ degenerate, which by Remark \ref{remarcanonF4} happens if and only if $\aa=\bb = 0$ and in addition  $ \dim \ker F_{\aa \bb} = 2$. Therefore $\dim \ker F = 2 k + 2$ and by \eqref{eqcharac}, $m_0 = 2 k + 4$ $ (>2)$. Thus $\rank F = d - \dim \ker F = d - m_0 + 2$. On the other hand, $\ker F$ non-degenerate if at most one of $\aa$ or $\bb$ vanish. If $\bb \neq 0$ (so that $\mu_s \neq 0$ and $\mu_t \neq 0$), $ \dim \ker F_{\aa \bb}=0$ and  $m_0 = 2 k = \dim \ker F$. Consequently $\rank F = d-m_0$. If $\bb = 0$ (and $\aa \neq 0$, so that exactly one of $\mu_s, \mu_t$ vanish), by Remark \ref{remarcanonF4} $\dim \ker F_{\aa \bb} = 2$ and by \eqref{eqcharac} $m_0 = 2 k +2$. Hence $\dim \ker F = 2 k +2$ and  $\rank F =  d - m_0$.  
  
  For $d$ odd, we have $\dim \ker F = 2 k + \dim \ker F_\cc = 2k + 1$, because $\dim \ker F_\cc = 1$ (c.f. Remark \ref{remarcanonF3}). $\ker F$ is degenerate if and only if $\ker F_\cc$ is degenerate, which by Remark \ref{remarcanonF3} occurs if and only if $\cc = 0$. Hence, by equation \eqref{eqcharac}, $m_0 = 2 k + 3$ $(>2)$ and $\rank F = d - \dim \ker F = d - m_0 +2$.  For the $\ker F$ non-degenerate case, $\cc \neq 0$ and also by \eqref{eqcharac} $m_0 = 2 k + 1 = \dim \ker F$. Therefore $\rank  F= d - m_0$.
 
%

\end{proof}

 \section{Simple endomorphisms}\label{secsimpleend}

 By {\it simple skew-symmetric endomorphism}  we mean a $G \in \skwend{V}$ satisfying $\rank G = 2$.  As usual $e_\flat \equiv \lrprod{e,\cdot}$ is the one-form obtained by lowering index to a vector $e \in V$. Then, a simple skew-symmetric endomorphism can be always written as
\begin{equation}\label{simpleend}
 G = e \otimes  v_\flat -v \otimes  e_\flat
\end{equation}
for two linearly {independent} vectors $e,v \in V$ and its action on any vector $w \in V$ is 
\begin{equation*}
 G(w) = \lrprod{v,w} e - \lrprod{e,w} v.
\end{equation*}
Since the two-fom associated to a simple endomorphism is $\mathbf{G} = e_\flat \wedge v_\flat$, it follows from elementary algebra  that two simple skew-symmetric endomorphisms $G = e \otimes  v_\flat -v \otimes  e_\flat$ and $G'= e' \otimes  v'_\flat -v' \otimes  e'_\flat$ are proportional  if and only if $\spn{e,v} = \spn{e',v'}$. This freedom in the pair  $\{e,v\}$ defining $G$ can be used
to choose them orthogonal.
\begin{lemma}\label{simpleortho}
 Let $G\in \skwend{V}$ be simple. Then there exist two non-zero orthogonal vectors $e,v \in V$ such that $G = e \otimes  v_\flat -v \otimes  e_\flat$ with $v$ spacelike.
\end{lemma}
\begin{proof}
  By definition $G = \tilde e \otimes  \tilde v_\flat -\tilde v \otimes  \tilde e_\flat$ for two linearly indepedent vectors $\tilde e, \tilde v \in V$. If one of them is non-null, we set $\tilde v := v$ and decompose $V =  \spn{v} \oplus  v^\perp$. Thus $\tilde e = a v + e$ with $a \in \mathbb{R}$ and $e \in v^\perp$ and $G$ takes the form  $G = (a v + e) \otimes  v_\flat -v \otimes  (a v + e)_\flat = e \otimes  v_\flat -v \otimes  e_\flat$, as claimed. 
  If $\tilde e$ and $\tilde v$ are both null, consider $V = \spn{\tilde e} \, \widetilde{\oplus} \, (\tilde e)^c$ (we use $\widetilde{\oplus}$ because this direct sum is not by orthogonal spaces) where $(\tilde e)^c$ is a 
  spacelike complement of $\spn{\tilde e}$. Then we can write $\tilde v = a \tilde e + v'$, with $a \in \mathbb{R}$ and $v' \in \tilde e^c$ non-null. Thus $G = \tilde e \otimes  v'_\flat -v' \otimes \tilde e_\flat$, with $v'$ non-null and we fall into the previous case. All in all, $G = e \otimes v_\flat - v \otimes e_\flat$ with $e,v$ orthogonal. Consequently, either one of the vectors is spacelike or both are null and proportional which would imply $G =0$, against our hypothesis $\rank G = 2$.
\end{proof}

The decomposition $G = e \otimes  v_\flat -v \otimes  e_\flat$ is not unique even with the restriction of $v$ being spacelike unit and orthogonal to $e$. One can easily show that the remaining freedom  is given by the transformation $e' = a e - b \lrprod{e,e}v$, $v' = b e + a v$ with $a,b\in  \mathbb{R}$ restricted to $a^2 + b^2\lrprod{e,e} = 1$. Nevertheless, the square norm $\lrprod{e',e'}$ is invariant under this change, so the following definition makes sense:
\begin{definition}\label{defsimend}
Let $G \in \skwend{V}$ be simple, with $G = e \otimes  v_\flat -v \otimes  e_\flat$, $e,v \in V$ orthogonal with $v$ spacelike unit. Then $G$ is said to be {\it spacelike, timelike} or {\it null} if the vector $e$ is {\it spacelike, timelike} or {\it null} respectively. In the non-null case, $G$ is called spacelike (resp. timelike) unit whenever $\lrprod{e,e} = +1$ (resp. $\lrprod{e,e} = -1$).
\end{definition}

By Lemma \ref{simpleortho}, it is immediate that Definition \ref{defsimend}  comprises any possible simple endomorphism (up to a multiplicative factor).

We next obtain the necessary and sufficient conditions for a simple endomorphism $G$ to commute with a given $F \in \skwend{V}$. We first make the simple observation that the composition
of a one-form $e_{\flat}$ and a skew-symmetric endomorphism $F$ satisfies
(simply apply for sides to any $w \in V$)
\begin{align*}
e_{\flat} \circ F = -F(e)_{\flat}.
\end{align*}
An immediate consequence is that
for any pair of vectors $e,v \in V$ and $F\in \skwend{V} $ it holds
\begin{align}
F \circ ( e \otimes v_{\flat} ) = F(e)\otimes v_{\flat},\quad\quad ( e \otimes v_{\flat} ) \circ F = - e \otimes F(v)_{\flat}
\label{comp}
\end{align}
The following commutation result will be used later.

\begin{lemma}\label{commutingrank1}
Let $F,G \in \skwend{V}$ with $G = e \otimes  v_\flat -v \otimes  e_\flat$ simple and $e,v \in V$ as in Definition \ref{defsimend}. Then $[F,G] = 0$ if and only if there exist $\mu \in \mathbb{R}$ such that: 
\begin{equation}\label{commplanes}
 F(e) = \lrprod{e,e} \mu v, \qquad F(v) = - \mu e.
\end{equation}
 \end{lemma}
\begin{proof}
The commutator is $\lrbrkt{F,G} = F \circ G - G \circ F$
\begin{align}
 \lrbrkt{F ,G}   
& =  F \circ G - G \circ F 
= F \circ \lr{e \otimes v_\flat - v \otimes e_\flat} - \lr{e \otimes v_\flat - v \otimes e_\flat} \circ F \nn \\
 & =  F(e) \otimes v_\flat - F(v) \otimes e_\flat
+ e \otimes F(v)_{\flat} -v \otimes F(e)_{\flat},\label{commuteq}
\end{align}
where we have used \eqref{comp}.  The ``if'' part is obtained by direct calculation inserting \eqref{commplanes} in \eqref{commuteq}. 
To prove the ``only if'' part, the condition $[F,G] = 0$ requires the two  endomorphisms $F(e) \otimes v_\flat - v \otimes F(e)_\flat$ and $F(v) \otimes e_\flat - e \otimes F(v)_\flat$ to be equal. One such endomorphism is either identically zero or simple.  This implies that  $\spn{F(e),v}$ and $ \spn{e,F(v)}$ are either  both one dimensional or both two-dimensional and equal. In the first case, $F(v) = -\mu e$ and $F(e) = \alpha v$ for $\mu, \alpha \in \mathbb{R}$, which are determined by skew-symmetry to satisfy $\alpha = \mu \lrprod{e,e}$, so the lemma follows. The second case is empty, for it is necessary that $v = a e + b F(v)$ with $a,b \in \mathbb{R}$, which implies $\lrprod{v,v} =
\lrprod{a e + b F(v),v} = b \lrprod{F(v),v} = 0$, against the hypothesis of $v$ being spacelike.
 \end{proof}

\begin{corollary}\label{commutingrank1_cor}
  Let $G,G' \in \skwend{V}$ be simple, spacelike and linearly independent. Let $\{ e,v\}$,
  $\{e',v'\}$ be orthogonal spacelike vectors such that
  $G = e \otimes  v_\flat -v \otimes  e_\flat$ and
  $G' = e' \otimes  v'_\flat -v' \otimes  e'_\flat$. Then
  $[G,G'] =0$ if and only if $\{ e,v,e',v'\}$ are mutually orthogonal.
\end{corollary}
\begin{proof}
By the previous lemma $[G,G'] =0$ if and only if
 there exist $\mu \in \mathbb{R}$ such that
  \begin{equation}
   G(e') = \lrprod{e',v} e - \lrprod{e',e}v  =\mu v', \quad \quad G(v') = \lrprod{v',v}e - \lrprod{v',e}v = -\mu e'. \label{Gep}
  \end{equation}
  If $\mu \neq 0$, then $\spn{e,v} = \spn{e',v'}$ and $G$ and $G'$ are proportional, agains hypothesis. Thus, $\mu = 0$ and by \eqref{Gep}
  the set $\{ e,v,e',v'\}$ is mutually orthogonal.
  \end{proof}

\section{$\Lor^+(1,d-1)$-classes} \label{seclorclass}

In this section we use the canonical form of Section \ref{seccanonform} to characterize skew-symmetric endomorphisms of $V$ under the adjoint action of the orthochronous Lorentz group $\Lor^+(1,d-1)$. Recall that this is the subgroup of $O(1,d-1)$ preserving time orientation. The corresponding classes of skew-symmetric endomorphisms are also known as the adjoint orbits or conjugacy classes and  we denote them by $[F]_{\Lor^+}$ for a given element $F \in \skwend{V}$. The charaterization of these orbits by a set of independent invariants is known and it
 can be found in \cite{Kdslike} in terms of two-forms and other references such as \cite{burgoyne77}. What we do here is, first, to give an alternative way to characterize  the orbits $[F]_{\Lor^+}$ by a convenient set of invariants and second, to show that the canonical form is the same for every element in a given orbit. This makes the canonical form specially useful as a tool for problems with $\Lor^+(1,d-1)$ invariance.
 
 Although we restrict here to the orthochronous component $O^+(1,d-1)$ because of its relation with conformal transformations of the sphere $\mathbb{S}^{d-2}$ (see Section \ref{secconfvecs}), from this case, the orbits of the full group $O(1,d-1)$ are easy to determine. Recall that the time-reversing component $O^-(1,d-1)$ is one-to-one with $O^+(1,d-1)$. In an orthonormal basis, we can map elements $\Lambda^- \in O^-(1,d-1)$ to elements in $\Lambda^+ \in O^+(1,d-1)$ by e.g. $\Lambda^+ := \Lambda^- \eta$, where $\eta = \mathrm{diag}(-,+,\cdots,+)$ is an element in $O^-(1,d-1)$ with the same matrix form as the metric. Then
 \begin{equation}
  \Lambda^+ F (\Lambda^+)^{-1} = \Lambda^- \eta F \eta (\Lambda^-)^{-1} =  -\Lambda^-  F  (\Lambda^-)^{-1},
 \end{equation}
where the last equality follows from skew-symmetry. Hence, the elements $F$ and $-F$ belong to the same orbit under the action of $O(1,d-1)$. If we denote the space of orbits of the orthochronous component and full group respectively by $[O^+] = \skwend{V}/O^+(1,d-1)$ and $[O] = \skwend{V}/O(1,d-1) $, this is expressed as $[O] = [O^+]/\mathbb{Z}_2$. This fact is of course well-known and can be inferred from general references e.g. \cite{hall2004symmetries}.  
 

A consequence of equation \eqref{paritycharpol} is that the characteristic polynomial of $F\in \skwend{V}$ must have the form
\begin{equation}\label{coefpoly}
 \mathcal{P}_F(x) = x^d + \sum_{b=1}^{\q} c_b x^{d-2b},
\end{equation}
where we have introduced  $\q := [\frac{d}{2}]$.
 The coefficients $c_b$ can be obtained using the Fadeev-LeVerrier algorithm, summarized by the following matrix determinant \cite{gantmacher1960theomat}:
\begin{equation}\label{fadlever}
 c_{b} = \frac{1}{(2b)!}\left| \begin{array}{ccccc}
                               \mathrm{Tr}~F & 2b-1 & 0 & \cdots & 0 \\
                               \mathrm{Tr}~F^2 & \mathrm{Tr}~F & 2b-2 & \cdots & 0 \\
                               \vdots & \vdots & & & \vdots \\
                              \mathrm{Tr}~F^{2b-1} & \mathrm{Tr}~F^{2b-2} & \cdots & \cdots & 1\\
                               \mathrm{Tr}~F^{2b} & \mathrm{Tr}~F^{2b-1} & \cdots & \cdots & \mathrm{Tr}~F
                              \end{array}
\right|.
\end{equation}
Since the traces of odd powers vanish by skew-symmetry, the coefficients $c_b$  depend on the entries of $F$ only through the traces of the squared powers of $F$:
\begin{equation}\label{Ilambmu}
I_b = \frac{1}{2} \mathrm{Tr}~(F^{2b}),~~b = 1,\cdots,\q.
\end{equation}

Recall that the adjoint respresentation $\mathrm{Ad}$ of a matrix Lie group $G$ is a linear representation of $G$ on its Lie algebra automorphisms $\mathrm{Aut}(\mathfrak{g})$ given by
\begin{align*}
  \mathrm{Ad}: G & \longrightarrow \mathrm{Aut} (\mathfrak{g}) \\
  g  &\longrightarrow \mathrm{Ad}(g):=\mathrm{Ad}_g : \begin{array}{lll}
                                      & & \\
                                      \mathfrak g & \rightarrow & \mathfrak g \\
                                     X  & \rightarrow & g X g^{-1}
                                          \end{array}.
                                   \end{align*}
The traces $I_b$ are obviously invariant under the adjoint action of $\Lor^+(1,d-1)$ and so are the coefficients $c_b$. Another invariant that plays an important role in the classification of conjugacy classes is the rank of $F$. Since this is always even, we denote it by 
\begin{equation}
 \rank F = 2  r,
\end{equation}
and clearly $r \leq \q$. From now we say {\bf rank parameter} to refer to
$r$.
In the following proposition we show that this set of invariants actually identifies the
canonical form.
\begin{proposition}\label{invsequalcanonform}
 Let $F, \widetilde F \in \skwend{V}$, for $V$ Lorentzian of dimension $d$. Then  the invariants $\{c_b,r\}$ and $\{\widetilde c_b,\widetilde r \}$ of $F$ and $\widetilde F$ respectively are equal if and only if their canonical forms given by Theorem \ref{theocanonicalF} are the same. 
\end{proposition}
\begin{proof}
The ``if'' part ($\Leftarrow$) is trivial, because the invariants $c_b, r$ are independent on the basis, so they can be calculated in a canonical basis. Hence, same canonical form implies same invariants. For the ``only if'' part ($\Rightarrow$), we notice that if the coefficients $c_b$ and $\widetilde c_b$ of $\mathcal{P}_F$ and $\mathcal{P}_{\widetilde F}$  are equal, so are their characterisic polynomials, the multiplicities of their zero eigenvalue and the polynomials $\mathcal{Q}_{F^2}$ and $\mathcal{Q}_{\widetilde F^2}$ (equation \eqref{defQF2}). Since $\rank F = \rank \widetilde F$, Corollary \ref{lemmaequivcases} implies that $\ker F$ and $\ker \widetilde F$ must have the same causal character. The canonical form only depends on the
 roots $\mathcal{Q}_{F^2}$  and the causal character of $\ker F$ through Definition \ref{defgammamu}. Thus, $F$ and $\widetilde F$ must have the same canonical form.
\end{proof}

We now characterize the classes $[F]_{\Lor^+}$ in terms of the same invariants given in Proposition \ref{invsequalcanonform}. As mentioned above, this result is known \cite{Kdslike}, but we give here an alternative and very simple proof based on our canonical form:

\begin{theorem}{\cite{Kdslike}}\label{theoLorendo}
 Let $F, \widetilde F \in \skwend{V}$, for $V$ Lorentzian of dimension $d$. Then their invariants $\{c_b,r\}$ and $\{\widetilde c_b,\widetilde r \}$ are the same if and only if $F$ and $\widetilde F$ are $\Lor^+(1,d-1)$-related. 
\end{theorem}

\begin{proof}
 The if ($\Leftarrow$) part is immediate, since it is trivial from their definitions that the quantities $\lrbrace{c_b,r}$ are Lorentz invariant.
 To prove the ``only if'' ($\Rightarrow$), by Proposition \ref{invsequalcanonform}, $F$ and $\widetilde F$ have the same canonical form in canonical bases $B$ and $\widetilde B$ respectively. By definition (c.f. Theorem \ref{theocanonicalF}), these bases are unit, future oriented and orthonormal. Thus, the transformation taking $B$ to $\widetilde B$ transforms $F$ into $\widetilde F$ and both must be $\Lor^+(1,d-1)$-related.
 
  \end{proof}

Theorem \ref{theoLorendo} stablishes the necessary and sufficient conditions for two endomorphisms to be $\Lor^+(1,d-1)$-related. Combining this result with Proposition \ref{invsequalcanonform}, we find that the canonical form (hence the parameters $\cc, \mu_i^2$ or $\aa,\bb,\mu_i^2$) totally define the equivalence class of skew-symmetric endomorphisms up to $\Lor^+(1,d-1)$ transformations. 
Moreover, we emphasize that this form is the same for every equivalence class, unlike other canonical (or normal) forms based on the classification of $\skwend{V}$, such as the one in \cite{djokovic83}, where they seek irreducibility of the blocks, so they must give two different forms to cover every case.

Next, we discuss some facts about the coefficients of the characteristic polynomial, also stated in \cite{Kdslike}, where the proof is only indicated, and which can now be easily proven using the canonical form.



\begin{lemma}\label{lemmaIrindepe}
 Let $F\in \skwend{V}$ be non-zero and let $2 r = \rank F$. Then $c_r>0,~c_r = 0,~c_r <0$ if and only if $\ker F$ is timelike, null or spacelike (or zero) respectively. Moreover, if $r < \q$,  $c_{\q} = c_{\q-1} = \cdots = c_{r+1} = 0$. 
\end{lemma}

\begin{proof}
 Taking into account that the parities of $d$ and $m_0$ are equal (Lemma \ref{remarkmultiplic}),  $\q - [\frac{m_0}{2}] = [\frac{d}{2}] - [\frac{m_0}{2}] = \frac{d-m_0}{2}$, so  equation \eqref{coefpoly}  can be rewritten 
%
\begin{equation}\label{cbmultip}
 \mathcal{P}_F(x) = x^{m_0}\lr{x^{d-m_0} + \sum\limits_{b= 1}^{\q-[m_0/2] }c_b x^{d-m_0-2b}} =x^{m_0}\lr{x^{d-m_0} + \sum\limits_{b= 1}^{\frac{d-m_0}{2} }c_b x^{d-m_0-2b}},
 \end{equation}
 where we have explicitly substituted all zero coefficients by extracting the common factor $x^{m_0}$, thus the remaining coefficients $c_b \neq 0$ for $b = 1, \cdots, (d-m_0)/2$. By Corollary \ref{lemmaequivcases}, $\ker F$ degenerate if and only if $2 r = d - m_0+2$ and $m_0 >2$, so the sum in \eqref{cbmultip} runs up to $(d-m_0 )/2 = r-1$, which means $c_r = c_{r+1} = \cdots = c_q = 0$, as stated in the lemma. Also by Corollary \ref{lemmaequivcases}, $\ker F$ non-degenerate if and only if $2 r = d - m_0$. In this case, the sum in \eqref{cbmultip} runs up to $(d-m_0)/2 = r$, hence $c_r \neq 0$ and if $r < \q$, the next coefficients vanish  $c_{r+1} = c_{r+2}= \cdots = c_\q = 0$. In addition $c_r$ is the independent term in the polynomial in parentheses. Let $\mu_1, \cdots, \mu_\lambda$ be all the non-zero parameters among the $\{ \mu_i \}$ of the canonical form of $F$ given in \eqref{eqFgammamu}. By equation \eqref{eqcharac}, $c_r$ can be written for $d$ odd: 
\begin{equation*}
 c_r = \cc \mu_1^2 \cdots \mu_\lambda^2.
\end{equation*}
Then, the sign of $\cc$ determines the sign of $c_r$ and,  by Remark \ref{remarcanonF3}, also the causal character of $\ker F_\cc$, hence, the causal character of $\ker F$ in accordance with the stament of the lemma. For $d$ even, also from \eqref{eqcharac} we have 
 \begin{equation*}
 c_r = -\frac{\bb^2}{4} \mu_1^2 \cdots \mu_\lambda^2<0 \quad(\bb \neq 0), \quad \quad c_r = \aa \mu_1^2 \cdots \mu_\lambda^2 \quad (\bb = 0),
\end{equation*}
 where the expression for $\bb = 0$ follows because in this case either $\mu_t$ or $\mu_s$ (or both) vanish, hence either $c_r = \mu_s^2 \mu_1^2 \cdots \mu_\lambda^2$ or  $c_r = -\mu_t^2 \mu_1^2 \cdots \mu_\lambda^2$ and $\aa$ equals $\mu_s^2$ in the first situation and
$-\mu_t^2$ in the second. By Remark \ref{remarcanonF4}, when $\bb \neq 0$ we have $\ker F_{\aa \bb} = \{0\}$ and hence ker F is always spacelike or zero and when $\bb = 0$, the causal character of $\ker F_{\aa \bb}$ (and that of $\ker F$) is determined by the sign of $\aa$ in accordance with the statement of the lemma.

\end{proof}
 
 \begin{remark}\label{remarkrankcb}
   A converse version of Lemma \ref{lemmaIrindepe} also holds, in the sense that the number $\nu$ of last vanishing coefficients  
   restricts the allowed rank parameters $r$.   Let $\nu$ be defined by $\nu=0$ if $c_q \neq 0$ and, otherwise, by the largest natural number satisfying
   $c_q = c_{q-1} = \cdots c_{q-\nu+1} =0$. By equation \eqref{cbmultip} it follows $\nu = [m_0/2]$, and since the dimension $d$ and $m_0$ have the same parity (cf. Lemma \ref{remarkmultiplic}), 
  $d- m_0 = 2 [ d/2 ] - 2 [ m_0/2] = 2 (q- \nu)$ which in particular shows that $\nu$ determines $m_0$ uniquely.
 If $m_0 > 2$, by Corollary \ref{lemmaequivcases} the rank parameter admits two possibilities $r = \lrbrace{\q-\nu,\q -\nu+1}$,
    each of which determined by the causal character of $\ker F$.  If $m_0 \leq 2$, also by Corollary \ref{lemmaequivcases} the $\ker F$ degenerate case cannot occur and $r = (q - \nu)$ is uniquely determined. In particular, if $d=4$, $r$ is always determined by $c_1, c_2$, because $r=2$ happens if and only if
    $\nu = 0$ and otherwise $r=1$ (unless $F$ is identically zero, in which case $r = 0$).
  

 \end{remark}

\subsection{Structure of $\skwend{V}/\Lor^+(1,d-1)$}\label{secstructure}

By Theorem \ref{theoLorendo}, the $q$-tuple $(c_1, \cdots, c_\q)$ corresponding to the coefficients of the characteristic polynomial of a skew-symmetric endomorphism, does not suffice to determine a point in the quotient space $\skwend{V}/\Lor^+(1,d-1)$, since generically two ranks are possible (dimensions three and four are an exception). As dicussed in Remark \ref{remarkrankcb}, for a number $\nu$ of last vanishing coefficients $c_b$, the allowed rank parameters are $r \in \{q-\nu, q -\nu + 1 \}$, and $r = q -\nu + 1$ is only possible provided $m_0>2$ (in particular, when $c_q \neq 0$ then necessarily $r = q$). One says that there is a {\it degeneracy} for the value of the rank at certain points in the space of coefficients $c_b$. In the submanifold 
$\{c_q = \cdots = c_{q-\nu+1} = 0, c_{q - \nu} \neq 0\}$, the possible rank parameters are $r \in \{q-\nu, q -\nu + 1 \}$. When a boundary point
where the number of last vanishing coefficients increases by exactly one
is approached, the rank parameter may remain equal to $q-\nu$ or jump to $q- \nu -1$ (note that while the coefficients $c_i$ are continuous functions of $F$, the rank is only lower
semicontinuous, e.g. \cite{lange13optm}).
As we shall see in this section, this behaviour
gives rise to special limit points which make the space of  parameters defining the canonical form (i.e. the space of conjugacy classes) a non-Hausdorf topological space, when endowed with the natural quotient topology.
Let us start by locating these limit points using the canonical form.
Degeneracies can only occur in dimensions $d=5$ or larger because
in dimension three the rank is two for any non-trivial $F$ and in dimension four the rank is uniquely determined by the invariants (c.f. Remark \ref{remarkrankcb}). We thus consider first the case $d=5$ and then extend to all values $d \geq 5$. In $d=5$ the space of parameters $\mathcal{A}$ defining the $[F]_{\Lor^+}$ classes is (see fig. \ref{plotregion} )
\begin{equation}
 \mathcal{A} := \lrbrace{ (\cc,\mu^2)  \in \mathbb{R} \times \mathbb{R}^+ \mid \cc \geq \mu^2~\mbox{if}~\cc>0}.
\end{equation}
Consider a $[F]_{\Lor^+}$ in the region $\mathcal{R}_+ :=\lrbrace{ \cc \geq \mu^2 > 0}$ and let $F$ be a representative of $[F]_{\Lor^+}$ in a canonical basis $B = \lrbrace{e_\alpha}_{\alpha = 0, \cdots, 4}$, that is 
\begin{equation}\label{F5can}
   F =\left(
\begin{array}{cccc}
 0 & 0 & -1 + \frac{\cc }{4}  \\
 0 & 0 & -1-\frac{\cc }{4}  \\
 -1+\frac{\cc }{4} & 1+\frac{\cc }{4} & 0  \\
\end{array}
\right)
\oplus
\left(
\begin{array}{cccc}
 0 & -\mu \\
 \mu & 0  \\
\end{array}
\right).
 \end{equation}
Let us define the functions $C_{\pm}(x) : = \frac{1}{x} \pm \frac{x}{4}$. Then, the following change of basis to $B' = \lrbrace{e'_\alpha}$ is well defined in $\mathcal{R}_+$:
\begin{equation}\label{basiscan4}
 \begin{array}{lclcl}
 e'_0 = C_{+}(\mu)\lr{C_{+}(\sqrt{\cc}) e_0 + C_{-}(\sqrt{\cc})e_1} - C_{-}(\mu)e_4, 
& & 
 e'_2  = -e_3,
 & &
 e'_4 = C_{-}(\sqrt{\cc})e_0 + C_{+}(\sqrt{\cc})e_1
 \\
 e'_1 = -C_{-}(\mu)\lr{C_{+}(\sqrt{\cc}) e_0 + C_{-}(\sqrt{\cc})e_1} + C_{+}(\mu)e_4,
 &  & 
e'_3  = -e_2.
 \end{array}
 \end{equation}
By direct calculation, $F$ is written in basis $B'$ as
\begin{equation}\label{switchgammamumatrix}
   F =\left(
\begin{array}{cccc}
 0 & 0 & -1 + \frac{\mu^2 }{4}  \\
 0 & 0 & -1-\frac{\mu^2 }{4}  \\
 -1+\frac{\mu^2 }{4} & 1+\frac{\mu^2 }{4} & 0  \\
\end{array}
\right)
\oplus
\left(
\begin{array}{cccc}
 0 & -\sqrt{\cc} \\
 \sqrt{\cc} & 0  \\
\end{array}
\right).
 \end{equation}
The basis $B'$ is non-canonical because $\mu^2< \cc$. However, if we vary the parameters so that $\mu \rightarrow 0$ (keeping $\cc$ unchanged), the matrix \eqref{switchgammamumatrix} becomes canonical (i.e. of the form \eqref{eqFgammamu}) in the limit and the class $[\lim_{\mu\rightarrow 0} F]_{\Lor^+}$ is given by $l_1=(0,\cc)$. On the other hand, $F$ in canonical form \eqref{F5can} also admits a limit $\mu \rightarrow 0$, which is also canonical and whose representative $[\lim_{\mu\rightarrow 0} F]_{\Lor^+}$ is given by $l_2=(\cc,0)$. Both limits are defined by the same sequence of points, because the transformation \eqref{basiscan4} is invertible in $\mathcal{R}_+$. However this sequence has two different limit points. As a consequence, the space of canonical matrices, and therefore the quotient space $\skwend{V}/\Lor^+(1,d-1)$, inherits a  non-Hausdorff topology.

Something similar happens in the region $\mathcal{R}_- :=\lrbrace{\cc <0, \mu > 0}$. Let $F$ be a representative in  canonical form of a point $[F]_{\Lor^+}$ in this region. Then, $F$ has a timelike eigenplane $\Pi_t$ with eigenvalue $\sqrt{|\cc|}$ (c.f. Remark \ref{remarcanonF3}), a spacelike eigenvector $e$ as well as spacelike eigenplane $\Pi_s$ with eigenvalue $\mu$. Thus $V = \Pi_t \oplus \spn{e} \oplus \Pi_s$ and there exist a (non-canonical) basis $B'$ adapted to this decomposition, into which $F$ takes the form
\begin{equation}\label{explicittime}
   F =\left(
\begin{array}{cccc}
 0 & \sqrt{|\cc|} & 0  \\
 \sqrt{|\cc|} & 0 & 0  \\
 0 & 0 & 0  \\
\end{array}
\right)
\oplus
\left(
\begin{array}{cccc}
 0 & -\mu \\
 \mu & 0  \\
\end{array}
\right).
 \end{equation}
 Keeping $\mu$ unchanged, expression \eqref{explicittime} has a limit $\cc \rightarrow 0$, which has a spacelike eigenplane $\Pi_s$ of eigenvalue $\mu$ and it is identically zero on $\Pi^\perp$ . Hence, $\ker F$ is timelike and using Definition \ref{defgammamu}, the canonical form of this limit $\lim_{\aa \rightarrow 0} F$ is given by $\cc' = \mu^2$ and $\mu' = 0$. Thus $[\lim_{\cc\rightarrow 0} F]_{\Lor^+}$ is represented by the point $l_2 = (\mu^2,0)$. On the other hand, in a canonical basis \eqref{F5can}, $F$ also admits a limit $\cc\rightarrow 0$, whose class $[\lim_{\cc\rightarrow 0} F]_{\Lor^+}$ is obviously represented by the point $l_1 = (0, \mu^2)$.

 \begin{figure}[h!] 
 \centering 
    \psfrag{Rm}{$\mathcal{R}_-$}
    \psfrag{Rz}{$\mathcal{R}_0$}
    \psfrag{Rp}{$\mathcal{R}_+$}
    \psfrag{pe}{$l_2$}
    \psfrag{qu}{$l_1$}
    \psfrag{mu2}{$\mu^2$} 
    \psfrag{sig}{$\cc$} 
  \includegraphics[scale=0.5]{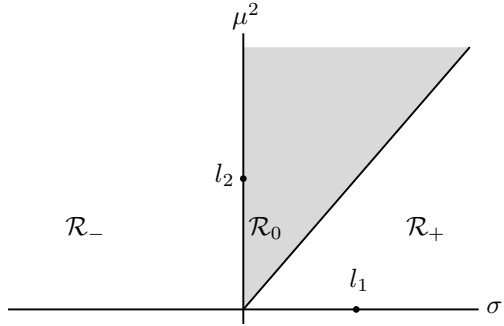}
  \caption{\small \it Representation of $\skwend{V}/\Lor^+(1,d-1)$ in the region $\mathcal{A} \subset \mathbb{R}^2$. The shadowed region is not included.}\label{plotregion} 
 \end{figure} 

The same reasoning can be carried out to arbitrary odd dimension. 
First, define the regions 
\begin{equation*}
 \mathcal{R}_+^{(d,0)} := \lrbrace{\cc \geq \mu_1^2 \geq \cdots \geq \mu_{p}^2 > 0}\quad \mbox{and}\quad\mathcal{R}_-^{(d,0)}:=\lrbrace{\cc < 0,\mu_1^2 \geq \cdots \geq \mu_{p}^2 > 0}
\end{equation*}
and also the limit regions
\begin{equation*}
 \mathcal{R}_0^{(d,0)} := \lrbrace{\cc = 0,\mu_1^2 \geq \cdots \geq \mu_{p}^2 > 0}\quad \mbox{and}\quad \mathcal{R}_+^{(d,1)}:=\lrbrace{\aa \geq \mu_1^2 \geq \cdots \geq \mu_{p-1}^2 > \mu_{p}^2 = 0}.
\end{equation*}
 Consider representatives $F_+$ and $F_-$ (in canonical form) of points $(\cc^+, (\mu_1^+)^2, \cdots,(\mu_p^+)^2)$ and $(\cc^-, (\mu_1^-)^2, \cdots, (\mu_p^-)^2)$ in the regions $\mathcal{R}_+^{(d,0)}$ and $\mathcal{R}_-^{(d,0)}$ respectively. Then $F_+$ has a spacelike eigenplane $\Pi_s^+$ with eigenvalue $\mu_{p}^+$ as well as timelike eigenvector $e^+$ and spacelike eigenplane $\Pi_t^+$ with eigenvalue $\sqrt{\cc^+}$. Restricting to the subspace $W^+ = \spn{e^+} \oplus \Pi_t^+ \oplus \Pi_s^+$ we can repeat the procedure followed for the five dimensional case and conclude that $[\lim_{\mu_{p}^+ \rightarrow 0} F_+]$ has simultaneously limits on the points $(\cc^+, (\mu_1^+)^2, \cdots,(\mu_{p-1}^+)^2,  0) \in \mathcal{R}_+^{(d,1)} $ and $( 0, (\mu_1^+)^2, \cdots,(\mu_{p}^+)^2) \in \mathcal{R}_0^{(d,0)}$. Analogously $F_-$ has a spacelike eigenplane $\Pi_s^-$ with eigenvalue $\mu_{p}^-$ as well as spacelike eigenvector $e^-$ and timelike eigenplane ${\Pi'}_s^-$ with eigenvalue $\sqrt{|\cc^-|}$. Restricting to the subspace $W^- =   \Pi_s^- \oplus \spn{e^-} \oplus {\Pi'}_s^-$, the above arguments for the five dimensional case show that  $[\lim_{\cc^- \rightarrow 0} F_-]$ has simultaneous limits on the points $((\mu_{p}^-)^2, (\mu_{1}^-)^2, \cdots, (\mu_{p-1}^-)^2,  0) \in \mathcal{R}_+^{(d,1)} $ and $(0, (\mu_{1}^-)^2, \cdots, (\mu_{p}^-)^2) \in \mathcal{R}_0^{(d,0)}$. Thus the regions $\mathcal{R}_+^{(d,0)}$ and $\mathcal{R}_ -^{(d,0)}$ limit simultaneously with $\mathcal{R}_+^{(d,1)} $ and $ \mathcal{R}_0^{(d,0)}$ as $\mu_{p}$ and $\sigma$ tend to zero respectively. Indeed, the same ideas can be applied again to $\mathcal{R}^{(d,1)}_+$ and $\mathcal{R}^{(d,1)}_- :=\lrbrace{\cc < 0, \mu_1^2 \geq \cdots \geq \mu_{p-1}^2 > \mu_{p}^2 = 0} $, so that they also limit simultaneously, as $\mu_{p-1}$ and $\sigma$ go to zero respectively,  with $\mathcal{R}_0^{(d,1)} := \lrbrace{ \cc = 0, \mu_1^2 \geq \cdots \geq \mu_{p-1}^2 > \mu_{p}^2 = 0}$ and $\mathcal{R}_+^{(d,2)}:= \lrbrace{ \cc > 0, \mu_1^2 \geq \cdots \geq \mu_{p-1}^2 > \mu_{p-1}^2 = \mu_p^2 = 0 }$. In general, the regions $\mathcal{R}^{(d,i)}_\pm$ analogously defined, i.e. where $i$ gives the number of vanishing parameters $\mu_p = \cdots = \mu_{p-i} = 0$ and the subindex $\pm$ gives the sign of $\cc$, have simultaneous limits in $\mathcal{R}_0^{(d,i)}$ and $\mathcal{R}_+^{(d,i+1)}$, where the subindex $_0$ stands for vanishing $\cc$.

For the even dimensional case (with $d\geq 6$), notice that the canonical form \eqref{eqFgammamu} with $\bb = 0$ is equivalent to the odd dimensional case direct sum with a one dimensional zero endomorphism (of a Riemannian line). Hence, the previous reasoning for odd dimensions also applies for even dimensions and $\bb =  0$. For example, consider in $d = 6$ dimensions the regions $\mathcal{R}_+ = \lrbrace{\bb = 0, \aa \geq \mu^2 > 0}$ and $\mathcal{R}_- = \lrbrace{\bb = 0, \aa <0, \mu^2 > 0}$. Then they both assume limit in $\mathcal{R}_0 = \lrbrace{\bb = 0, \aa = 0,  \mu^2 > 0}$ and $\mathcal{R}_+^{(d,1)} = \lrbrace{\bb = 0, \aa >   \mu^2 = 0}$. 
Notice that if we keep $\bb \neq 0$ no {{\it degenerate} limits of this kind occur. This can be justified as follows. Let $\mu_t, \mu_s$ be defined as in \eqref{eqmutmus}. Then, it can be readily checked that $\det F = -\mu_t^2 \mu_s^2 \mu^2$, so if $\aa, \bb, \mu \neq 0$, then $\rank F = 6$. If we keep $\bb \neq 0 $ (thus both $ \mu_s, \mu_t$ are different from zero), $\mu \neq 0$ and make $\aa \rightarrow 0$, the limit must have always $\rank F = 6$. Hence, it is not possible that a limit $\aa \rightarrow 0$ ends at two points with different rank. Similarly, keeping $\bb \neq 0$, the limit $\mu \rightarrow 0$ always has $\rank F = 4$ and therefore, $\mu \rightarrow 0$ limits cannot be degenerate either. The generalization to arbitrary even dimensions with $\bb  = 0 $ is also straightforward from the odd dimensional case discussed above, which we now summarize in the following remark:

\begin{remark}\label{remarklimits}
 In the case of $d$ odd, consider the subset of $\mathbb{R}^q$ given by
 \begin{align*}
\mathcal{A}^{(odd)} := \{  \lr{\cc, \mu_1^2, \cdots, \mu_{p}^2} \in  \mathbb{R} \times \lr{\mathbb{R}^+}^{p} \mid \mu_1^2 \geq \cdots \geq \mu_{p-1}^2~\mbox{and if $\cc >0$}, ~ \cc \geq \mu_1^2 \geq \cdots \geq \mu_{p}^2 \}.
\end{align*}
 Define also the subsets of $\mathcal{A}^{(odd)}$ given by
\begin{align*}
 \mathcal{R}_+^{(d,i)} & := \lrbrace{(\sigma, \mu_1^2, \cdots, \mu_{p}^2) \in \mathcal{A}^{(odd)} \mid \cc \geq \mu_1^2 \geq \cdots \geq \mu_{p-i}^2 > \mu_{p-i+1}^2 = \cdots = \mu_{p}^2 = 0}, \\
 \mathcal{R}_-^{(d,i)} & := \lrbrace{(\sigma, \mu_1^2, \cdots, \mu_{p}^2) \in \mathcal{A}^{(odd)} \mid \cc <0, \mu_1^2 \geq \cdots \geq \mu_{p-i}^2 > \mu_{p-i+1}^2 = \cdots = \mu_{p}^2 = 0},\\
 \mathcal{R}_0^{(d,i)} & := \lrbrace{(\sigma, \mu_1^2, \cdots, \mu_{p}^2) \in \mathcal{A}^{(odd)} \mid \cc = 0, \mu_1^2 \geq \cdots \geq \mu_{p-i}^2 > \mu_{p-i+1}^2 = \cdots = \mu_{p}^2 = 0}.
\end{align*}
Then in the quotient topology of $\skwend{V}/\Lor^+(1,d-1)$ the sequences of $\mathcal{R}_\pm^{(p-i)}$ with limit at $\mathcal{R}_0^{(d,i)}$ also have limit at $\mathcal{R}_+^{(d,i+1)}$.

In the case of $d$ even, first define $\mu_s$ as in \eqref{eqmutmus} and let  $\mathcal{A}^{(even)}$ be the subspace of $\mathbb{R}^q$ given by:
\begin{align*}
  \mathcal{A}^{(even)}  :=  \{  \lr{\aa, \bb, \mu_1^2, \cdots, \mu_{p}^2} \in  \mathbb{R} \times \lr{\mathbb{R}^+}^{p} \mid \mu_1^2 \geq \cdots \geq \mu_{p}^2 \mbox{~and if $\bb \neq 0$ or $\aa >0$, } 
   \mu_s^2 \geq \mu_1^2 \geq \cdots \geq \mu_{p}^2 \}. 
\end{align*}
Define also the following subsets of $\mathcal{A}^{(even)}$
\begin{align*}
 \mathcal{R}_+^{(d,i)} & := \lrbrace{(\sigma, \mu_1^2, \cdots, \mu_{p}^2) \in \mathcal{A}^{(even)} \mid  \bb = 0, \cc \geq \mu_1^2 \geq \cdots \geq \mu_{p-i}^2 > \mu_{p-i+1}^2 = \cdots = \mu_{p}^2 = 0}, \\
 \mathcal{R}_-^{(d,i)} & := \lrbrace{(\sigma, \mu_1^2, \cdots, \mu_{p}^2) \in \mathcal{A}^{(even)} \mid \bb = 0,  \cc <0, \mu_1^2 \geq \cdots \geq \mu_{p-i}^2 > \mu_{p-i+1}^2 = \cdots = \mu_{p}^2 = 0},\\
 \mathcal{R}_0^{(d,i)} & := \lrbrace{(\sigma, \mu_1^2, \cdots, \mu_{p}^2) \in \mathcal{A}^{(even)} \mid \bb = 0, \cc = 0, \mu_1^2 \geq \cdots \geq \mu_{p-i}^2 > \mu_{p-i+1}^2 = \cdots = \mu_{p}^2 = 0}.
\end{align*}
Then in the quotient topology of $\skwend{V}/\Lor^+(1,d-1)$ the sequences of $\mathcal{R}_\pm^{(d,i)}$ with limit at $\mathcal{R}_0^{(d,i)}$ also have limit at $\mathcal{R}_+^{(d,i+1)}$.

%
%

\end{remark}

%

%

\section{Conformal vector fields}\label{secconfvecs}

 One interesting applications of our previous results is based on the relation between skew-symmetric endomorphisms and the set of conformal Killing vector fields (CKVFs) of the $n$-sphere, $\mathrm{CKill}(\mathbb{S}^n)$. These vector fields are the generators of the conformal transformations of the $n$-sphere $\conf{\mathbb{S}^n}$, i.e. the group of transformations $\psi_\Lambda$ that scale the spherical metric $\gamma$, $\psi_\Lambda^\star(\gamma) = \Omega^2 \gamma$, where $\Omega$ is a smooth positive function of $\mathbb{S}^n$.  A standard technique to describe these transformations consists in viewing $\mathbb{S}^n$ as the (real) projectivization of the null future cone in $\mathbb{M}^{1,n+1}$, in such a way that $\conf{\mathbb{S}^n}$ is induced from the isometries of $\mathbb{M}^{1,n+1}$. This is discussed in detail for the four dimensional case in \cite{PenroseRindVol1} and in arbitrary dimensions in \cite{Kdslike} and in \cite{IntroCFTschBook} (the latter considers arbitrary signature and the projectivization of the null ''cone'' in $\mathbb{M}^{p+1,q+1}$, giving $\mathbb{S}^p \times \mathbb{S}^q$).  
 This procedure stablishes a group homomorphism $\psi: O(1,n+1) \rightarrow \mathrm{Conf}(\mathbb{S}^n),~\Lambda \mapsto \psi_\Lambda$, which is one-to-one when restricted to the orthochronous component $\Lor^+(1,n+1) \subset O(1,n+1)$. 
 
 The Euclidean space $\mathbb{E}^n = (\mathbb{R}^n,g_E)$ and $\mathbb{S}^n$ are well-known to be conformally related via the stereographic projection $St_N: \mathbb{S}^n\backslash\{N\} \rightarrow \mathbb{E}^n$, where $N$ denotes the point  w.r.t. which the projection is taken. Observe that the stereographic projection depends not only on the point $N$ but also on the (signed) distance between $N$
 and the plane onto which the projection is performed. We do not reflect this dependence in the notation for simplicity.

 Hence, the composition of a transformation $\psi_\Lambda \in \conf{\mathbb{S}^n}$ with the stereographic projection yields $St_N \circ \psi_\Lambda \circ St_N^{-1}=:\phi_\Lambda \in \conf{\mathbb{E}^n}$, which is a conformal transformation of $\mathbb{E}^n$. Strictly speaking, these transformations are not diffeomorphisms of $\mathbb{E}^n$, as they require to remove the two points $p_1, p_2 \in \mathbb{E}^n$ satisfying $\psi_\Lambda \circ St_N^{-1}(p_1) = N$ and $\psi_\Lambda^{-1} \circ St_N^{-1}(p_2) = N$, which are the ``preimage'' and the ``image'' of infinity under $\phi_\Lambda$ respectively. Nevertheless, since  $\psi: O(1,n+1) \rightarrow \mathrm{Conf}(\mathbb{S}^n)$ is a group homomorphism, so is $\phi: O(1,n+1) \rightarrow \mathrm{Conf}(\mathbb{E}^n),~\Lambda \mapsto \phi_\Lambda$ as well as the map which assigns $\psi_\Lambda \mapsto \phi_\Lambda$. In that sense $\conf{\mathbb{S}^n}$ and $\conf{\mathbb{E}^n}$ are the same. 
These group homomorphisms, induce Lie algebra homomorphisms between $\skwend{\mathbb{M}^{1,n+1}}$,  $\mathrm{CKill}(\mathbb{S}^n)$ and $\mathrm{CKill}(\mathbb{E}^n)$ (the vector fields generating $\conf{\mathbb{E}^n}$). 
 The precise form of these maps depends, firstly, on the representative used to describe the projective cone (i.e. $\mathbb{S}^n$) and secondly on the point $N$ as well as on the signed distance from this point to the plane. In \cite{Kdslike}, the morphism 
  \begin{equation}\label{ximap}
  \begin{array}{rcl}
   \xi:=\phi_\star: \skwend{\mathbb{M}^{1,n+1}} & \longrightarrow & \mathrm{CKill}(\mathbb{E}^n),\\
   F & \longmapsto & \xi(F) =:\xi_F,
  \end{array}
 \end{equation} 
  is constructed\footnote{The method in \cite{Kdslike} is based on the unit spacelike hyperboloid in Minkowski instead of on the null cone. However, the two methods are easily seen to be equivalent to the one we describe}
  related to each other}
using the representative with $\{x^0 = 1\}\cap \{x_\alpha x^\alpha = 0 \}$ for the projective cone, where $\{x^\alpha\}$ ($\alpha,\beta = 0, \cdots, n+1$)
are Minkowskian coordinates of $\mathbb{M}^{1,n+1}$, $N$ is  the point with coordinates  $\{ x^0 = -x^1 = 1,x^{A+1} = 0 \}$ ($A,B = 1, \cdots, n$) and
the image plane for the stereographic projection is 
$\{x^0 = x^1 = 1\}$. The result is a representation of $\mathrm{CKill}(\mathbb{E}^n)$ where the vector vector fields are expressed in Cartesian coordinates $\{ y^A \}$  induced from the
Minkowskian coordinates by means of $\{ x^0 = x^1 =1, x^{A+1} = y^A\}$.

\begin{theorem}\label{theoisomskwconf}[\cite{Kdslike}]
  Let $\mathbb{M}^{1,n+1}$ endowed with Minkowskian coordinates $\lrbrace{x^\alpha}$ and consider any element $F \in \skwend{\mathbb{M}^{1,n+1}}$ written in the basis $\{ \partial_{x^{\alpha} } \}$ in the form
\begin{equation}\label{skwmatrix}
 F = \begin{pmatrix}
  0 &  -\nu & -a^t + b^t/2 \\
  -\nu & 0 & - a^t - b^t/2 \\
  -a + b/2 & a +b/2 & -\pmb{\omega}
  \end{pmatrix},
\end{equation}
where $a , b \in \mathbb{R}^n$ are column vectors, $^t$ stands for the transpose and $\pmb{\omega}$ is a skew-symmetric $n \times n$ matrix ($\pmb{\omega} = -\pmb{\omega}^t$). Then, in the Cartesian coordinates $\{ y^A\}$ of $\mathbb{E}^n $ defined by the embedding
$i: \mathbb{E}^n  \hookrightarrow \mathbb{M}^{1,n+1}$, $i(\mathbb{E}^n) = \{x^0 = x^1 = 1, x^{A+1} = y^A\}$, the CKVFs of $\mathbb{E}^n$ are
\begin{equation}\label{CKVFgeneral}
\xi_{F} = \lr{b^A + \nu y^A + (a_B y^B)y^A - \frac{1}{2}(y_B y^B)a^A - {\omega^A}_B y^B} \partial_{y^A}.
\end{equation}
Moreover,  $\xi_{Ad_\Lambda(F)} = \phi_{\Lambda \star}(\xi_F)$ for every $\Lambda \in \Lor^+(1,n+1)$ and $\xi$ is a Lie algebra antihomomorphism, i.e. $[\xi_F,\xi_G] =-\xi_{[F,G]}$.
\end{theorem}

\begin{remark}
 For later use, we write explicitly the parameters of the vector field $\nu, a^A, b^A, {\omega^A}_B$ in terms of the entries ${F^\alpha}_\beta$ of the endomorphism $F$:
\begin{equation}\label{vecendocomponents}
 \begin{aligned}
  \nu & = - {F^0}_1 ,
  \qquad
  & a_A & = -\frac{1}{2} \lr{{F^0}_{A+1} + {F^1}_{A+1}}, \\
  b_A & = \frac{1}{2} \lr{{F^0}_{A+1} - {F^1}_{A+1}}, \qquad
  & {\omega^A}_{B} & = - {F^{A+1}}_{B+1}.
 \end{aligned}
\end{equation}
where capital Latin indices are lowered with the
Kronnecker $\delta_{AB}$. Unless otherwise stated, $\xi$ without subindex
refers to the map $\xi$ given in \eqref{ximap} while $\xi_F$ refers to the CKVF which is image under $\xi$ of the skew-symmetric endomorphism $F$. 
\end{remark}

The freedom of choosing a representative for $\mathbb{S}^n$ (as well as the point $N$ and the projection stereographic plane) can be also seen in a more ``passive'' picture. Consider two different sets of Minkowskian coordinates $\{ x^\alpha\}$ and $\{ x'^\alpha\}$ related by a $\Lor^+(1,n+1)$ transformation $\Lambda$, $x'^\alpha = {\Lambda^\alpha}_\beta x^\beta$. Using Theorem \ref{theoisomskwconf}, we obtain two different embeddings $i,i':\mathbb{E}^n \hookrightarrow \mathbb{M}^{1,n+1}$ associated to $\{x^\alpha \}$ and $\{x'^\alpha \}$ respectively, for which $i(\mathbb{E}^n) = \{x^0 = x^1 = 1, x^{A+1} =: y^A\}$ and $i'(\mathbb{E}^n) = \{x'^0 = x'^1 = 1, x'^{A+1} =: y'^A\}$, as well as two associated maps $\xi, \xi'$. Let $F \in \skwend{\mathbb{M}^{1,n+1}}$, defined by \eqref{skwmatrix} with parameters $\{\nu, a^A, b^A, {\omega^A}_B\}$ and $\{\nu', a'^A, b'^A, {\omega'^A}_B\}$ in the bases $\{\partial_{x^\alpha} \}$ and $\{\partial_{x'^\alpha} \}$ respectively. Then, $F$ can be associated to two vector fields
\begin{align}
\xi_{F} & = \big(b^A + \nu y^A + (a_B y^B)y^A - \frac{1}{2}(y_B y^B)a^A - {\omega^A}_B y^B \big) \partial_{y^A}, \\
\xi'_{F} & =\big({b'}^A + \nu' y'^A + (a'_B {y'}^B)y'^A - \frac{1}{2}(y'_B {y'}^B){a'}^A - {{\omega'}^A}_B {y'}^B \big) \partial_{y'^A},
\end{align}
which are equal in the following sense.  If we transform the representative $\mathbb{S'}^n = \{ x'^0 = 1\}\cap \{ x'_\alpha x'^\alpha = 0\}$ with $\Lambda$, we obtain a new representative of the projective cone which in coordinates $x^\alpha$ is precisely $\mathbb{S}^n = \{ x^0 = 1\}\cap \{ x_\alpha x^\alpha = 0\}$. Abusing the notation, the map $\wphi_{\Lambda} := St_{N} \circ \Lambda \circ St_{N'}^{-1}$ is such that $\wphi_{\Lambda \star} (\xi'_F) = \xi_F$. Then, considering $i(\mathbb{E}^n)$ and $i'(\mathbb{E}^n)$ as respresentations of the  same space in two different global charts $( y^A, \mathbb{R}^n)$ and $(y'^A, \mathbb{R}^n)$, $\wphi_{\Lambda}$ can be seen as a change of coordinates  $y^A = (\wphi_\Lambda(y'))^A$, with the property that the Euclidean metric
in coordinates $\{ y'^A \}$ transforms as
\begin{equation}
 g_E = \delta_{AB} \dif y'^A \dif y'^B = \Omega^2(y) \delta_{AB} \dif y^A \dif y^B 
\end{equation}
for a smooth positive function $\Omega$. In other words, changing to different Minkowskian coordinates in $\mathbb{M}^{1,n+1}$ induces a change of coordinates in $\mathbb{E}^n$ in such a way that the form \eqref{CKVFgeneral} of the map $\xi$ is preserved. Notice that a similar result holds if we change the point w.r.t. which we take the stereographic projection, because any two $N, N' \in \mathbb{S}^n$ must be related by a $SO(n) \subset O^+(1,n+1)$ transformation.

Therefore, for the rest of this section, we will often adapt our choice
of Minkowskian coordinates
$\lrbrace{x^\alpha}$ of $\mathbb{M}^{1,n+1}$ to simplify the problem at hand. With this choice, it comes a corresponding set of cartesian coordinates
$\lrbrace{y^A}$ of $\mathbb{E}^n$ such that $\xi_F$ is given by equation \eqref{CKVFgeneral} and the Euclidean metric is $g_E = \Omega(y)^2 \delta_{AB} \dif y^A \dif y^B$. Which coordinates are adequate obviously depends on the problem. For example,
from the block form \eqref{decompFeven} and \eqref{decompFodd} of skew-symmetric endomorphisms, 
consider each of the blocks $\restr{F}{\mathbb{M}^{1,3}}$ $\restr{F}{\mathbb{M}^{1,2}}$ as endomorphisms of $\mathbb{M}^{1,n+1}$, extended as the zero map  in $(\mathbb{M}^{1,3})^\perp$ and $(\mathbb{M}^{1,2})^\perp$ respectively, and similarly for each $\restr{F}{\Pi_i}$.
If we denote by $\xi_{\restr{F}{\mathbb{M}^{1,3}}},~\xi_{\restr{F}{\mathbb{M}^{1,2}}}$ and $\xi_{\restr{F}{\Pi_i}}$ the corresponding images by $\xi$, one readily gets following decomposition:
 \begin{equation}\label{precanCKVF}
  \xi_F = \xi_{\restr{F}{\mathbb{M}^{1,3}}} + \sum\limits_{i=1}^p \xi_{\restr{F}{\Pi_i}}\quad\mbox{($n$ even)},\quad \quad  \xi_F = \xi_{\restr{F}{\mathbb{M}^{1,2}}} + \sum\limits_{i=1}^p \xi_{\restr{F}{\Pi_i}}\quad\mbox{($n$ odd)},
 \end{equation}
 where in terms of $n$, $p$ is given by
\begin{align}
  p = \left[ \frac{n+1}{2}\right] - 1 \label{pn}
\end{align}
(because the dimension of the Minkowski space where $F$ is defined is $d=n+2$, cf. Theorem \ref{theoisomskwconf}). 
 The explicit form of each of the terms in \eqref{precanCKVF} is direct from  \eqref{vecendocomponents}. Namely, the terms $\xi_{\restr{F}{\mathbb{M}^{1,3}}}$ and $\xi_{\restr{F}{\mathbb{M}^{1,2}}}$ are given by \eqref{CKVFgeneral} with vanishing parameters $a^A,b^A, {\omega^A}_B$ for $A,B\geq 3$ and $A,B\geq 2$ respectively, and each $\xi_{\restr{F}{\Pi_i}}$ is proportional to a vector field of the form 
 \begin{equation}\label{CAKVFgeneral}
  \eta:= y^{A_0} \partial_{y^{B_0}} - y^{B_0}\partial_{y^{A_0}}
 \end{equation}
 with $A_0,B_0 \in \lrbrace{1,\cdots, n}$ such that $A_0 \neq B_0$. More specifically, $\xi_{\restr{F}{\Pi_i}} = \mu_i \eta_i$, where $\eta_i$ is given by equation \eqref{CAKVFgeneral} with $B_0 = A_0 +1$ and $A_0 = 2i$ if $n$ even while  $A_0 = 2i+1$ if $n$ odd. 
 Vector fields of the form \eqref{CAKVFgeneral} will play an important role in the following analysis. They have the form of axial Killing vector fields, although in general they are CKVFs because of the conformal factor in $g_E = \Omega(y)^2 \delta_{AB} \dif y^A \dif y^B$. From the previous discussion, it follows that there exists a conformal transformation $\wphi_\Lambda \in \conf{\mathbb{E}^n}$ such that $g'_E := \wphi^\star_\Lambda(g_E) = \delta_{AB} \dif y^A \dif y^B$. Then by the properties of the Lie derivative it is immediate
 \begin{equation}
  0 = \mathcal{L}_\eta \wphi_\Lambda^\star(g_E) = \mathcal{L}_{\wphi_{\Lambda\star} \eta} g_E.
 \end{equation}
In other words, $\eta$ is an axial Killing vector of $g'_E$ and $\wphi_{\Lambda\star}\eta$ is an axial Killing vector of $g_E$. Thus, we define:
 \begin{definition}
  A CKVF of an Euclidean metric $g_E$, $\eta$,  is said to be a {\bf conformally axial} Killing vector field (CAKVF) if and only if the exist a $\wphi_{\Lambda} \in \conf{\mathbb{E}^n}$ such that $\wphi_{\Lambda\star}(\eta)$ is an axial Killing vector field of $g_E$. Equivalently, $\eta$ is a CAKVF if and only if it is an axial Killing vector field of $\wphi^\star_\Lambda(g_E)$.
 \end{definition}
 
\begin{remark}\label{remarkGsplkunit}
 Using Theorem \ref{theoisomskwconf}, it is immediate to verify that a CKVF is a CAKVF if and only if it is the image under $\xi$ of a simple unit spacelike endomorphism $G$.
\end{remark}

Notice that the terms in \eqref{precanCKVF} form a commutative subset of $\mathrm{CKill}\lr{\mathbb{E}^n}$. This is an immediate consequence
of the fact that $\xi$ is a Lie algebra antihomomorphism (c.f. Theorem \ref{theoisomskwconf}) and the blocks $\restr{F}{\mathbb{M}^{1,2}}$ (resp. $\restr{F}{\mathbb{M}^{1,3}}$) and $\restr{F}{\Pi_i}$ are pairwise commuting. In addition, a straightforward calculation shows that they form an orthogonal set
 \begin{equation}
 g_E(\widetilde \xi, \eta_i) = 0, \quad \quad  g_E(\eta_i,\eta_j) = 0\quad\quad(i\neq j)
 \end{equation}
 where $\widetilde \xi := \xi_{\restr{F}{\mathbb{M}^{1,3}}}$  for $n$ even and $\widetilde \xi := \xi_{\restr{F}{\mathbb{M}^{1,2}}}$ for $n$ odd. 
 In fact, as we show next, orthogonality of two CKVFs implies commutativity provided one of them is a CAKVF.  If both are CAKVF, then orthogonality turns out to be equivalent to commutativity.
 \begin{lemma}\label{lemmaortho}
  Let $\eta,\eta'$ be non-proportional CAKVFs and $\xi_F$ a CKVF. Then  $[\eta,\eta']=0$ if and only if there exist cartesian coordinates such that $\eta = y^{n-2}\partial_{y^{n-3}}- y^{n-3} \partial_{y^{n-2}}$ and $\eta' = y^{n-1} \partial_{y^n} - y^n \partial_{y^{n-1}}$. Equivalently  
  $[\eta, \eta']=0$ if and only if $g_E(\eta,\eta') = 0$. In addition, $[\xi_F, \eta]=0$ if $g_E(\xi_F,\eta) = 0$. 
 \end{lemma}
 \begin{proof}
   Let $G,G' \in \skwend{\mathbb{M}^{1,n+1}}$ be such that $\xi(G) = \eta,~\xi(G') = \eta'$. Since $G$ and $G'$ are simple, spacelike and unit (cf. Remark \ref{remarkGsplkunit}), we can write $G = e \otimes v_\flat - v \otimes e_\flat$ and $G' = e' \otimes v'_\flat - v' \otimes e'_\flat $ for spacelike, unit vectors $\{e,e',v,v'\}$, such that $0 = \lrprod{e,v} = \lrprod{e',v'}$. By Corollary \ref{commutingrank1_cor}, it follows that $[G,G'] = 0$ if and only if $\{e,e',v,v'\}$ are  mutually orthogonal.  Let us take cartesian coordinates of $\mathbb{M}^{1,n+1}$ such that $e = \partial_{x^{n-2}}, v = \partial_{x^{n-1}},e' = \partial_{x^{n}}, v' = \partial_{x^{n+1}}$.
Then, in the associated coordinates $\lrbrace{y^A}$ of $\mathbb{E}^n$ it follows $\eta = y^{n-2}\partial_{y^{n-3}}- y^{n-3} \partial_{y^{n-2}}$ and $\eta' = y^{n-1} \partial_{y^n} - y^n \partial_{y^{n-1}}$. This proves the first part of the lemma. From this result, it is trivial that $[\eta, \eta'] = 0$ implies $g_E(\eta, \eta') = 0$.

To prove that $g_E(\eta,\xi_F)=0$ implies $[\eta,\xi_F] = 0$ (which in
particular establishes the converse $g_E(\eta, \eta') = 0 \Longrightarrow
[\eta, \eta'] = 0$ for CAKVFs), let us take coordinates $\lrbrace{y^A}$ such that $\eta = y^{n-1} \partial_{y^{n}} - y^{n} \partial_{y^{n-1}}$. Then, 
writing $\xi_F$ as a general CKVF \eqref{CKVFgeneral}, we obtain by direct calculation:
\begin{equation}\label{eqgEetaxi}
 g_E(\eta, \xi_F) = \Omega^2 \lr{y^n b^{n-1} - y^{n-1} b^n -\frac{y_B y^B}{2}(a^n y^{n-1} - a^{n-1} y^n) + {\omega^{n-1}}_B y^B y^n - {\omega^{n}}_B y^B y^{n-1}}=0.
\end{equation}
Therefore $a^n, a^{n-1},b^n,b^{n-1}, {\omega^n}_B, {\omega^{n-1}}_B$ must vanish.
This implies that the associated endomorphisms $G$ and $F$ to $\eta$ and $\xi_F$ adopt a block structure from which it easily follows that $[G,F] = 0$ and hence
$[\eta, \xi_F] = 0$.
 \end{proof}

 \begin{definition}\label{defCAKVF}
   Let $\xi_F \in \mathrm{CKill}\lr{ \mathbb{E}^n}$. Then  a {\bf decomposed form}  of $\xi_F$ is $\xi_F = \widetilde \xi + \sum_{i=1}^p \mu_i \eta_i$ for an orthogonal subset $\{\widetilde \xi, \eta_i\}$, where $\eta_i$ are CAKVFs, $\mu_i \in \mathbb{R}$ for $i=1, \cdots, p$. A set of cartesian coordinates $\lrbrace{y^A}$ such that $\eta_i = y^{A_i} \partial_{y^{A_i+1}} - y^{A_i+1} \partial_{y^{A_i}}$, for  $A_i = 2 i$ for $n$ odd and $A_i = 2i+1$ for $n$ even, is called a set of {\bf decomposed} coordinates. 
 \end{definition}

 \begin{remark}
   \label{tildexi}
   Observe that the $\widetilde{\xi}$ is a CKVF. By Lemma \ref{lemmaortho} and its proof, the parameters $\{\nu, a, b, \boldsymbol \omega\}$
   defining $\widetilde{\xi}$ in a set of decomposed coordinates must all vanish
   except possibly $\{ \nu, a^1, a^2, b^1, b^2, \omega^1{}_{2} = - \omega^{2}_1\}$
   when $n$ is even or $\{ \nu, a^1, b^1\}$ when $n$ is odd.
   This means that there is a skew-symmetric endomorphism  $\widetilde{F}$ 
   with restricts
   to  $\mathbb{M}^{1,3} \subset
   \mathbb{M}^{1,n}$  ($n$  even) or
   $\mathbb{M}^{1,2} \subset
   \mathbb{M}^{1,n}$  ($n$ odd) and vanishes identically on their
   respective orthogonal complements such that  $\widetilde{\xi}
   = \xi_{\widetilde{F}}$. We will exploit this fact in an essential way below.
 \end{remark}

 With the definition of decomposed form of CKVFs, we can reformulate  Theorem \ref{theoremclasif} in terms of CKVFs. 
 
 \begin{proposition}\label{proprecan}
   Let $\xi_F \in \mathrm{CKill}\lr{\mathbb{E}^n}$. Then there exist an orthogonal set $\{ \eta_i\}_{i=1}^p$ of CAKVFs such that $[\xi_F,\eta_i] = 0$. For every such a set $\{ \eta_j\}_{j=1}^p$ and $i \in \{ 1, \cdots, p\}$
   there exist $\mu_i \in \mathbb{R}$ such that  $g_E(\eta_i, \eta_i) \mu_i = g_E(\xi_F,\eta_i)$. In addition,with the definition $\widetilde \xi := \xi_F - \sum  \mu_i \eta_i$ the expression $\xi_F = \widetilde \xi + \sum \mu_i \eta_i$ provides a decomposed form of $\xi_F$.
   \end{proposition}
 \begin{proof}
   The existence of $p$ commuting CAKVFs is a direct consequence of decompositions \eqref{decompFeven} and \eqref{decompFodd} of the associated skew-symmetric endomorphism $F$, for $n$ even and odd respectively. Indeed, for each such decomposition of $F$, it follows a set of $p$ CAKVFs commuting with $\xi_F$.
   {Let us denote $\{ \eta_i \}$ any such set}. Each    $\eta_i$ is associated to a simple, spacelike unit endomorphism $G_i$ that commutes with $F$.
   By Lemma \ref{commutingrank1}, {$G_i$} defines a spacelike eigenplane $\Pi_i$ of $F$. The orthogonality of any two such eigenplanes $\Pi_i, \Pi_j$, $ i \neq j$   is a consequence of Corollary \ref{commutingrank1_cor} because $[G_i,G_j] = 0$. In other words, given a set of $p$ CAKVFs commuting with $\xi_F$, we have a block form of $F$, thus, defining $\widetilde \xi := \xi_F - \sum  \mu_i \eta_i$, it is immediate that $\xi_F = \widetilde \xi + \sum \mu_i \eta_i$ is a decomposed form with 
  $g_E(\eta_i, \eta_i) \mu_i = g_E(\xi_F,\eta_i)$.
 \end{proof}

 The next  step now is to give a definition of canonical form for CKVFs, which we induce  from the canonical form of the associated skew-symmetric endomorphism.  
 

 \begin{definition}
 A CKVF $\xi_F$ is in {\bf canonical form} if it is the image of a skew-symmetric endomorphism $F$ in canonical form, i.e. $\xi_F = \widetilde \xi + \sum \mu_i \eta_i$ such that $\widetilde \xi$ is given, in a cartesian set of coordinates $\{y^A\}$ denoted {\bf canonical coordinates}, by the parameters $a^1 = 1,~b^1 = \sigma/2,~a^2 = 0,~b^2 = \tau/2$ if $n$ even and $a^1 = 1,~b^1 = \sigma/2$ if $n$ odd (the non-specified parameters all vanish) and $\eta_i$ are CAKVFs  $\eta_i = y^{A_i} \partial_{y^{A_i+1}} - y^{A_i+1} \partial_{y^{A_i}}$, for  $A_i = 2 i$ for $n$ odd and $A_i = 2i+1$ for $n$ even, and where $\aa, \bb, \mu_i$ are given by Definition \ref{defgammamu}.
\end{definition}

Given a CKVF $\xi_F$, the existence of a canonical form and canonical coordinates are guaranteed by Theorem \ref{theocanonicalF}. By Theorem \ref{theoisomskwconf}, the conformal class $[\xi_F]_{Conf}$ of a CKVF $\xi_F$ is equivalent to the equivalence class $[F]_{\Lor^+}$ of $F$ under the adjoint action of $\Lor^+(1,n+1)$, and this is determined by the canonical form of $F$ (c.f. Theorem \ref{theoLorendo}). This argument together with the results of Section \ref{seclorclass} yield the following statement.

\begin{theorem}\label{theoconfclass}
 Let $\xi_F \in \mathrm{CKill}\lr{\mathbb{S}^{n}}$ be in canonical form. Then its conformal class $[\xi_F]_{Conf}$ is determined by $(\aa, \bb, \mu_i^2)$ if $n$ even and $(\cc, \mu_i^2)$ if $n$ odd. Moreover, the structure of $\mathrm{CKill}(\mathbb{S}^n)/\conf{\mathbb{S}^n}$ corresponds with that of Remark \ref{remarklimits}. 
\end{theorem}

Given a canonical form  $\xi_F = \widetilde \xi + \sum \mu_i \eta_i$ the set of vectors $\{ \widetilde \xi, \eta_i \}$ are pairwise commuting and linearly independent. As we will next prove, in the case of odd dimension this set is a maximal (linearly independent) pairwise commuting set of CKVFs commuting with $\xi$ (i.e. it is not contained in a larger set of linearly independent vectors commuting with $\xi$). In the case of even dimension it is not maximal.
By Remark \ref{tildexi}, $\widetilde{\xi} =
\widetilde \xi (\nu, a^1, a^2, b^1, b^2, \omega)$, where the right-hand side
denotes a CKVF of the form \eqref{CAKVFgeneral} whose parameters vanish, except possibly $\{ \nu, a^1, a^2, b^1, b^2, \omega := {\omega^1}_2 \}$. As mentioned
in the Remar, the corresponding skew-symmetric endomorphism $\widetilde F$ satisfying $\xi_{\widetilde{F}}  = \widetilde \xi$ 
can be understood as an element $\widetilde F \in \skwend{\mathbb{M}^{1,3}}$, with $\mathbb{M}^{1,3} = \spn{e_0,e_1,e_2,e_3}$, that is identically zero in $\lr{\mathbb{M}^{1,3}}^\perp$. 
Fix the orientation in $M^{1,3}$ so that the basis $\{e_0, e_1, e_2, e_3\}$ is positively oriented. The Hodge star maps two-forms into two-forms. This defines a natural map 
\begin{equation}
  \begin{array}{rcl}
   \star: \skwend{\mathbb{M}^{1,3}}  & \longrightarrow & \skwend{\mathbb{M}^{1,3}},\\
 \widetilde F & \longmapsto & \widetilde F^\star.
  \end{array}
 \end{equation}
From standard properties of two-forms, (see also \cite{marspeon20}) it follows that 
 ${\widetilde F}^{\star}$ commutes with $\widetilde{F}$.
We may extend ${\widetilde F}^{\star}$ to an endomorphism on $\mathbb{M}^{1,n+1}$ that vanishes identically on $(\mathbb{M}^{1,3})^{\perp}$, just as $\widetilde{F}$. It is clear that the commutation property is preserved by this extension. The image of $\widetilde F^\star$ under $\xi$ is the vector field
\begin{equation}
 \widetilde \xi^\star := \lr{\widetilde \xi (\nu, a^1, a^2, b^1, b^2, \omega)}^\star =  \widetilde \xi (-\omega, a^2, -a^1, -b^2, b^1, \nu),
\end{equation}
which by construction commutes with $\widetilde \xi$. In the case that $\widetilde \xi$ is the first element in a  decomposed form $\xi_F = \widetilde \xi + \sum \mu_i \eta_i$, 
it is immediately true that $\widetilde \xi^\star$ also commutes with all of the CAKVFs $\eta_i$. Hence, $\{\widetilde \xi, \widetilde \xi^\star, \eta_i \}$ is a pairwise commuting set, all of them commuting with $\xi$.
This set can be proven to be maximal:

 \begin{proposition}
   Let $\xi_F = \widetilde \xi + \sum \mu_i \eta_i$ be a CKVF in canonical form. If $n$ is odd, $\{ \widetilde \xi, \eta_i \}$ is a maximal linearly independent pairwise commuting set of elements that commute with $\xi_F$.   If $n$ is even,  $\{ \widetilde \xi, \widetilde \xi^\star, \eta_i \}$ is a maximal linearly independent pairwise commuting set of elements that commute with $\xi_F$.
 \end{proposition}
\begin{proof}
  Suppose that there is an additional CKVF $\xi'$ commuting
  with each element in $\{ \widetilde \xi, \eta_i \}$ if $n$ odd or $\{ \widetilde \xi, \widetilde \xi^\star, \eta_i \}$ if $n$ even (in either case
  $\xi'$ clearly commutes with $\xi_F$ also). Since it commutes with each $\eta_i$, by Proposition \ref{proprecan}, it admits a decomposed form $\xi' = \widetilde \xi' + \sum_{i=1}^{p} \mu'_i \eta_i$, where $\widetilde \xi'$ is a CKVF orthogonal to each $\eta_i$ and which must verify $[\widetilde \xi', \widetilde \xi] = 0$. Equivalently,
  their associated endomorphisms satisfy
  $\widetilde F' \in \mathcal{C}(\widetilde F)$, where
  $\mathcal{C}(\widetilde F)$ denotes the centralizer of $F$, i.e. the set
  of all skew-symmetric endomorphisms that commute with $F$. 
{From the results in \cite{marspeon20}, $\mathcal{C}({\widetilde F}\mid_{\mathbb{M}^{1,2}}) = \spn{{\widetilde F}\mid_{\mathbb{M}^{1,2}}}$ when $n$ is odd and $\mathcal{C}({\widetilde F}\mid_{\mathbb{M}^{1,3}}) = \spn{\widetilde F\mid_{\mathbb{M}^{1,3}},{\widetilde F^\star}\mid_{\mathbb{M}^{1,3}}}$ when $n$ is even. Here, $\widetilde F^\star$ is the skew-symmetric endomorphim associated with $\widetilde \xi^\star$ and we restrict to $\mathbb{M}^{1,3}$ because the action of the endomorphisms is identically zero in  $(\mathbb{M}^{1,3})^\perp$.} Thus $\widetilde \xi' = a \widetilde \xi,~a \in \mathbb{R},$ if $n$ odd and  $\widetilde \xi' = b \widetilde \xi + c\widetilde \xi^\star,~b,c \in \mathbb{R}$ if $n$ even.


\end{proof}

 \section{Adapted coordinates}\label{secadapted}

 In the previous section we obtained a canonical form for each CKVF of euclidean space based on the canonical form of skew-symmetric endomorphisms in Section
 \ref{seccanonform}. As an application, we consider in this section the problem
 of adapting coordinates in $\mathbb{E}^n$ to a given CKVF $\xi_F$. The use of the canonical form  will allow us to solve the problem for every possible $\xi_F$ essentially in one go. Actually it will suffice to consider the case
 of even dimension $n$ and assume that at least one of the parameters $\aa, \bb$ in the canonical form of $\xi_F$ is non-zero. The case where both
 $\aa$ and $\bb$ vanish will be obtained as a limit (and we will check that
 this limit does solve the required equations). The case of odd dimension $n$ wil be obtained from the even dimensional one by expliting the  property that $\mathbb{E}^{2m +1}$
 can be viewed as a hyperplane of
 $\mathbb{E}^{2m+2}$ in such a way that the given CKVF $\xi_F$ in
 $\mathbb{E}^{2m +1}$ extends conveniently to $\mathbb{E}^{2m+2}$. Restricting
 the adapted coordinates already obtained in the even dimensional case to the appropriate  hyperplane we will be able to infer the odd dimensional case. As we will justify the process of adapting coordinates is different for $n=2$
 and $n \geq 4$ even. The case $n =2$ has been treated in detail
 in \cite{marspeon20}, so it will suffice to consider even $n \geq 4$ here.


 Consider $\mathbb{E}^n$ endowed with a CKVF $\xi_F$. First of all we adapt the
 Cartesian coordinates of $\mathbb{E}^n$ so that $\xi_F$ takes its canonical
 form and we fix the metric of $\mathbb{E}^n$ to take the explictly flat form in these coordinates.  We further assume (for the moment)
 that  $n$ is even. For notational reasons it is convenient to rename
 the canonical coordinates\footnote{The fact that we tag the coordinates $\{z_1,z_2,x_i,y_i \}$ with lower indices has no particular meaning. It is simply
   to avoid a notational clash of upper indices and powers that will appear later}  as $z_1 := y^1,~z_2 := y^2 $ and $x_i :=y^{2i +1},~y_i := y^{2i+2}$ for $i = 1, \cdots, p$, where in the even case  case $p = n/2 -1$ (see \eqref{pn}). 
 By Proposition \ref{proprecan}, $\xi_F$ can be decomposed as a sum of CKVFs
 $\widetilde \xi$ and $\eta_i$ and, additionally one can construct canonically yet another CKVF $\widetilde \xi^\star$. This collection of CKVFs  defines a maximal commutative set. Moreover, $\{ \eta_i \}$ are all mutually orthogonal and perpendicular to $\widetilde \xi$ and $\widetilde \xi^{\star}$. It is therefore most natural to try and find coordinates adapted simultaneously to
 the whole family $\{ \widetilde \xi, \widetilde \xi^{\star}, \eta_i\}$. This will lead a (collection of) coordinate systems where the components of $\xi_F$
 are simply constants. From here one can immediately find coordinates that
 rectify $\xi_F$, if  necessary. It is important to emphasize that selecting
 the whole set $\{ \widetilde \xi, \widetilde \xi^{\star}, \eta_i\}$ to adapt
 coordinates provides enough restrictions so that the coordinate change(s) can be fully determined. Imposing the (much weaker) condition that the system of coordinates rectifies only $\xi_F$ is just a too poor condition to solve the problem.
 This is an interesting example where the structure of the canonical
 decomposition of $\xi_F$ (or of $F$) is exploited in full.

By Theorem \ref{theoisomskwconf}, the explicit form of $\{ \widetilde \xi, \widetilde \xi^{\star}, \eta_i\}$
 in the canonical  coordinates is 
 \begin{align}
  \widetilde \xi &= \lr{\frac{\aa}{2} + \frac{1}{2}\lr{z_1^2 - z_2^2 - \sum\limits_{i=1}^{p}(x_i^2 + y_i^2)}} \partial_{z_1} + \lr{\frac{\bb}{2} + z_1 z_2} \partial_{z_2} + z_1 \sum\limits_{i=1}^{p}\lr{x_i \partial_{x_i} + y_i \partial_{y_i}} \label{CKVFxi}\\
  \widetilde \xi^\star &= -\lr{\frac{\bb}{2} + z_1 z_2} \partial_{z_1} +  \lr{\frac{\aa}{2} - \frac{1}{2}\lr{z_2^2 - z_1^2 - \sum\limits_{i=1}^{p}(x_i^2 + y_i^2)}} \partial_{z_2} -  z_2 \sum\limits_{i=1}^{p}\lr{x_i \partial_{x_i} + y_i \partial_{y_i}} \label{CKVFxistar}\\
  \eta_i & = x_i \partial_{y_i} - y_i \partial_{x_i}.\label{CAKVF}
 \end{align}
We are seeking coordinates $\{ t_1, t_2, \phi_i, v_i \}$ adapted to  these vector fields, i.e. such that $\partial_{t_1} = \widetilde \xi,~\partial_{t_2} = \widetilde \xi^\star,~\partial_{\phi_i} = \eta_i$.  
It is clear that if $\{ t_1, t_2, \phi_i, v_i \}$ is an adapted coordinate system, so it is $\{ t_1 - t_{0,1}(v), t_2 - t_{0,2}(v),\phi_i - \phi_{0,i}(v), v_i \}$ for arbitrary functions $t_{0,1}(v)$, $t_{0,2}(v)$ and $\phi_{0,i}(v)$, where  $v = (v_1,\cdots,v_p)$. This will be used to simplify the process of integration. This freedom, may be restored at the end if so desired.
Hence 
\begin{align}
    \frac{\partial z_1}{ \partial t_1}&=\frac{\aa}{2} + \frac{1}{2}\lr{z_1^2 - z_2^2 - \sum\limits_{i=1}^{p}(x_i^2 + y_i^2)}, & \frac{\partial z_2}{\partial t_1}&=\frac{\bb}{2} + z_1 z_2, & \frac{\partial x_i}{\partial t_1}&=z_1 x_i, & \frac{\partial y_i}{\partial t_1}&=z_1 y_i,  \label{PDEsxi}\\
      \frac{\partial z_2}{\partial t_2}&=\frac{\aa}{2} - \frac{1}{2}\lr{z_2^2 - z_1^2 - \sum\limits_{i=1}^{p}(x_i^2 + y_i^2)}, & \frac{\partial z_1}{\partial t_2}&=-\frac{\bb}{2} - z_1 z_2, & \frac{\partial x_i}{\partial t_2}  &=-z_2 x_i, & \frac{\partial y_i}{\partial t_2}&=-z_2 y_i,\label{PDEsxistar} \\
       \frac{\partial z_1}{\partial \phi_i}&=0 &  \frac{\partial z_2}{\partial \phi_i}&=0 & \frac{\partial x_i}{\partial \phi_i}&=-y_i  & \frac{\partial y_i}{\partial \phi_i}&=x_i \label{PDEsetai}
 \end{align}
 The additional $p$ coordinates $v_i$, will appear through functions of integration. It is clear that the structure of the equations is different for $n=2$, where
 there are no $\{x_i, y_i\}$, which implies that the process
 of integration follows a different route. The case $n=2$ has been treated in full detail in \cite{marspeon20}, where the the complex structure of $\mathbb{S}^2$ can be exploited to simplify the problem. Here we adress the problem for $n \geq 4$ which we assume from now on.

 We may start by integrating \eqref{PDEsetai}. The first pair 
 gives $ z_1 = z_1(t_1,t_2,v),~ z_2 = z_2(t_1,t_2,v)$, so that the second pair becomes a harmonic oscillator in  $x_i,y_i$, whose solution is
 \begin{equation}\label{xiyi1}
 \quad\quad x_i = \rho_i(t_1,t_2,v) \cos(\phi_i - \phi_{0,i}(t_1,t_2,v)), \quad\quad y_i = \rho_i(t_1,t_2,v) \sin(\phi_i - \phi_{,0,i}(t_1,t_2,v)),
 \end{equation}
 where $\rho_i$ and $\phi_{0,i}$ are arbitrary functions (depending only on the variables indicated) and $\rho_i$ is not identically zero.
 

Inserting \eqref{xiyi1} in any of the two right-most equations of \eqref{PDEsxi} and \eqref{PDEsxistar} and equating terms multiplying $\sin(\phi_i + \phi_{(0)i})$ and $\cos(\phi_i + \phi_{(0)i})$ yields:
\begin{equation*}
 z_1 = \frac{1}{\rho_i} \frac{\partial \rho_i}{\partial t_1},\quad \quad z_2 = -\frac{1}{\rho_i} \frac{\partial \rho_i}{\partial t_2},\quad \quad \frac{\partial\phi_{(0)i}}{\partial t_1} = 0,\quad \quad  \frac{\partial\phi_{(0)i}}{\partial t_2} = 0. 
\end{equation*}
Thus, $\phi_{0,i}$ is a function only of $v$, which may be absorbed on the coordinate $\phi_i$ as discussed above. The two first equations imply
\begin{equation*}
 \frac{1}{\rho_i} \frac{\partial \rho_i}{\partial t_1} = \frac{1}{\rho_j} \frac{\partial \rho_j}{\partial t_1}, \quad \frac{1}{\rho_i} \frac{\partial \rho_i}{\partial t_2} = \frac{1}{\rho_j} \frac{\partial \rho_j}{\partial t_2} \quad\Longleftrightarrow\quad \rho_i = \hat \alpha_i(v) \hat \rho(t_1,t_2,v),
\end{equation*}
 for arbitrary (non-zero) functions $\hat \alpha_i$ and $\hat \rho$.
Defining $\rho^2 := \sum \limits_{i=1}^{p} \rho_i^2 = \left (  \sum\limits_{i=1}^{p} \hat \alpha_i^2 \right ) \hat \rho^2$ we can write 
\begin{equation*}
 \rho_i = \hat \alpha_i \hat \rho = \frac{\hat \alpha_i\epsilon}{\sqrt{\sum_{j=1}^p\hat \alpha_j^2}} \rho = \alpha_i \rho,
\end{equation*}
where $\alpha_i :=  \hat \alpha_i\epsilon/\sqrt{\sum_{j=1}^p\hat \alpha_j^2}$, with $\epsilon^2 = 1$, form a set of arbitrary (non-zero) functions of $v$ such that $\sum\limits_{i=1}^{p} \alpha_i^2 = 1$. The function $\rho$ satisfies
\begin{equation}\label{z1z2rho}
z_1 = \frac{1}{\rho} \frac{\partial \rho}{\partial t_1},\quad \quad z_2 = -\frac{1}{\rho} \frac{\partial \rho}{\partial t_2}.
\end{equation}
Inserting \eqref{z1z2rho} in the two left-most equations in \eqref{PDEsxi} and \eqref{PDEsxistar}, with the change of variable $U = \rho^{-1}$, we obtain after some algebra the following covariant system of PDEs (indices $a, b = 1,2$ refer to $\{ t_1, t_2\}$)
\begin{equation}\label{PDEUt1t2}
\nabla_a \nabla_b U = U A_{ab}  + \frac{1}{2 U}(1 + \nabla_c U \nabla^c U)g_{ab}\quad\mbox{with}\quad 
A = \frac{1}{2}(-\aa \dif t_1^2 + \aa \dif t_2^2 + 2 \bb \dif t_1 \dif t_2),
\quad g = \dif t_1^2 + \dif t_2^2,
\end{equation}
and where $\nabla$ is the Levi-Civita covariant derivative of $g$.
\begin{lemma}
  Up to shifts $t_1 \rightarrow t_1 - t_{0,1}(v)$ and
  $t_1 \rightarrow t_1 - t_{0,1}(v)$, the general solution of \eqref{PDEUt1t2} with either $\aa$ or $\bb$  non-zero is given by
 \begin{equation}\label{solUgen}
 U  = \frac{\epsilon}{\mu_t^2 + \mu_s^2}\lr{  \beta \cosh(t_+)-\alpha \cos(t_-)} \quad \quad\mbox{with}\quad\quad  \beta = \sqrt{\alpha^2 + \mu_t^2 + \mu_s^2} 
\end{equation}
where $\alpha$ is a  function of integration (depending on $v$), $\epsilon^2 = 1$ and $t_+ := \mu_t t_1 + \mu_s t_2 $, $t_- := \mu_t t_2 - \mu_s t_1 $, with $\mu_s, \mu_t$ given by \eqref{eqmutmus}. The solution \eqref{solUgen} admits a limit $\aa = \bb = 0$ (i.e. $\mu_t = \mu_s = 0$) provided $\alpha >0$, which is
\begin{equation}\label{Ulim}
 \lim_{\mu_s \mu_t \rightarrow 0} U = \epsilon \frac{\alpha}{2}(t_1^2 + t_2^2) + \frac{\epsilon}{2 \alpha}.
\end{equation}
Up to shifts $t_1 \rightarrow t_1 - t_{0,1}(v)$ and
$t_2 \rightarrow t_2 - t_{0,2}(v)$, this function 
is the general solution of \eqref{PDEUt1t2} for $\aa = \bb = 0$.
\end{lemma}
\begin{proof}
  The coordinates $t_+, t_-$ defined in the lemma diagonalize $A$ and $g$
  simultaneously and yield
   \begin{equation}
  A= \frac{1}{2} (\dif t_+^2 - \dif t_-^2),\quad\quad g = \frac{1}{\mu_s^2 + \mu_t^2}(\dif t_+^2 + \dif t_-^2).
 \end{equation}
 From this and equation \eqref{PDEUt1t2} it follows that $\partial^2 U/\partial t_+ \partial t_- = 0$ or, equivalently, $U(t_+,t_-) = U_+(t_+) + U_-(t_-)$.
 Substracting the $\{ t_+, t_+\}$ and $\{ t_-, t_-\}$ components of \eqref{PDEUt1t2} one obtains
 \begin{equation}
  \frac{\dif^2 U_+}{\dif t_+^2} - \frac{\dif^2 U_-}{\dif t_-^2} = U = U_+ + U_- \quad\Longrightarrow\quad \frac{\dif^2 U_+}{\dif t_+^2} - U_ +  = \frac{\dif^2 U_-}{\dif t_-^2} + U_- = \hat a
 \end{equation}
 for an arbitrary separation function $\hat a(v)$. The general solution
 is clearly
\begin{equation}\label{eqUpUm}
 U_+ = -\hat a + a \cosh(t_+) + b \sinh(t_- )  \quad\quad U_- = \hat{a} + c \cos(t_-- \delta),
\end{equation}
where $a,b,c, \delta$ are also functions of $v$. Since $\hat{a}$ drops out in
 $U = U_+ + U_-$ we may set  $\hat a = 0$ {w.l.o.g.} Inserting \eqref{eqUpUm} in (any of) the diagonal terms of \eqref{PDEUt1t2} and one simply gets 
\begin{equation}\label{eqconstr}
 a^2 -b^2 = \frac{1}{\mu_s^2 + \mu_t^2} + c^2.
\end{equation}
Hence $|a|> |b|$ and we may use the freedom of translating $t_+$ by a function of $v$ to write $U_+ = a \cosh(t_+)$ (i.e. $b=0$). A similar translation in $t_{-}$ sets
$\delta=0$. Rescaling the functions $a,c$ as $a = (\mu_s^2+ \mu_t^2)^{-1} \beta$
and $c = - (\mu_s^2+ \mu_t^2)^{-1} \alpha$ we get
\begin{equation}\label{eqUpUm2}
 U = U_-+ U_- =  \frac{\beta}{\mu_s^2 + \mu_t^2} \cosh(t_+) - \frac{\alpha}{\mu_s^2 + \mu_t^2}\cos(t_-),\quad\quad \beta^2 = \mu_s^2 + \mu_t^2 + \alpha^2.
\end{equation}
 It is obvious that $\sign(U) = \sign(\beta)$. Thus taking $\beta$ as the positive root $\beta = \sqrt{\alpha^2 + \mu_s^2 + \mu_s^2}$ and adding a multiplicative sign $\epsilon$ in \eqref{eqUpUm2}, we obtain \eqref{solUgen}.  
 To evaluate the convergence as both $\aa, \bb$ tend to zero, or equivalently $\mu_s,\mu_t \rightarrow 0$, consider the series expansion
\begin{align*}
 \beta  \cosh(t_+) & = \lr{|\alpha|+ \frac{\mu_s^2 + \mu_t^2}{2|\alpha|} + o_{\mu_t,\mu_s}^{(4)}} \lr{1 + \frac{(\mu_s t_2 + \mu_t t_1)^2}{2} + o_{\mu_t,\mu_s}^{(4)}},\\
 \alpha \cos(t_-) & = \alpha - \alpha\frac{(\mu_t t_2 - \mu_s t_1)^2}{2} + o_{\mu_t,\mu_s}^{(4)},
\end{align*}
where $o_{\mu_t,\mu_s}^{(4)}$ denotes a sum of homogeneous polynomials in $\mu_t, \mu_s$ starting at order four, whose coeficients may depend on $t_1,t_2$ and $\alpha$. Then, the expansion of $U$ is 
\begin{equation}
 U = \frac{\epsilon}{\mu_s^2 + \mu_t^2}\lr{(|\alpha| - \alpha)(1 + \mu_s \mu_t t_1 t_2)+ \frac{|\alpha|\mu_s^2 + \alpha \mu_t^2}{2}t_2^2 + \frac{|\alpha|\mu_t^2 + \alpha \mu_s^2}{2}t_1^2+ \frac{\mu_s^2 + \mu_t^2}{2|\alpha|} + o_{\mu_t,\mu_s}^{(4)}}.
\end{equation} 
It is clear that $\lim_{\mu_s,\mu_t \rightarrow 0} o_{\mu_t,\mu_s}^{(4)}/(\mu_s^2 + \mu_t^2) = 0$ and the rest of the equation converges if and only if $\alpha>0$ in which case the limit is \eqref{Ulim}. 
An easy calculation shows that this limit is (up to shifts in $t_1, t_2$)
is the general solution of \eqref{PDEUt1t2} when $\aa, \bb = 0$.
\end{proof}

 Having the general general solution \eqref{solUgen} of \eqref{PDEUt1t2}  we can give the expression of the adapted coordinates
\begin{align}
 z_1 &=-\frac{1}{U}\frac{\partial U}{\partial t_1} =\left|\frac{1}{U}\right|\frac{\alpha  \mu_s \sin (t_-) - \beta  \mu_t \sinh (t_+)}{\mu_s^2 + \mu_t^2},
 \label{adaptz1}\\
 z_2 &= ~~\frac{1}{U}\frac{\partial U}{\partial t_2}  = \left|\frac{1}{U}\right| \frac{\alpha  \mu_t \sin (t_-)  +   \beta  \mu_s \sinh (t_+)}{\mu_s^2 + \mu_t^2},
 \label{adaptz2}\\
 x_i &= \frac{\alpha_i}{U} \cos(\phi_i),
\quad\quad
 y_i = \frac{\alpha_i}{U} \sin(\phi_i) \label{adaptxiyi},
\end{align}
where no sign of $\alpha$ is in principle assumed\footnote{The domain of definition of $\alpha$ will be later restricted under the condition that the adapted coordinates define a one to one map.}, except for the case $\mu_s = \mu_t  = 0$, where $U$ must be understood as the limit (with $\alpha>0$) \eqref{Ulim} and $z_1 = -U^{-1}\partial U/\partial t_1$, $z_2 = U^{-1}\partial U/\partial t_2$. This coincides with the limit of the RHS expressions \eqref{adaptz1}, \eqref{adaptz2}, which is\begin{equation} \label{adaptz1z2lim}
 z_1 =\frac{-2  \alpha^2 t_1}{1 + \alpha^2 (t_1^2 + t_2^2)},
 \quad\quad
 z_2 = \frac{2  \alpha^2 t_2}{1 + \alpha^2 (t_1^2 + t_2^2)} .
\end{equation}
From equations \eqref{adaptz1}, \eqref{adaptz2} and \eqref{adaptxiyi} it is obvious that the sign $\epsilon$ is not relevant in the definition of the adapted coordinates. This is because the two branches $\epsilon =1$ and $\epsilon = -1$ correspond with $U>0$ and $U<0$ respectively, which in terms of the adapted coordinates, is equivalent to a rotation of $\pi$ in the $\phi_i$ angles. Hence, w.l.o.g. we consider $\epsilon = 1$, i.e. $U>0$. Also notice that the dependence on the variables $v_i$  appears through the functions $\alpha_i$ and $\alpha$, with $\sum_{i=1}^{p}\alpha_i^2 = 1$. The set $\{\alpha_i, \alpha\}$ define $p$ independent arbitrary functions of the variables $v_i$, so it is natural to use
as coordinates  $\lrbrace{\alpha_i,\alpha}$ themselves, provided they are
restricted to satisfy $\sum_{i=1}^{p}\alpha_i^2 = 1$.

We now  calculate the region of $\mathbb{E}^n$ covered by the adapted coordinates. It is clear that in no case this region can include neither the  zeros of the vector fields  $\widetilde \xi$ and $\widetilde \xi^\star$ and $\eta_i$ nor the points where these $p+2$ vectors are linearly dependent. We therefore start by locating those points. Denoting the loci of the zeros of $\widetilde \xi$ and $\widetilde \xi^\star$ and $\eta_i$ by $\mathcal{Z}(\widetilde \xi)$, $\mathcal{Z}(\widetilde \xi)^\star$ and $\mathcal{Z}(\eta_i)$ respectively, a simple calculation gives
\begin{align}
  \mathcal{Z}(\widetilde \xi) &=
  \Big ( \{\bigcap_{j=1}^{p} \lrbrace{x_j = y_j = 0}\} \cap \lrbrace{z_1 = \pm \mu_t, z_2 = \mp \mu_s} \Big ) \cup \Big ( \{z_1 = 0\}\cap \{z_2^2 + \sum_{j=1}^p (x_j^2 + y_j^2) = \mu_s^2 - \mu_t^2 \} ~\mbox{if}~\mu_s \mu_t =0 \Big ), \label{zxi} \\
 \mathcal{Z}(\widetilde \xi^\star) &= \Big ( \{\bigcap_{j=1}^{p} \lrbrace{x_j = y_j = 0}\} \cap \lrbrace{z_1 = \pm \mu_t, z_2 = \mp \mu_s} \Big ) \cup \Big ( \{z_2 = 0\}\cap \{z_1^2 + \sum_{j=1}^p(x_j^2 + y_j^2) = \mu_t^2 - \mu_s^2 \} ~\mbox{if}~\mu_s \mu_t=0 \Big ), \label{zxitilde} \\
 \mathcal{Z}(\eta_i) &= \lrbrace{x_i = y_i = 0}. \nonumber 
\end{align}
These expressions are valid for every value of $\mu_s, \mu_t$ and imply that
in the case $\mu_s = \mu_t = 0$, $\mathcal{Z}(\widetilde \xi) = \mathcal{Z}(\widetilde \xi^\star) =\big\{ \bigcap_{j=1}^{p} \lrbrace{x_j = y_j = 0} \big\} \cap \lrbrace{z_1 = z_2 = 0}$, which is contained in each $\mathcal{Z}(\eta_i) = \lrbrace{x_i = y_i = 0}$.

On the other hand, since $\{\widetilde \xi, \eta_i \}$ is an orthogonal set of CKVFs (cf. Lemma \ref{lemmaortho}), they are pointwise linearly independent at all points where they do not vanish. Similarly, $\{ \widetilde \xi^\star, \eta_i \}$ is also an orthogonal set, so linear independence is guaranteed away from the zero set. Away from this set, the set of vectors $\{ \widetilde \xi,
\widetilde \xi^\star, \eta_i \}$ is linearly dependent only at points where
$\widetilde \xi$ and $\widetilde \xi^\star$ are proportional to each other
with a non-zero
proportionality factor, $\widetilde \xi = a \widetilde \xi^\star$,  $a \neq 0$.
One easily checks that, away from $\mathcal{Z}(\widetilde \xi)$ and  $\mathcal{Z}(\widetilde \xi^{\star})$,  the set of point where $\widetilde \xi - a \widetilde \xi^\star$ vanishes is empty except when $\mu_s \neq 0, \mu_t \neq 0$ and $a = \frac{\mu_t}{\mu_s}$.  It turns out to be useful to determine the set of points
where $\mu_s \widetilde \xi - \mu_t \widetilde \xi^\star=0$ when at least one
of $\{ \mu_s, \mu_t\}$ is non-zero. We call this set $\mathcal{Z}(\mu_s \widetilde \xi -\mu_t   \widetilde \xi^\star)$, and  a straightforward analysis gives
\begin{align}\label{xiproptoxistar}
  \mathcal{Z}(\mu_s \widetilde \xi -\mu_t   \widetilde \xi^\star ) & =
\left \{   \begin{array}{ll}
             \{\mu_s z_1 = -\mu_t z_2\}\cap \{ (\mu_s^2 + \mu_t^2)z_2^2 + \mu_s^2 \sum\limits_{i=1}^p (x_i^2 + y_i^2) = (\mu_s^2 + \mu_t^2)\mu_s^2 \} & 
\mbox{if } \mu_s \neq 0 \\
 \{\mu_s z_1 = -\mu_t z_2\}\cap \{ (\mu_s^2 + \mu_t^2)z_1^2 + \mu_t^2 \sum\limits_{i=1}^p (x_i^2 + y_i^2) = (\mu_s^2 + \mu_t^2)\mu_t^2 \} & \mbox{if }
\mu_t \neq 0 
\end{array}
\right .
\end{align}
Obviously, the two expressions are equivalent when both 
$\mu_s$ and $\mu_t$ are non-zero.
The interest of this set is that it happens to always contain $\mathcal{Z}(\widetilde \xi)$ and $ \mathcal{Z}(\widetilde \xi^\star)$. This, together with the fact that when $\mu_s=\mu_t=0$ these sets are contained in the axes
${\mathcal Z}(\eta_i)$
will allow us to ignore them altogether.
\begin{lemma}
 Assume that at least one of $\{ \mu_s, \mu_t \}$ is non-zero.  Then 
$\mathcal{Z}(\widetilde \xi), \mathcal{Z}(\widetilde \xi^\star) \subset \mathcal{Z}(\mu_s \widetilde \xi -\mu_t   \widetilde \xi^\star)$.
\end{lemma}
\begin{proof}
 Consider first $\mu_s, \mu_t \neq 0 $. Then at $\mathcal{Z}(\mu_s \widetilde \xi -\mu_t \widetilde \xi^\star) \cap \big\{ \bigcap_{j=1}^{p} \lrbrace{x_j = y_j = 0} \big\}$ we have that $z_1 = \pm \mu_t$ and $z_2 = \mp \mu_s$
which establishes $\mathcal{Z}(\widetilde \xi), \mathcal{Z}(\widetilde \xi^\star) \subset \mathcal{Z}(\mu_s \widetilde \xi -\mu_t   \widetilde \xi^\star)$ in this case.  When $\mu_t= 0,~\mu_s \neq 0$, by definition of the respective sets
we have $\mathcal{Z}(\widetilde \xi) = \mathcal{Z}(\mu_s \widetilde \xi -\mu_t \widetilde \xi^\star)$. 
Moreover, directly from \eqref{zxitilde} one finds
\begin{align*}
 \mathcal{Z}(\widetilde \xi^\star) &=  \bigcap_{j=1}^{p} \lrbrace{x_j = y_j = 0} \cap \lrbrace{z_1 = 0, z_2 = \pm \mu_s},
\end{align*}
which (cf. the first expression in  \eqref{xiproptoxistar}) 
 is clearly contained  in  $\mathcal{Z}(\mu_s \widetilde \xi -\mu_t \widetilde \xi^\star)$. An analogous argument applies in the case $\mu_t \neq 0,~ \mu_s = 0$.
\end{proof}

Let us define the following auxiliary coordinates
\begin{equation} 
 \hat z_+ := \frac{\mu_s z_1 + \mu_t z_2}{\sqrt{\sum_{i=1}^p (x_i^2 + y_i^2)}},\quad\quad \hat z_- := \frac{\mu_s z_2 - \mu_t z_1}{\sqrt{\sum_{i=1}^p (x_i^2 + y_i^2)}},\quad\quad \hat x_i := x_i, \quad\quad \hat y_i := y_i.
\end{equation}
Except for the case $\mu_s = \mu_t = 0$ (which will be analyzed later) the coordinates $\lrbrace{\hat z_+, \hat z_-,\hat x_i,\hat y_i}$ obviously cover
$\mathbb{R}^n \backslash \big\{\bigcap_{j=1}^{p} \{x_j = y_j = 0\}\big\}$. In terms of the adapted coordinates, they read
\begin{equation}\label{auxcart}
\hat z_+ = \alpha \sin(t_-),\quad\quad\hat z_- = \beta \sinh(t_+)\quad\quad \hat x_i = \frac{\alpha_i}{U} \cos(\phi_i),\quad\quad \hat y_i = \frac{\alpha_i}{U} \sin(\phi_i).
\end{equation}
Let us analyze the points where \eqref{auxcart} fails to be a change of coordinates  and hence restrict the domain of definition of $\{\alpha,t_-,t_+,\alpha_i,\phi_i\}$.  The first thing to notice is that a change of sign in the coordinate $\alpha_i$ is equivalent to a rotation of angle $\pi$ in the coordinate $\phi_i$.
Moreover, at points where  $\alpha_i = 0$, i.e. the axis of
$\eta_i$,  the coordinate $\phi_i$ is completely degenerate, which obviously
excludes $\bigcup_{j=1}^{p} \lrbrace{x_j = y_j = 0}$ from the region covered by the adapted coordinates. To avoid duplications, we must restrict $\alpha_i \in (0,1)$ and $\phi \in [-\pi,\pi)$ or alternatively $\alpha_i \in (-1,1)\backslash\{0\}$ and $\phi_i \in [0,\pi)$. We choose the former for definiteness.

The hypersurface $\{\alpha = \mbox{const}, t_- = \mbox{const}, t_+ = \mbox{const}\}$ is an $n-3$ dimensional sphere of radius $U^{-1}$, namely $\{\hat z_-=
\mbox{const},~\hat z_+ = \mbox{const} \}\cap \{  \sum_{i=1}^p (x_i^2 + y_i^2) = U^{-2} = \mbox{const}\}$. This gives a straightforward splitting of {$\mathbb{R}^n \backslash  \{ 0_{n-2} \}$, with $0_{n-2}  := \{\bigcap_{j=1}^{p} \{x_j = y_j = 0\}\big\}$, into $\mathbb{R}^2\times (\mathbb{R}^{n-2} \backslash \{ 0_{n-2} \} )$, where $\mathbb{R}^{n-2} \backslash \{ 0_{n-2} \}$
is foliated by $n-3$ dimensional spheres.} The set
${\mathcal Z}(\mu_s \widetilde{\xi} - \mu_t \widetilde \xi^{\star})$
respects this foliation, so it descends to
$\mathbb{R}^2  \times \mathbb{R}^+$ (the last factor is the radius of the $n-3$ sphere). To avoid extra notation
we also use ${\mathcal Z}(\mu_s \widetilde{\xi} - \mu_t \widetilde \xi^{\star})$
to denote this quotient set. We next show that the adapted coordinates actually cover the largest possible domain, namely $\mathbb{R}^n \setminus \{
{\mathcal Z}(\mu_s \widetilde{\xi} - \mu_t \widetilde \xi^{\star}) \cup 
\bigcup_{j=1}^{p} \lrbrace{x_j = y_j = 0} \}$. From the previous discussion, this is a consequence of the following result.
\begin{lemma}
Assume that at least one of $\{ \mu_s, \mu_t\}$ is not zero. Then. the transformation 
 \begin{equation}\label{mapauxz}
  \begin{array}{rcl}
  (\hat z_+, \hat z_-, U): \mathbb{R}\times[-\pi,\pi) \times \mathbb{R}^+  & \longrightarrow & \lr{\mathbb{R}^2\times \mathbb{R}^+} \backslash \mathcal{Z}(\mu_s \widetilde \xi -\mu_t   \widetilde \xi^\star)\\
(t_+,t_-,\alpha) & \mapsto & (\hat z_+, \hat z_-, U). 
  \end{array}
 \end{equation}
is a diffeomorphism.
\end{lemma}
\begin{proof}
 The determinant of the jacobian of \eqref{mapauxz} reads
\begin{equation}
 \left|\frac{\partial(\hat z_+, \hat z_-, U)}{\partial(t_+,t_-,\alpha)}\right| = \alpha U.
\end{equation}
Since $U$ is strictly positive (cf. \eqref{solUgen}
and recall that we chose $\epsilon =1$ w.l.o.g.), the conflictive points are $\alpha = 0$. To calculate the locus $\{\alpha = 0\}$ we obtain the inverse transformation of $\alpha$ in terms of $U, \hat z_+,\hat z_-$ by solving \eqref{solUgen} and the first two in \eqref{auxcart}. The result is, after a straightforward
computation,
\begin{equation}\label{inversealpha}
 \alpha = \pm\lr{\hat z_+^2 +\frac{1}{4 U^2 (\mu_s^2 + \mu_t^2)^2}(\hat z_+^2 + \hat z_-^2 -U^2(\mu_s^2 + \mu_t^2)^2 + (\mu_s^2 + \mu_t^2))^2}^{1/2}.
 \end{equation}
 It follows that $\alpha = 0$ is equivalent to $\hat z_+ = 0$ and $\hat z_-^2  +\mu_s^2 +  \mu_t^2 = U^2 (\mu_s^2 +  \mu_t^2)^2$. When translated into the
 original coordinates$\{z_1,z_2,x_i,y_i\}$ this set is
 precisely  $\mathcal{Z}(\mu_s \widetilde \xi -\mu_t   \widetilde \xi^\star)$. Also, from \eqref{inversealpha} it is obvious that $\alpha$ is multivalued, which also implies that $t_-$ is multivalued after substituting $\alpha$ as a function of $\hat z_+, \hat z_-, U$ in the first equation in \eqref{auxcart}\footnote{This was already evident by observing that a change of sing in $\alpha$ is cancelled 
   by a rotation of $\pi$ in $t_-$}. We solve this issue by restricting
 $\alpha$ to be strictly positive and let $t_-$  take values in $[-\pi,\pi)$.
\end{proof}

%
%
We have shown that the adapted coordinates cover all $\mathbb{R}^n$
except
$\bigcup_{j=1}^{p} \mathcal{Z}(\eta_j) \cup \mathcal{Z}(\mu_s \widetilde \xi -\mu_t   \widetilde \xi^\star)$ . The domain of definition of the coordinates $t_1,t_2$ depends on $\mu_t$ and $\mu_s$, because $ -\pi \leq t_- = \mu_t t_2 - \mu_s t_1 < \pi $. This defines a band $B(\mu_s,\mu_t) := \{-\pi \leq t_- = \mu_t t_2 - \mu_s t_1 < \pi\}$, whose width and tilt is determined by $\aa, \bb$ through $\mu_s, \mu_t$ (see figure \ref{fig2}). Nevertheless, the coordinate change
is well defined for all values of $t_1$ and $t_2$ and involves
only periodic functions of $t_-$. Thus, we can extend the domain
of definition of  $t_1,t_2$ to all of $\mathbb{R}^2$. This defines
a covering of the original space 
$\mathbb{R}^n \backslash (\bigcup_{j=1}^{p} \mathcal{Z}(\eta_j) \cup \mathcal{Z}(\mu_s \widetilde \xi -\mu_t   \widetilde \xi^\star))$  which unwraps completely
the orbits of $\widetilde{\xi}$ and $\widetilde \xi^{\star}$. It is not the universal covering because it does not unwrap the orbits of the axial vectors. 
This result is a generalization to higher dimensions of the covering dicussed in detail in \cite{marspeon20}.

The limit case $\mu_s = \mu_t = 0$ (that is $\aa = \bb = 0$) corresponds with a band of infinite width, i.e. $ B(\mu_s,\mu_t) =  \mathbb{R}^2$. In this case, the adapted coordinates also cover the largest possible set
$\mathbb{R}^n \backslash (\bigcup_{j=1}^{p} \mathcal{Z}(\eta_j))$. Recall that in  this case the only points where $\{ \widetilde \xi, \widetilde \xi^{\star},\eta_i \}$ is not a linearly independent set is the union of
${\mathcal Z} (\widetilde \xi), {\mathcal Z} (\widetilde \xi^{\star})$, and
${\mathcal Z} (\eta_i)$ and we have already seen that in this case
${\mathcal Z} (\widetilde \xi) = {\mathcal Z} (\widetilde \xi^{\star})
\subset {\mathcal Z} (\eta_i)$, for $i = 1, \cdots, p$.
This limit case is 
the same result that we would have obtained, had we performed
a direct analysis using $U$ as given by \eqref{Ulim}.

 \begin{figure}[h!]   
 \begin{center}
 \psfrag{Z}{$\theta$}
  \psfrag{W}{$w$}
    \psfrag{T1}{$t_1$}
      \psfrag{T2}{$t_2$}
   \includegraphics[width=7cm]{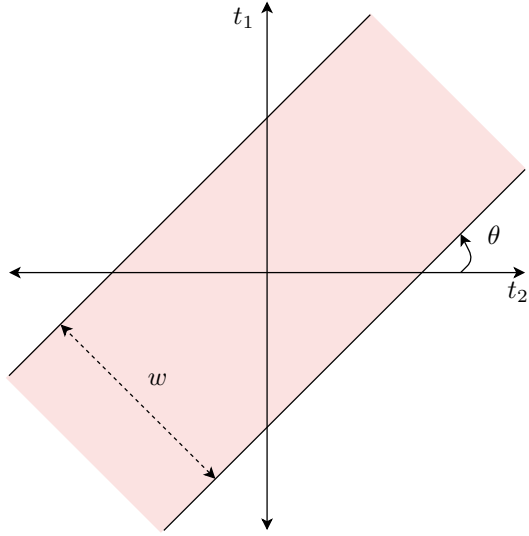}  
 \caption{Band $B(\mu_s,\mu_t)$ where the coordinates $t_1,t_2$ are defined. The tilt is given by $\theta = \arctan\lr{\frac{\mu_s}{\mu_t}}$ and the width $w$ is $2 \pi/\mu_t$ if $\mu_t \neq 0$, $2 \pi/\mu_s$ if $\mu_t = 0,~\mu_s \neq 0$ and $w \rightarrow \infty$ if $\mu_s = \mu_t = 0$.}\label{fig2}     
 \end{center} 
 \end{figure}

Once we have determined the adapted coordinates and the region they cover, we may proceed to calculate the expression of the Euclidean metric 
\begin{equation}\label{flatcart}
 g_E = \dif z_1^2 + \dif z_2^2 + \sum\limits_{i=1}^{p} \left ( \dif  x_i^2 + \dif y_i^2 \right ).
\end{equation}
in adapted coordinates. We start with the term $\sum\limits_{i=1}^{p} \left ( \dif x_i^2 + \dif y_i^2 \right )$, which is straightforward 
\begin{equation}\label{flatxiyi}
  \sum\limits_{i=1}^{p} \left ( \dif x_i^2 + \dif y_i^2 \right ) = \frac{\dif U^2}{U^4} +\frac{1}{U^2}\sum\limits_{i=1}^{p}\restr{\lr{\dif \alpha_i^2 + \alpha_i^2 \dif \phi_i^2}}{\sum_{i=1}^p \alpha_i^2 = 1} - \frac{2 \dif U}{U^2} \left ( \sum\limits_{i=1}^p \alpha_i \dif \alpha_i \right ) = \frac{\dif U}{U^4} +\frac{1}{U^2}
  \gamma_{\mathbb{S}^{n-3}},
\end{equation}
where in the last equality we used $ \sum_{i=1}^p \alpha_i \dif \alpha_i = 0$, which follows from $\sum_{i=1}^p \alpha_i^2 = 1$ and we have defined
\begin{equation} 
\gamma_{\mathbb{S}^{n-3}}:=\sum\limits_{i=1}^{p}\restr{\lr{\dif \alpha_i^2 + \alpha_i^2 \dif \phi_i^2}}{\sum_{i=1}^p \alpha_i^2 = 1}. \label{gammasph}
\end{equation}
The notation is justified because the right-hand side corresponds
to the standard unit metric on $\mathbb{S}^{n-3}$. This follows because
$\sum_{i=1}^{p}\lr{\dif \alpha_i^2 + \alpha_i^2 \dif \phi_i^2}$ is obviously flat and the restriction $\sum_{i=1}^p \alpha_i^2 = 1$ defines a unit sphere.
We emphasize, however that the notation $\gamma_{\mathbb{S}^{n-3}}$ refers to  the
quadratic form above, not to the spherical metric in any other coordinate
system. Observe also that $\dif U$ in \eqref{flatxiyi} should be understood as a
short name for the explicit differential of $U$ in terms of
$\dif t_1, \dif t_2$, $\dif \alpha$. Using \eqref{flatcart} and \eqref{flatxiyi}, we have
\begin{equation}
 g_{t_1 t_1} = \lr{\frac{\partial z_1}{\partial t_1}}^2 + \lr{\frac{\partial z_2}{\partial t_1}}^2 + \frac{1}{U^4} \lr{\frac{\partial U}{\partial t_1}}^2,
\end{equation}
which after a explicit calculation reduces to
\begin{equation}
 g_{t_1 t_1} = \frac{\alpha^2 +\mu_t^2}{ U^2}.
\end{equation}
Notice that $g_{t_1 t_1} = g_E(\widetilde \xi, \widetilde \xi)$, $g_{t_2 t_2} = g_E(\widetilde \xi^\star, \widetilde \xi^\star)$ and $g_{t_1,t_2} = g_E(\widetilde \xi, \widetilde \xi^\star)$. From the expressions in cartesian coordinates it is straightforward to show
\begin{equation}
 g_E(\widetilde \xi, \widetilde \xi) = g_E(\widetilde \xi^\star, \widetilde \xi^\star) - \aa\sum\limits_{i=1}^{p}  (x_i^2 +  y_i^2) = g_E(\widetilde \xi^\star, \widetilde \xi^\star) - \frac{\aa}{ U^2}, \quad \quad g_E(\widetilde \xi, \widetilde \xi^\star) = \frac{\bb}{2}\sum\limits_{i=1}^{p}  (x_i^2 +  y_i^2) = \frac{\bb}{2 U^2}
\end{equation}
where we have used $U^{-2} = \sum\limits_{i=1}^{p}  (x_i^2 +  y_i^2)$ (see \eqref{adaptxiyi}). Thus
\begin{equation}
 g_{t_2 t_2} = g_{t_1 t_1} + \frac{\aa}{U^2}= \frac{\alpha^2 + \mu_s^2}{ U^2},\quad \quad g_{t_1 t_2} = \frac{\bb}{2 U^2} = \frac{\mu_s \mu_t}{ U^2}.
\end{equation}
The remaining terms are rather long to calculate. With the aid of a computer algebra system one gets
\begin{align}
 g_{\alpha \alpha} &= \lr{\frac{\partial z_1}{\partial \alpha}}^2 + \lr{\frac{\partial z_2}{\partial \alpha}}^2 + \frac{1}{U^4} \lr{\frac{\partial U}{\partial \alpha}}^2 = \frac{1}{\beta^2 U^2}\\
 g_{\alpha t_1} &= \frac{\partial z_1}{\partial \alpha} \frac{\partial z_1}{\partial t_1} + \frac{\partial z_2}{\partial \alpha}\frac{\partial z_2}{\partial t_1} + \frac{1}{U^4} \frac{\partial U}{\partial \alpha}\frac{\partial U}{\partial t_1} = 0,\\
 g_{\alpha t_2} &= \frac{\partial z_1}{\partial \alpha}\frac{\partial z_1}{\partial t_2} + \frac{\partial z_2}{\partial \alpha}\frac{\partial z_2}{\partial t_2} + \frac{1}{U^4} \frac{\partial U}{\partial \alpha}\frac{\partial U}{\partial t_2} = 0.
\end{align}
Notice that no terms in $\dif \alpha_i,~\dif \phi_i$ appear but those in 
$\gamma_{\mathbb{S}^{n-3}}$, since neither $U$ nor $z_1,~z_2$ depend on $\alpha_i, \phi_i$.
Putting all these results  together we obtain the following expression: 
\begin{lemma}
 In adapted coordinates $\{t_1,t_2,\alpha, \alpha_i, \phi_i \}$, the Euclidean metric $g_E$ takes the form
 \begin{equation}\label{flatadatped}
g_E = \frac{1}{U^2} \lr{(\alpha^2 + \mu_t^2)\dif t_1^2 + (\alpha^2 + \mu_s^2) \dif t_2^2 + 2 \mu_s \mu_t  \dif t_1 \dif t_2 + \frac{\dif \alpha^2}{\alpha^2 + \mu_s^2 + \mu_t^2} +\gamma_{\mathbb{S}^{n-3}} }.
\end{equation}
\end{lemma}
We would like to stress the simplicity of this result. Except in the a
global conformal factor, the metric does not depend  in $t_1$ and $t_2$ (so, both $\widetilde \xi$ and $\widetilde \xi^{\star}$ are Killing vectors
of $U^2 g_E$). The dependence in the coordinate $\alpha$ and
the conformal class constants $\{ \mu_s, \mu_t\}$ is also extremely simple. Even more,
the fact that all dependence in $\{ \alpha_i, \phi_i \}$ arises only in
$\gamma_{\mathbb{S}^{n-3}}$ allows us to use any other coordinate system on the unit $\mathbb{S}^{n-3}$. Any such coordinate system is still adapted to $\widetilde{\xi}$ and $\widetilde \xi^{\star}$ but (in general)
no longer to  $\{ \eta_i \}$.  
This enlargement to partially adapted coordinates
is an interesting consequence
of the foliation of $\mathbb{R}^n$ by $(n-3)$-spheres described above.

We now work out the odd $n$ case. As already discussed, we will base the
analysis on the even dimensional case by restricting to a suitable a hyperplane. The underlying reason why this is possible is given
in the following lemma.
\begin{lemma}\label{lemmaextCKVF}
Fix $n \geq 3$  odd. Let  $\xi_F$ be a CKVF of $\mathbb{E}^n$ in canonical form and let $\{z_1, x_i, y_i\}$ be canonical coordinates. Consider the embedding $\mathbb{E}^n \hookrightarrow \mathbb{E}^{n+1}$ where $\mathbb{E}^n$ is identified with the hyperplane $\lrbrace{z_2=0}$, for a cartesian coordinate $z_2$ of $\mathbb{E}^{n+1}$. Then $\xi_F$ extends to a CKVF of $\mathbb{E}^{n+1}$ with the same value of $\aa,\mu_i $ and $\bb = 0$.
\end{lemma}
\begin{proof}
  By Remark \ref{tildexi} and Theorem \ref{theoisomskwconf},
  the expression of $\xi_F$ in the canonical coordinates
  $\{ z_1, x_i, y_i\}$ is
  \begin{align*}
    \xi_F = & \lr{\frac{\aa}{2} + \frac{1}{2}\lr{z_1^2  - \sum\limits_{i=1}^{p}(x_i^2 + y_i^2)}} \partial_{z_1} + 
              z_1 \sum\limits_{i=1}^{p}\lr{x_i \partial_{x_i} + y_i \partial_{y_i}} + \sum \limits_{i=1}^{p} \mu_i \lr{x_i \partial_{y_i} - y_i \partial_{x_i}} :=  \widetilde \xi +  \sum\limits_{i=1}^{p} \mu_i \eta_i.
\end{align*}
Define $\xi'_F$ on $\mathbb{E}^{n+1}$ in cartesian coordinates
$\{ z_1, z_2, x_i ,y_i \}$ by  $\xi'_F = \widetilde \xi' 
+ \sum \limits_{i=1}^{p} \mu_i \lr{x_i \partial_{y_i} - y_i \partial_{x_i}} 
$ where $\widetilde \xi'$ is given by \eqref{CKVFxi}
with $\tau =0$. It is clear that this vector is a CKVF of $\mathbb{E}^{n+1}$ 
written in canonical form, that it is tangent to the hyperplane $z_2=0$ and
that it agrees with $\xi_F$ on this submanifold.  
\end{proof}
%

Consequently, introducing adapted coordinates for the extended CKVF
and  restricting to $\{z_2 = 0\}$ will provide adapted coordinates for $\xi_F$. The restriction will obviously reduce the domain of definition of the adapted coordinates $(t_1,t_2,\alpha,\alpha_i,\phi_i)$ to a hypersurface. It is straightforward from equation \eqref{adaptz2} and the second equation in \eqref{adaptz1z2lim} that for the three cases $\cc >0$, $\cc = 0$ or $\cc < 0$,  the hyperplane $\{z_2 = 0 \}$ corresponds to $\{t_2 = 0 \}$. 
 It follows that the remaining coordinates  $\{t_1,\alpha,\alpha_i,\phi_i\}$
are adapted to $\widetilde \xi$ and all  $\eta_i$. Their domain of definition is $t_1 \in \mathbb{R}$, $ \alpha \in \mathbb{R}^+ $, $ \alpha_i \in (0,1)$, $ \phi_i \in [-\pi,\pi)$  and the coordinate change is given by 
\eqref{adaptz1} (or the first in \eqref{adaptz1z2lim}) together with \eqref{adaptxiyi} after setting $\tau = 0$ and $t_2 = 0$. 
Depending on the sign of $\aa$ one gets for $z_1$
\begin{equation}\label{adaptoddz}
z_1 = \left \{ \begin{array}{ll}
\frac{-1}{|U^+|}\frac{\alpha  \sin (\sqrt{\aa}t_1)}{\sqrt{\aa}}, & \aa >0 \\
\frac{-1}{|U^-|}\frac{\sqrt{\alpha^2 + |\aa|}  \sinh (\sqrt{|\aa|}t_1)}{\sqrt{|\aa|}}, & \aa < 0 \\
\frac{-1}{|U^0|} \alpha t_1, & \aa =0
\end{array} \right .,
\end{equation}
where 
\begin{equation}
 U^{+} := \frac{1}{\aa}(\sqrt{\alpha^2 + \aa} - \alpha \cos(\sqrt{\aa}t_1)),\quad\quad U^{-} := \frac{1}{-\aa}(\sqrt{\alpha^2 - \aa}\cosh(\sqrt{-\aa} t_1) - \alpha), \quad\quad U^{0} := \frac{1}{2}(\alpha t_1^2 + \frac{1}{\alpha}),\label{Uodd}
\end{equation}
and for all three cases
 \begin{equation}
 x_i = \frac{\alpha_i}{U^\epsilon} \cos(\phi_i),\quad\quad
 y_i = \frac{\alpha_i}{U^\epsilon} \sin(\phi_i), \label{adaptoddxy}
\end{equation}
 where we write $U^\epsilon$ for the function $U^+, U^-$ or $U^0$ according with sign of $\aa$. 
 
The range of variation of $\{ t_1, \alpha, \alpha_i, \phi_i\}$
was inferred before from the corresponding range 
of variation of $\{ t_1, t_2, \alpha, \alpha_i, \phi_i \}$ 
in $\mathbb{E}^{n+1}$. It may happen, however, that when we restrict to
the hyperplane $\{ z_2 =0\}$, the range gets enlarged and additional points
get covered by the adapted coordinate system. The underlying reason is that, in effect, we are no longer adapting coordinates to  $\widetilde \xi'{}^{\star}$, so the points on $z_2=0$ where this vector is linearly dependent to
$\widetilde \xi'$ (or zero) are no longer problematic. When
$\tau=0$, one has
\begin{align*}
(\mu_s = \sqrt{\sigma}, \quad \mu_t =0) \quad \mbox{ if } \sigma \geq 0, \qquad \qquad 
(\mu_s = 0, \quad \mu_t = \sqrt{|\sigma|}) \quad \mbox{ if } \sigma \leq 0.
\end{align*}
We may ignore the case $\sigma=0$ because ${\mathcal Z} 
(\widetilde \xi') = 
{\mathcal Z} (\widetilde \xi'{}^{\star})$. It follows from \eqref{zxi} 
and \eqref{xiproptoxistar} that
\begin{align*}
\left .   \mathcal{Z}(\widetilde \xi') \right |_{z_2=0} &=
\left \{ \begin{array}{ll}
       \{ z_1 = 0 \}\cap \left \{  \sum\limits_{i=1}^p (x_i^2 + y_i^2) = \sigma \right \} & \mbox{if } \sigma > 0 \\
 \bigcap_{j=1}^{p} \lrbrace{x_j = y_j = 0} \cap \lrbrace{z_1 = \pm \sqrt{|\sigma|}} & \mbox{if } \sigma < 0 
\end{array} \right . \\
\left .  \mathcal{Z}(\mu_s \widetilde \xi' -\mu_t   \widetilde \xi'{}^\star )
\right |_{z_2 =0}  & =
\left \{   \begin{array}{ll}
             \{ z_1 = 0 \}\cap \left \{  \sum\limits_{i=1}^p (x_i^2 + y_i^2) = \sigma \right \} & 
\quad \qquad \mbox{if } \sigma > 0 \\
 \{ z_1^2 +  \sum\limits_{i=1}^p (x_i^2 + y_i^2) = |\sigma| \} & 
\quad \qquad \mbox{if } \sigma < 0.  
\end{array}
\right . 
\end{align*}
When $\sigma >0$, the two sets are the same and no extension of the coordinates
$\{ t_1, \alpha, \alpha_i, \phi_i\}$ is possible. However, when $\sigma <0$,
the set  
$\mathcal{Z}(\mu_s \widetilde \xi' -\mu_t   \widetilde \xi'{}^\star )
|_{z_2 =0}$ is strictly larger than
$\mathcal{Z}(\widetilde \xi') |_{z_2=0}$.
From expressions \eqref{adaptoddz} and \eqref{adaptoddxy} one checks that
$\mathcal{Z}(\mu_s \widetilde \xi' -\mu_t   \widetilde \xi'{}^\star )
|_{z_2 =0} \setminus \mathcal{Z}(\widetilde \xi') |_{z_2=0}$ corresponds exactly to the value $\alpha =0$
and that  
$\mathcal{Z}(\widetilde \xi) =  \mathcal{Z}(\widetilde \xi') |_{z_2=0}$ 
is at the limit $t_1 \rightarrow \pm \infty$.
Thus, a priori there is the possibility that the adapted
coordinates $\{ t_1, \alpha, \alpha_i, \phi_i \}$ can be extended regularly
to  $\alpha=0$ when $\sigma <0$. It follows directly from 
\eqref{adaptoddz} that this is indeed the case (observe that, to the contrary,
 the limit $\alpha \rightarrow 0$ in \eqref{adaptoddz}
is singular when $\sigma \geq 0$, in agreement with the previous discussion).
Thus, the range of definition of $\alpha$ is $[0, \infty)$ when
$\sigma <0$. The conclusion is that,
irrespectively of the value of 
$\sigma$,
the adapted coordinates $\{ t_1, \alpha, \alpha_i, \phi_i\}$ cover the largest possible domain of $\mathbb{E}^n$, namely all points where
$\widetilde \xi$ is non-zero away from the axes of $\{ \eta_i \}$.

To obtain the Euclidean metric in $\mathbb{E}^{n}$ for $n$ odd
in adapted coordiantes we simply restrict \eqref{flatadatped} (with $n \rightarrow n+1$)
to the hypersurface $t_2 =0$, and get 
\begin{align}\label{eqflatodd}
g_E^\epsilon = \frac{1}{(U^{\epsilon})^2} \left ( 
\left (\alpha^2 + \frac{(1 - \epsilon) | \sigma|}{2}  \right ) dt_1^2 
 + \frac{d\alpha^2}{\alpha^2 + |\sigma|} + \gamma_{\mathbb{S}^{n-2}}
\right ), 
\end{align}
where $\epsilon = -1, 0, 1$  respectively if $\cc <0, \cc=0, 
\cc >0$.

\begin{remark}\label{remarkextensiong}
  {The three odd dimensional cases can be unified into one. The function $U^0$ coincides with the limits of $U^+$ and $U^-$ when $\aa\rightarrow 0$. However, the analytical continuation of  $U^+$ to negative values of $\cc$ does not directly yield $U^-$. To solve this we introduce the function
    \begin{align*}
      W_1 (y) = \frac{1}{\cc} \left (\sqrt{y^2 + \cc}
      - y \cos\left ( \sqrt{\cc} t_1 \right ) \right ),
    \end{align*}
      which is analytic in $\cc$ and takes real values for real $\cc$.
      We observe that
      $U^+(\alpha=y) = W_1 (y)$ for  $\cc >0$,
      $U^0(\alpha=y) = W_1 (y)$ ($\cc =0$) and
      $U^-(\alpha = + \sqrt{ y^2 + \cc}) = W_1 (y)$ ($\cc <0$). This suggests introducing the coordinate change $\alpha =y$ for $\cc \geq 0$ and
      $\alpha = + \sqrt{y^2 + \cc}$ for $\cc<0$. From the domain of $\alpha$, it follows that  $y$ takes values in
      $y > 0$ when $\cc \geq 0$ and
      $y \geq \sqrt{-\cc}$ when $\cc <0$. In terms of $y$,
      the three metrics metric $g^{\epsilon}$ take the unified form
      \begin{align*}
        g^{\epsilon}_E = \frac{1}{W_1(y)^2} \left ( y^2 \dif t_1^2 + \frac{\dif y^2}{y^2+
          \cc} + \gamma_{\mathbb{S}^{n-2}} \right ).  
      \end{align*}
      The function $W_1$ is the analytic continuation of $U^+$ to negative values of $\cc$. We could have started with $U^{-}$ and continued analytically to positive values of $\cc$. Instead of repeating the argument, we simply introduce a new variable $z$ defined by $y = \sqrt{z^2 - \cc}$ with range
      of variation $z > \sqrt{\cc}$ for $\cc \geq 0$ and $z \geq 0$ for
      $\cc<0$. The metric takes the (also unified and even more symmetric) form
      \begin{align*}
  g^{\epsilon}_E = \frac{1}{W_2(z)^2} \left ( (z^2- \cc) \dif t_1^2 +
  \frac{\dif z^2}{z^2 - \cc} + \gamma_{\mathbb{S}^{n-2}} \right ),
  \quad
  W_2(z):= \frac{1}{\cc} \left ( z - \sqrt{z^2 - \cc} \cos \left ( \sqrt{\cc} t_1 \right ) \right ).
     \end{align*}
      The function $W_2(z)$ is again  analytic in $\cc$,  takes real values on the real line, and now it extends $U^-$. More specifically,
            $U^-(\alpha=z) = W_2 (z)$ ($\cc <0$),
$U^0(\alpha=z) = W_2 (z)$ ($\cc =0$) and
$U^+(\alpha=\sqrt{z^2 - \cc} ) = W_2 (z)$ ($\cc >0$).}
\end{remark}


{Remark \ref{remarkextensiong} allows us to work with all the odd dimensional cases at once, which will be useful for Section \ref{secTTtens}. However,
  this unified form does not arise naturally when the odd
  dimensional case is viewed as a consequence of the $n+1$ even dimensional case. So, leaving aside this remark for Section \ref{secTTtens}, we summarize the results of this section in the following Theorem. }

\begin{theorem}{
 Given a CKVF $\xi_F$ of $\mathbb{E}^n$, with $n \geq 4$ even, in canonical form $\xi_F = \widetilde \xi + \sum_{i=1}^p \mu_i \eta_i$, the coordinates $t_1,t_2, \phi_i, \alpha, \alpha_i$, for $i=1, \cdots p$ and $\sum_{i=1}^p \alpha_i^2 = 1$, defined 
 by
 \begin{equation}
 z_1 =-\frac{1}{U}\frac{\partial U}{\partial t_1},\quad\quad
 z_2 = \frac{1}{U}\frac{\partial U}{\partial t_2} \quad\quad
 x_i = \frac{\alpha_i}{U} \cos(\phi_i),
\quad\quad
 y_i = \frac{\alpha_i}{U} \sin(\phi_i) 
 \end{equation}
with
\begin{equation}
  U  = \frac{\sqrt{\alpha^2 + \mu_t^2 + \mu_s^2}\cosh(\mu_t t_1 + \mu_s t_2)-\alpha \cos(\mu_t t_2 -\mu_s t_1)}{\mu_t^2 + \mu_s^2},
\end{equation}
which admits a limit  $\lim_{\mu_s \mu_t \rightarrow 0} U = \frac{\alpha}{2}(t_1^2 + t_2^2) + \frac{1}{2 \alpha}$, furnish adapted coordinates to $\widetilde \xi =  \partial_{t_1}$ $ \widetilde \xi^\star = \partial_{t_2}$ $\eta_i = \partial_{\phi_i}$, which cover the maximal possible domain, namely
$\mathbb{E}^n \backslash \left ( \bigcup_{j=1}^{p} \mathcal{Z}(\eta_j) \cup \mathcal{Z}(\mu_s \widetilde \xi -\mu_t   \widetilde \xi^\star)\right )$ for $t_1, t_2 \in B(\mu_s,\mu_t)$, $\phi_i \in [-\pi, \pi)$, $\alpha_i \in (0,1)$ and $\alpha \in \mathbb{R}^+$. Moreover, the metric $g_E$, which is flat in canonical cartesian coordinates, is given by 
 \begin{equation}
   g_E = \frac{1}{U^2} \Big((\alpha^2 + \mu_t^2)\dif t_1^2 + (\alpha^2 + \mu_s^2) \dif t_2^2 + 2 \mu_s \mu_t  \dif t_1 \dif t_2 + \frac{\dif \alpha^2}{\alpha^2 + \mu_s^2 + \mu_t^2} +\sum\limits_{i=1}^{p}\restr{\lr{\dif \alpha_i^2 + \alpha_i^2 \dif \phi_i^2}}{\sum_{i=1}^p \alpha_i^2 = 1} \Big).
   \label{n+1metric}
\end{equation} 
 If $n \geq 3$ is odd and $\xi_F$ is in canonical form, $\xi_F = \widetilde \xi + \sum_{i=1}^p \mu_i \eta_i$, the coordinates $\{t_1, \phi_i, \alpha, \alpha_i \}$ adapted to $\widetilde \xi =  \partial_{t_1}$ $\eta_i = \partial_{\phi_i}$ are given by the case of $n+1$ (even) dimensions, for $\bb = 0$ restricted to $t_2 = 0$ (which defines the embedding $\mathbb{E}^n = \{z_2 = 0\} \subset \mathbb{E}^{n+1}$) and cover again the maximal possible domain, given by $\mathbb{E}^n \backslash \left ( \bigcup_{j=1}^{p} \mathcal{Z}(\eta_j) \cup \mathcal{Z}( \widetilde \xi) \right )$ for $t_1 \in \mathbb{R}$, $\phi_i \in [-\pi, \pi)$, $\alpha_i \in (0,1)$ and $\alpha \in \mathbb{R}^+$ when
   $\cc \geq 0$ and $\alpha \in \mathbb{R}^+ \cup \{0\}$ when $\cc <0$. Moreover, the metric $g_E$, which is flat in canonical cartesian coordinates, is given by the pull-back of \eqref{n+1metric} at $t_2 = 0$ after setting $\bb = 0$. Explicitly
   $g_E$ is, depending on the sign of $\cc$,  given by
   \eqref{eqflatodd} with $\gamma_{\mathbb{S}^{n-2}}$ as in \eqref{gammasph}.} 
\end{theorem}

\section{TT-Tensors}\label{secTTtens}
 
{The adapted coordinates derived in Section \ref{secadapted} provide a useful tool to solve geometric equations involving CKVFs. In this section we give an example of this in the context of $\Lambda$-vacuum spacetimes admitting a smooth null conformal infinity.}

{Recall that for such spacetimes
the data at $\scri$ is a
conformal class $[g]$ of riemannian metrics and a conformal class
of transverse and divergence-free tensors.
More specifically, for a representative metric   $g$ in the
conformal class, there is associated a symmetric tensor $D^{AB}$ satisfying
$g_{AB} D^{AB}=0$ (divergence-free) and $\nabla_A D^{AB} =0$ (transverse). For any other metric $\tilde{g}=\Omega^2 g$ in the conformal class, the associated tensor is $\Omega^{-(n+2)} D^{AB}$, which is again a TT tensor with respect to $\tilde{g}$. In dimension $n=3$, it has been shown in \cite{KIDPaetz} that the spacetime generated by the Cauchy data at $\scri$ admits a Killing vector if and only if the metric $g$ admits a CKV $\xi$ (which is the restriction of the Killing vector to $\scri$) and $D$ satisfies the so-called Killing initial data (KID) equation. This equation admits a natural generalization to arbitrary dimension which is
 \begin{equation}\label{eqKID}
  \mathcal{L}_\xi D^{AB} + \frac{n+2}{n} \mbox{div}_g \xi  D^{AB} = 0,
 \end{equation}  
 where $\mbox{div}_g \xi $ is the divergence of $\xi$. Equation \eqref{eqKID} reduces to the KID equation of Paetz in dimension $n=3$ and it is conformally covariant, i.e.
 if $\{ g_{AB},D^{AB},\xi^A\}$ is a solution, then so it is $\{ \Omega^2 g_{AB},
 \Omega^{-(n+2)} D^{AB}, \xi^A\}$. We emphasize however, that in higher dimension
 ($n \geq 4$) it is not known whether a spacetime admiting a smooth $\scri$
 such that the corresponding data at null infinity solves the KID equation for some CKV $\xi$, must necessarily admit a Killing vector.}

A CKVF satisfying \eqref{eqKID} will be called KID vector for short.  An important property of KID vectors is that they form a Lie subalgebra of CKVFs, i.e. if $\xi, \xi'$ are KIDs for a given $TT$ tensor $D$, then $[\xi, \xi']$ is also a
KID for $D$.
The problem of obtaining all TT-tensors with generality for a given conformal structure is hard, even in the conformally flat case (see e.g. \cite{Beig1997}).  {In this section we exploit the results above  to
obtain the general solution of the KID equations in 
dimension $n=3$}
for spacetimes which possess two commuting symmetries, one of which is axial.
{This case is specially relevant since $n=3$ corresponds to the
physical case of four spacetime  dimensions
and the class necessarily contains the 
Kerr-de Sitter family of spacetimes, which is a particularly interesting explicit
familiy of spacetimes.} 
Our strategy is to take an arbitrary CKVF $\xi$, derive its canonical form $\xi_F = \widetilde \xi + \mu \eta$, adapt coordinates to $\widetilde \xi$ and $\eta$ and impose the KID equations to $\widetilde \xi$ and $\eta$. 
The problem simplifies notably in the conformal gauge to $g := (U^\epsilon)^2 g_E^\epsilon$ because both $\widetilde \xi$ and $\eta$ become Killing vector fields. From Remark \ref{remarkextensiong}, we may treat all cases
$\cc<0,~\cc = 0,~\cc>0$ at the same time by using the form of the metric
\begin{equation}\label{eqgkill}
  g =  \frac{\dif z^2}{z^2-\cc}   + (z^2-\cc) \dif t^2 
  + \dif \phi^2,
  \qquad \widetilde \xi = \partial_t, \qquad \eta = \partial_\phi.
  \end{equation} 
  We remark that even though we solve the problem by fixing the coordinates and conformal gauge, we shall write the final result in fully covariant form (cf. Theorem \ref{theoTTcov} below).

  In the conformal gauge of $g$, the condition that
  a TT-tensor $D$ satisfies KID equations for both $\widetilde \xi$ and $\eta$  (which is equivalent to imposing that $\xi$ and $\eta$ are KID vectors) is trivial in the adapted coordinates obtained in the previous section:
\begin{equation}
 \mathcal{L}_{\widetilde \xi} D^{AB} = \partial_t D^{AB} = 0,\quad\quad \mathcal{L}_{\eta} D^{AB} = \partial_\phi D^{AB} = 0. 
\end{equation}
Thus, $D^{AB}$ are only functions of $z$. The transversality condition is also
quite simple in adapted  coordinates:


\begin{align}
 \frac{\dif D^{zz}}{\dif z} - z \left ( \frac{D^{zz}}{z^2 - \cc} + (z^2 - \cc) D^{tt} \right) & = 0, \label{eqtrns1}\\
 \frac{\dif D^{zt}}{\dif z} +  \frac{2 z}{z^2 - \cc}D ^{zt} & = 0\label{eqtrns2}\\
 \frac{\dif D^{z\phi}}{\dif z} & = 0, \label{eqtrns3}
\end{align}
while the traceless condition imposes
\begin{equation}
 g_{AB} D^{AB} = \frac{D^{zz}}{z^2 - \cc} + (z^2 - \cc) D^{tt} + D^{\phi \phi} = 0\label{eqtrace}.
\end{equation}
There are no equations for $D^{t\phi}$ so
$D^{ \phi t} =h(z)$ with $h(z)$  an arbitrary function.
The general solution of equations \eqref{eqtrns2} and \eqref{eqtrns3} is
obtained at once and reads
\begin{equation}
  D^{zt} = \frac{K_1}{z^2- \cc},\quad\quad D^{z \phi } = K_2,\quad\quad K_1,K_2 \in \mathbb{R}.
\end{equation}
For equations \eqref{eqtrns1} and \eqref{eqtrace}, we let $D^{zz} =: f(z)$ be an arbitrary function and obtain the remaining components
\begin{equation}
 D^{\phi \phi} =- \frac{1}{z} \frac{\dif f}{\dif y},\quad \quad D^{ t t} = \frac{1}{z (z^2 - \cc)} \frac{\dif f}{\dif z} -\frac{f}{(z^2 - \cc)^2}.
\end{equation}
Summarizing 
\begin{lemma}\label{lemmaTTadapt}
  In the three-dimensional conformally flat class $[g]$, let $\xi_F$ be a CKVF.
  Decompose $\xi$ in canonical form $\xi = \widetilde \xi + \mu \eta$ and fix the conformal gauge so that $g$ given by \eqref{eqgkill}. Then the most general
  symmetric TT-tensor $D$ satifying the KID equations for $\xi$ and $\eta$ simultaneously is, in adapted coordinates $\{z,t, \phi\}$, a combination (with constants) of the following tensors
\begin{align}
 D_f & := f \partial_z \otimes \partial_z + \left (\frac{1}{z (z^2 - \cc)} \frac{\dif f}{\dif z} -\frac{f}{(z^2 - \cc)^2}\right ) \partial_t \otimes \partial_t   - \frac{1}{z} \frac{\dif f}{\dif z} \partial_\phi \otimes \partial_\phi,\\
  D_h & := h (\partial_t \otimes \partial_\phi +
        \partial_\phi \otimes \partial_t ),\quad\quad
        D_{\widetilde{\xi}} := \frac{1}{z^2 -\cc} (\partial_z \otimes \partial_t +
        \partial_t \otimes \partial_z )
        , \quad\quad D_{\eta} = \partial_z \otimes \partial_{\phi} +
        \partial_\phi \otimes \partial_z,
        \end{align}
where $f$ and $h$ are arbitrary functions of $z$.
\end{lemma}

Having obtained the general solution in a particular gauge, our next aim is to give a (diffeomorphism and conformal) covariant form of the generators in Lemma \ref{lemmaTTadapt}. From \cite{Kdslike}, we know that, for any
CKV $\xi$ of any $n$-dimensional metric $g$ (not necessarily conformally flat)
the following tensors are TT w.r.t. to $g$ and satisfy the KID equation with respect to $\xi$. 
\begin{equation}\label{eqDcan}
 \mathcal{D}(\xi) = \frac{1}{|\xi|_g^{n+2}}\left (\xi \otimes \xi - \frac{|\xi|_g^2}{n}g^\sharp  \right ),
\end{equation}
where $|\cdot|_g$ denotes the norm w.r.t. $g$ and $g^\sharp$ the contravariant form of $g$. Thus, we can rewrite $D_f$ as
\begin{equation}
 D_f = \left (- 2 (z^2 - \cc)^{1/2}f + \frac{(z^2 - \cc)^{3/2}}{z}\frac{\dif f}{\dif z} \right ) \mathcal{D}(\widetilde \xi) - \left (
 \frac{f}{z^2 - \cc} + \frac{1}{z} \frac{\dif f}{\dif z}\right ) \mathcal{D}(\eta).
\end{equation}
We now restore the conformal gauge freedom by considering the metric $\widehat g = \Omega^2 g$ and $\widehat D_f = D_f/\Omega^5$, for any (positive) conformal factor $\Omega$. Since the tensors $ \mathcal{D}(\widetilde \xi), \mathcal{D}(\eta)$ are already conformal and diffeomorphism covariant, we must impose their multiplicative factors in $\widehat D_f$ to be conformal and diffeomorphism invariant. With the gauge freedom restored, the norms of the CKVFs now are
\begin{equation} 
 |\widetilde \xi|_{\widehat g} = \Omega \sqrt{z^2 - \cc} ,\quad\quad  |\eta|_{\widehat g} =  \Omega.
\end{equation}
Then, considering $f =: \sqrt{\X} \widehat f(\X)$ as function of the conformal invariant quantity $\X = |\widetilde \xi|_{\widehat g}/|\eta|_{\widehat g} = \sqrt{z^2 - \cc}$, one can directly cast $\widehat D_f$ in the following form:
\begin{equation}
 \widehat D_f = \X^4 \frac{\dif}{\dif \X} \lr{\frac{\widehat f(\X)}{\X^{3/2}}} \mathcal{D}(\widetilde \xi) - \frac{1}{\X^2} \frac{\dif}{\dif \X} \lr{\X^{3/2} \widehat f(\X)} \mathcal{D}(\eta),
\end{equation}
which is a conformal and diffeomorphism covariant expression. Notice that the
expression is symmetric under the interchange $\widetilde \xi \leftrightarrow \eta$ because the coefficient of
$\mathcal{D}(\eta)$ expressed in the variable $\mathrm{Y} = \X^{-1}$
is identical in form to the coefficient
of $\mathcal{D}(\xi)$.

For the tensor $\widehat D_h := D_h/\Omega^5$, redifining $h =: \widehat h |\widetilde \xi|^{-5/2}$, it is immediate to write
\begin{equation}\label{eqDhcov}
 \widehat D_h = \widehat{D}_{\widehat{h}} := \frac{\widehat h}{|\eta|_{\widehat g}^{5/2} |\widetilde \xi|^{5/2}_{\widehat g}}(\widetilde \xi \otimes \eta + \eta \otimes \widetilde \xi),
\end{equation}
which is obviously conformal and diffemorphism covariant if and only if $\widehat h$ is conformal invariant, e.g. considering $\widehat h \equiv \widehat h(\X)$. We remark that the form \eqref{eqDhcov} already appeared (with different powers due to the different dimension) in the classification \cite{marspeon20}
of TT tensors in dimension two satisfying the KID equation.

For the remaining tensors $\widehat D_{\widetilde{\xi}} := D_{\widetilde{\xi}} / \Omega^5$ and $\widehat D_{\eta} := D_{\eta} / \Omega^5$, we define a conformal class of vector fields $\chi$, which in the original gauge coincides with $\chi := \partial_{z}$.
This vector is divergence-free $\nabla_A \chi^A = 0$, and this equation is conformally invariant provided the conformal weight of $\chi$ is $-3$ (i.e.
for  $\widehat{g} = \Omega^2 g$, the corresponding vector
is $\widehat{\chi}  = \Omega^{-3} \chi$). We therefore impose this conformal
behaviour of $\chi$.\footnote{This choice may appear somewhat ad hoc at this point. However, the condition of vanishing divergence appears naturaly
  when studying (for more general metrics) under which
  conditions  a tensor $\xi \otimes W  + W \otimes \xi$ is a TT tensor satisfying the KID equation for $\xi$. We leave this general analysis for a future work.}
The direction of $\chi$ is fixed by orthogonality to $\widetilde \xi$ and $\eta$. The combination of norms that has this conformal weight and
recovers the appropriate expression in the gauge of Lemma \ref{lemmaTTadapt}
is $ |\chi|_{\widehat g} := |\widetilde \xi|_{\widehat g}^{-1} |\eta|_{\widehat g}^{-2}$
(note that the orthogonality and norm conditions fix $\chi$ uniquely  up to an irrelevant sign in any gauge). Thus, we may write
\begin{equation}\label{eqD1D2cov}
 D_{\widetilde{\xi}} = \frac{1}{|\widetilde \xi|_{\widehat g}^2}(\chi \otimes \widetilde \xi + \widetilde \xi \otimes \chi), \quad\quad D_{\eta} = \frac{1}{|\eta|_{\widehat g}^2}(\chi \otimes \eta + \eta \otimes \chi),
\end{equation}
which are conformally covariant expressions. Therefore, we get to the final result: 

\begin{theorem}\label{theoTTcov}
 Let $\xi$ be a CKVF of the class of three dimensional conformally flat metrics and let $\xi = \widetilde \xi + \mu \eta$ a canonical form. For each conformal gauge, let us define a vector field $\chi$ with norm $ |\chi|_{\widehat g} := |\widetilde \xi|_{\widehat g}^{-1} |\eta|_{\widehat g}^{-2}$,  orthogonal to $\widetilde \xi$ and $\eta$. Then, any TT-tensor satisfying the KID equations \eqref{eqKID} for $\widetilde \xi$ and $\eta$ is a combination (with constants) of the following tensors:
 \begin{align}
 &\widehat D_{\widehat{f}} = \X^4 \frac{\dif}{\dif \X} \lr{\frac{\widehat f(\X)}{\X^{3/2}}} \mathcal{D}(\widetilde \xi) -\frac{1}{\X^2} \frac{\dif}{\dif \X} \lr{\X^{3/2} \widehat f(\X)} \mathcal{D}(\eta), & & \widehat D_{\widehat{h}} = \frac{\widehat h}{|\eta|_{\widehat g}^{5/2} |\widetilde \xi|^{5/2}_{\widehat g}}(\widetilde \xi \otimes \eta + \eta \otimes \widetilde \xi),\\ 
 & D_{\widetilde{\xi}} = \frac{1}{|\widetilde \xi|_{\widehat g}^2}(\chi \otimes \widetilde \xi + \widetilde \xi \otimes \chi), & & D_{\eta} = \frac{1}{|\eta|_{\widehat g}^2}(\chi \otimes \eta + \eta \otimes \chi),
\end{align}
for arbitrary functions $\widehat f$ and $\widehat h$ of $\X = |\widetilde \xi|_{\widehat g}/|\eta|_{\widehat g}$.
\end{theorem}


\begin{remark}
  The vector field $\chi$ defined in this Theorem is divergence-free. This property would have been difficult to guess (and even to prove) in the original Cartesian coordinate system.
\end{remark}

\begin{remark}  
  A corollary of this theorem is that the general solution of the
  $\Lambda$-vacuum Einstein field equation in four dimensions with a smooth conformally  flat  null infinity and admitting an axial symmetric and a second
  commuting Killing vecor can be parametrized by two functions of one variable and  two constants.  Recall that in the $\Lambda=0$ case, the general asympotically
  flat stationary and axially symmetric  solution of the Einstein field equations can be parametrized (in a neighbourhood of spacelike infinity, by two
  numerable sets of mass and angular multipole moments (satisfying appropriate
  convergence properties), see \cite{ace09}, \cite{BSch}, \cite{back} for details. There is an intriguing paralelism between
  the two situations, at least at the level  of crude counting of degrees of freedom.  This suggests that maybe in the
  $\Lambda>0$ case it is possible to define a set of multipole-type moments that characterizes de data at null infinity (and hence the spacetime), at least in the case  of a conformally flat null infinity. This is an interesting problem, but well beyond the scope of the present paper.
\end{remark}

\begin{remark}
  It is natural to ask whether Theorem \ref{theoTTcov}
  is general for TT-tensors admitting two commuting KIDs, $\widetilde \xi, \eta$, without the condition of $\eta$ being conformally axial. In Appendix $C$ of  \cite{Kdslike} one can explicitly find, for an arbitrary CKVF $\xi$, the set $\mathcal{C}(\xi)$ of elements that commute with $\xi$. Then, from a case by case analysis, one concludes that except in one special situation, for any linearly independent pair $\xi, \xi'$, with $\xi' \in \mathcal{C}(\xi)$ it is the case that there is a CAKVF $\eta \in \mathcal{C}(\xi)$ such that
  $\mbox{span} \{ \xi,\eta \} = \mbox{span} \{ \xi, \xi'\}$.
  Thus, all these cases
are covered by Theorem \ref{theoTTcov}. The exceptional case is when 
$\xi, \xi'$ are conformal to translations. It is immediate to solve the TT and KID equations for such a case directly in Cartesian coordinates.
%
%
\end{remark}
 
The solution given in Theorem \ref{theoTTcov} provides a large class of initial data, which we know must contain the so-called Kerr-de Sitter-like {\it class} with conformally flat $\mathscr{I}$
(see \cite{Kdslike} for precise definition and properties of this class), which in turn contains the Kerr-de Sitter {\it family} of spacetimes. It is interesting to identify this class within the general  solution given in Theorem \ref{theoTTcov}. The characterizing property of the Kerr-de Sitter-like class in the conformally flat case
is  $D = D(\xi)$ for some CKVF $\xi$, {where moreover, only the conformal class of $\xi$ matters to determine the family associated to the data.} Decomposing canonically
$\xi = \widetilde{\xi} + \mu \eta$, a straightforward computation yields
\begin{align}
  \mathcal{D}(\xi) & =
  \frac{\X^5}{(\X^2 + \mu^2)^{5/2}} \mathcal{D}(\widetilde \xi) + \frac{\mu^2 }{(\X^2 + \mu^2 )^{5/2}} \mathcal{D}(\eta) + \frac{\mu \X^{5/2}}{(\X^2 + \mu^2)^{5/2}} \widehat D_{\widehat{h}=1},
\end{align}
which comparing with Theorem \ref{theoTTcov} yields the following corollary:
\begin{corollary}\label{corolKds}
 The Kerr-de Sitter-like class with conformally flat $\mathscr{I}$ is determined by the TT-tensor $D_{KdS}= \widehat D_f + \widehat D_{\widehat{h}}$ with 
 \begin{equation}
  \widehat f = -\frac{1}{3} \frac{\X^{3/2}}{(\X^2 + \mu^2)^{3/2}},\quad\quad  \widehat h =  \mu \frac{\X^{5/2}}{(\X^2 + \mu^2)^{5/2}}.
 \end{equation}
\end{corollary}
It is also of interest to identify the the Kerr-de Sitter family. To that aim
we combine the results in \cite{Kdslike} to those in the present paper to show that
this family corresponds to  $\cc < 0$. The classification of conformal classes
of $\xi$ in \cite{Kdslike}  is done in terms of the invariants
$\widehat c = -c_1$ and $\widehat k = - c_2$ together with
the rank parameter $r$, where $c_1$ and $c_2$ are the coefficients of the characteristic polynomial of the skew-symmetric endomorphism $F$ associated to $\xi$. In terms of these objects, it is shown in \cite{Kdslike} that the Kerr-de Sitter family corresponds to either
$\mathcal{S}_1 = \{\widehat k>0,~\widehat c \in \mathbb{R}$ and $r = 2 \}$, or $\mathcal{S}_2 = \{\widehat k=0,~\widehat c >0 $ and $r = 1\}$, the latter defining the Schwarzschild-de Sitter family. It is immediate to verify that, since (cf. Corollary \ref{corolcharac}) $\widehat k =  -\cc \mu^2 < 0$ and $\widehat c= -\cc - \mu^2$, then $\mathcal{S}_1 =  \{ \cc <0, \mu \neq 0 \}$ and $\mathcal{S}_2 = \{ \cc <0, \mu = 0 \}$ (the condition $\mu \neq 0$ implies $r=2$ and $\mu = 0$ implies $r = 1$). Thus, in terms of the classification developed in this paper, the Kerr-de Sitter
family corresponds  to $\cc < 0$. It is interesting that in the present scheme
we no longer need to specify the rank parameter to identify
the Kerr-de Sitter family (unlike in \cite{Kdslike}) and that the whole family is represented by an open domain. We emphazise that 
the dependence in $\cc$ in the solutions given in Theorem \ref{theoTTcov} and Corollary \ref{corolKds} is implicit through the norm of $\widetilde \xi$. 

%
%
%
%

%
%

 \section*{Acknowledgements}

The authors acknowledge financial support under the projects
PGC2018-096038-B-I00
(Spanish Ministerio de Ciencia, Innovaci\'on y Universidades and FEDER)
and SA083P17 (JCyL). C. Pe\'on-Nieto also acknowledges the Ph.D. grant BES-2016-078094 (Spanish Ministerio de Ciencia, Innovaci\'on y Universidades).

\end{document}